\pdfoutput=1
\documentclass[a4paper, 11pt, twoside]{report}

\usepackage[utf8]{inputenc}

\usepackage{url} 
\usepackage{amssymb} 
\usepackage{bbm} 
\usepackage{mathtools}
\DeclarePairedDelimiter{\norm}{\lVert}{\rVert}

\usepackage{algorithm}
\usepackage{amsmath}
\usepackage{algpseudocode}

\algdef{SE}[FORMOD]{ForMod}{EndForMod}[3]{\algorithmicfor\ #1\ \textbf{from}\ #2\ \textbf{to}\ #3\ \textbf{do}}{\algorithmicend\ \algorithmicfor}%
\algdef{SE}[FORMOD]{ForModS}{EndForModS}[4]{\algorithmicfor\ #1\ \textbf{from}\ #2\ \textbf{to}\ #3\ \textbf{step}\ #4\ \textbf{do}}{\algorithmicend\ \algorithmicfor}%
\algdef{SE}[FORMODP]{ForModPS}{EndForModPS}[3]{\algorithmicfor\ #1\ \textbf{from}\ #2\ \textbf{to}\ #3\ \textbf{in parallel do}}{\algorithmicend\ \algorithmicfor}%
\algdef{SE}[FORMODP]{ForModP}{EndForModP}[4]{\algorithmicfor\ #1\ \textbf{from}\ #2\ \textbf{to}\ #3\ \textbf{step} #4\ \textbf{in parallel do}}{\algorithmicend\ \algorithmicfor}%
\algnewcommand{\LineComment}[1]{\State \(\triangleright\) #1}
\algnewcommand{\And}{\textbf{and}\xspace}

\usepackage{graphicx}
\usepackage{epstopdf}
\usepackage{pgfplots}
\usepackage{pgfplotstable}
\usepackage{tikz}
\usepackage{float}
\usetikzlibrary{matrix, calc,arrows,shapes.arrows,chains,positioning,scopes, backgrounds,decorations.pathreplacing, patterns}
\usetikzlibrary{shapes}
\usetikzlibrary{svg.path}
\usepgfplotslibrary{groupplots, units}
\usetikzlibrary{external}
\tikzset{
	external/export=true,           
	external/mode=list and make,    
	external/prefix=externalized/,  
	external/verbose IO=false,      
	external/optimize command away=\dm1,
	external/optimize command away=\cite1,
	external/optimize command away=\citeauthor1,
	external/optimize command away=\fullcite1,
}
\pgfdeclarelayer{backgroundlayer}

\usepackage{subfig}
\usepackage[export]{adjustbox}

\usepackage[intoc]{nomencl}
\makenomenclature
\usepackage{etoolbox}
\renewcommand\nomgroup[1]{%
	\item[\bfseries
	\ifstrequal{#1}{P}{Parallel prefix sum}{%
		\ifstrequal{#1}{I}{Image Registration}{%
			\ifstrequal{#1}{O}{Other Symbols}{}}}%
	]
}

\usepackage{setspace}       

\usepackage[
  urldate=comp,
  maxnames=99,
  abbreviate=false,
  defernumbers=true,
  sorting=none,
]{biblatex}
\biburlsetup
\renewcommand\bibfont\small

\usepackage{amsthm}
\theoremstyle{definition}
\newtheorem{theorem}{Theorem}[section]
\newtheorem{definition}{Definition}[section]
\newtheorem{remark}{Remark}[section]
\newcommand*\mean[1]{\bar{#1}}

\usepackage{import}

\usepackage{standalone}

\usepackage{caption}
\makeatletter
\renewcommand\fps@figure t  
\renewcommand\fps@table t
\makeatother

\usepackage{color}
\usepackage{listings}
\newcommand{\code}[1]{\texttt{#1}}
\newfloat{codeblock}{H}{myc}
\definecolor{darkblue}{rgb}{0,0,.6}
\definecolor{darkred}{rgb}{.6,0,0}
\definecolor{darkgreen}{rgb}{0,.6,0}
\definecolor{red}{rgb}{.98,0,0}
\definecolor{gray}{rgb}{.6,.6,.6}
\usepackage[]{empheq}
\newcommand{\for}{\text{for }}

\usepackage{siunitx}
\usepackage[]{geometry}

\newcommand{\thesistype}{Master's Thesis}
\newcommand{\thesistitle}{Parallel Prefix Algorithms \protect\\ for the Registration of Arbitrarily Long \protect\\Electron Micrograph Series}
\newcommand{\thesisauthor}{Marcin Copik}
\newcommand{\thesisFirstSupervisor}{Prof. Paolo Bientinesi}
\newcommand{\thesisFirstSupervisorUniversity}{\protect{RWTH Aachen University}}
\newcommand{\thesisFirstSupervisorDepartment}{AICES}

\newcommand{\thesisSecondSupervisor}{Prof. Benjamin Berkels}
\newcommand{\thesisSecondSupervisorUniversity}{\protect{RWTH Aachen University}}
\newcommand{\thesisSecondSupervisorDepartment}{AICES}

\newcommand{\thesisuniversity}{\protect{RWTH Aachen University}}
\newcommand{\thesisinstitute}{Aachen Institute for Advanced Study in Computational Engineering Science}
\newcommand{\thesisgroup}{High-Performance and Automatic Computing Group}

\begin{document}


\graphicspath{{.}}

\newgeometry{centering}
\begin{titlepage}
	\centering
	
	{\fontsize{15}{20}\selectfont \textsf{\thesisuniversity}} \\[4mm]
	{\fontsize{13}{18}\selectfont \textsf{\thesisinstitute}} \\
	{\fontsize{13}{18}\selectfont \textsf{\thesisgroup}} \\
	
	\vspace{\fill} 
	{\fontsize{15}{20}\selectfont \thesistype} \\[5mm]
	{\fontsize{20}{25}\selectfont\bfseries \thesistitle \\[10mm]}
	{\fontsize{15}{20}\selectfont \thesisauthor} \\
	
	\vspace{\fill}
	\begin{minipage}[t]{.27\textwidth}
		\raggedleft
		\textbf{First Supervisor:}
	\end{minipage}
	\hspace*{15pt}
	\begin{minipage}[t]{.65\textwidth}
		{\Large \thesisFirstSupervisor} \\
		{\small \thesisFirstSupervisorDepartment} \\[-1mm]
		{\small \thesisFirstSupervisorUniversity}
	\end{minipage} \\[5mm]
	\begin{minipage}[t]{.27\textwidth}
		\raggedleft
		\textbf{Second Supervisor:}
	\end{minipage}
	\hspace*{15pt}
	\begin{minipage}[t]{.65\textwidth}
		{\Large \thesisSecondSupervisor} \\
		{\small \thesisSecondSupervisorDepartment} \\[-1mm]
		{\small \thesisSecondSupervisorUniversity}
	\end{minipage} \\[10mm]
	
\end{titlepage}
\restoregeometry
\cleardoublepage

\begin{abstract}
Recent advances in the technology of transmission electron microscopy have allowed for a more precise visualization of materials and physical processes, such as metal oxidation. Nevertheless, the quality of information is limited by the damage caused by an electron beam, movement of the specimen or other environmental factors. A novel registration method has been proposed to remove those limitations by acquiring a series of low dose microscopy frames and performing a computational registration on them to understand and visualize the sample. This process can be represented as a prefix sum with a complex and computationally intensive binary operator and a parallelization is necessary to enable processing long series of microscopy images. With our parallelization scheme, the time of registration of results from ten seconds of microscopy acquisition has been decreased from almost thirteen hours to less than seven minutes on 512 Intel IvyBridge cores.
\end{abstract}
\cleardoublepage

\tableofcontents

\nomenclature[P]{$S(N, P)$}{Span, length of critical path of an algorithm for input data of length $N$ and $P$ workers}
\nomenclature[P]{$W(N, P)$}{An amount of work performed by an algorithm for input data of length $N$ and $P$ workers}
\nomenclature[P]{$SP(N, P)$}{The theoretical speedup of an algorithm, defined as a ratio of serial and parallel span}
\nomenclature[P]{$P$}{Total number of allocated process cores}
\nomenclature[P]{$T_{S}$}{Execution time of a serial execution}
\nomenclature[P]{$T_{P}$}{Execution time of a parallel execution}
\nomenclature[P]{$\mathcal{SP}$}{The measured speedup of an algorithm, defined as a ratio of serial and parallel exeuction time}
\nomenclature[I]{$f_{i}$}{Image with index $i$}
\nomenclature[I]{$\phi_{i, j}$}{Deformation matching image $f_{j}$ to $f_{i}$}
\nomenclature[I]{$\mathrm{NCC}[f, g]$}{The normalized cross--correlation of images f and g}
\nomenclature[I]{$\mathbf{A}$}{The registration function for two neighboring frames}
\nomenclature[I]{$\mathbf{B}$}{The registration function for two non-neighboring frames}

\printnomenclature[2.5 cm] 
\cleardoublepage

\chapter{Introduction}
\label{chap:intro}

Modern electron microscopes allow an observation of specimens at a nanometer resolution. Recent advances in the technology of scanning transmission electron microscopy (STEM) allowed for a more precise visualization of materials and physical processes such as metal oxidation. Nevertheless, the quality of information obtained during the acquisition is limited by the damage caused by an electron beam, movement of the specimen or other environmental factors. Restricting the electron dose to avoid the damage results in obtaining data with an undesirably low \textit{signal--to--noise} ratio. Image processing algorithms have been successively applied to extract reliable information from a noisy electron microscopy data. \\
Berkels et.al.\cite{Berkels201446} have proposed a new approach to increase the amount of information gathered by observation with STEM. Instead of using a single high--dose frame, a series of low--dose noisy frames $f_{0}, f_{1}, \dots, f_{n}$ is acquired. The quality of frames is affected not only by the noise but also by the movement of observed object during the acquisition. Therefore, frames are aligned to the first image $f_{0}$ to represent only the physical change of the observed object and not the movement of the specimen. \\ The information encoded in images is extracted in a two--step series \textit{registration} method. First, for each pair of neighboring images $(f_{i}, f_{i + 1})$ an image deformation $\phi_{i, i + 1}$ is approximated such that $f_{i} \approx f_{i+1} \circ \phi_{i, i+1}$. The process of finding a good estimation for $\phi_{i, i + 1}$ is based on a \textit{multilevel gradient flow} minimizing a \textit{normalized cross--correlation} between $f_{i}$ and $f_{i+1} \circ \phi_{i, i + 1}$. A composition of two deformations $\phi_{i, j}$ and $\phi_{j, k}$ with a common center $j$ can be used as a starting to point for registration of frame $f_{k}$ to $f_{i}$. A recursive application of this procedure allows aligning each frame $f_{i}$ to the first one $f_{0}$. An example of the process is presented on Figure~\ref{fig:chapter_introduction_series}. A series of deformed images $f_{i} \circ \phi_{0, i}$ is averaged to produce a single frame representing the observed object. \\
The theoretical background of the new method and a formal statement of the problem is introduced in the the Chapter~\ref{chap:image_registration}.
\begin{figure}[htb]
	\centering
	\includegraphics[width=\textwidth]{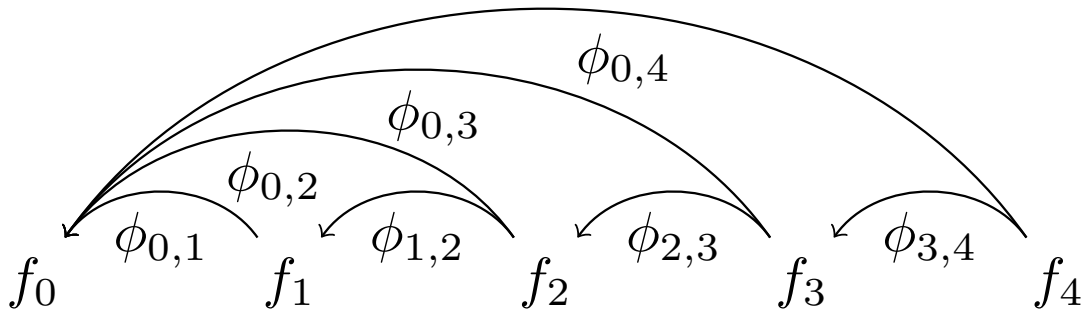}
	\caption{An example of series registration for images $f_{0}, f_{1}, \dots, f_{4}$.}
	\label{fig:chapter_introduction_series}
\end{figure}
The process of image registration is computationally intensive. As a matter of fact, the convergence of a gradient flow on three levels $8$, $9$ and $10$, usually takes a few seconds. In our experiments, a serial registration of $4096$ frames requires almost 13 hours of computation. Since each second of data acquisition generates 400 frames, the registration algorithm becomes impractical for an acquisition running longer than a few dozen seconds. This problem becomes even more apparent in applications such as the \textit{series averaging procedure}, where the registration algorithm is repeated many times to improve the quality of averaged frame. The computation time could be reduced by manipulating algorithm parameters to perform fewer iterations, but this raises the likelihood of gradient flow finding a local minimum which decreases the quality of results. In this dissertation, we discuss another approach to speed up the registration procedure which is to employ parallel computing techniques.

In the Chapter~\ref{chap:prefix_sum}, we prove that the image registration process can be represented as a \textit{prefix sum}. The prefix sum, also known as \textit{scan} or \textit{cumulative sum}, accepts a sequence of input data $x_{0}, x_{1}, \dots, x_{n}$ with a binary operator $\odot$ and for each element $x_{i}$, computes a sum of all preceding elements and the selected item, such as
\begin{align*}
x_{i} &= x_{0} \odot x_{1} \odot \dots \odot x_{i}
\end{align*}
Parallelization strategies for a prefix sums have been researched for decades to construct fast and efficient prefix adders, a class of digital circuits performing binary addition. In the parallel programming, the scan primitive has been proposed as a basic block for building parallel applications. Therefore, we intend to use the parallel prefix sum as a basis for parallelization strategy of the image registration problem.

However, properties of the image registration process are entirely different from prefix sum problems discussed in the literature. Previous work is focused on parallelization strategies for memory--bound operators where the cost of accessing and moving data is significantly higher than an application of the operator. Furthermore, the iterative nature of registration does not allow to predict a total cost of computation, and we have observed huge variances in execution time between different pairs of frames. We have not found any examples of a parallel scan with an operator of unpredictable runtime. Thus, we construct new guidelines for an efficient parallelization and reevaluate existing strategies. \\
For an arbitrarily long series of data acquisition, we need a distributed implementation of a parallel prefix sum. We derive a general strategy which attempts to minimize the synchronization between workers. The distributed approach for a parallel prefix sum is discussed in the Chapter~\ref{chap:distr_prefix_sum}.

We implemented the general strategy for a parallel distributed prefix sum as an extension of QuocMesh\cite{quocmesh}, a library for Finite Elements computations on Cartesian grids developed at the Institute for Numerical Simulation at the University of Bonn. The strategy has been applied to the image registration process, and a comparative evaluation of different algorithms is presented in Chapter~\ref{chap:results}. Our results prove that we cannot solve a fixed size problem efficiently on an arbitrary number of processors. However, we can use more hardware to register longer series without a significant increase in the execution time. Further improvements are provided by parallelization of the registration algorithm.

Finally, we present conclusions and suggestions for future work in Chapter~\ref{chap:summary}.

\chapter{Image registration}
\label{chap:image_registration}

In this chapter, we introduce the problem of image registration for series of frames. We use names \textit{image} and \textit{frame} interchangeably for a single item of data acquired by an electron microscope. Furthermore, we use terms \textit{transformation} and \textit{deformation} interchangeably for a function defining the deformation of an image. Definitions provided here are based on those given by Modersitzki in \cite{modersitzki2004numerical}\cite{modersitzki2009fair}.

\section{Problem statement}

An image is a function which assigns a gray value to each position in the region of interest, as defined below
\begin{definition}{\textbf{Image}}\label{def:image} A $d$--dimensional image is a function $f$
\begin{align}
f: \Omega \longrightarrow \mathbb{R}
\end{align}
where $\Omega \in \mathbb{R}^d$ is a region of interest and $d \in \mathbb{N}$. $\text{Img}(d)$ is a set of all d--dimensional images.
\end{definition}
In this chapter, we introduce methods for \textit{registration} of two--dimensional images where $d = 2$. We denote two particular kinds of image: $\mathcal{R}$ known as the reference and $\mathcal{T}$ known as the template image. In the image registration problem, we want to find the transformation $\phi: \mathbb{R}^{d} \longrightarrow \mathbb{R}^{d}$ of $\mathcal{T}$ such that the deformed image is aligned to the reference image and
\begin{align}
\mathcal{T} \circ \phi &\approx \mathcal{R}
\end{align}
For the sake of simplicity, we restrict ourselves to rigid transformations in the description, but techniques described below allow to approximate both rigid and non--rigid deformations. A rigid deformation is defined as follows
\begin{definition}{\textbf{Rigid transformation}}\label{def:rig_deform} A transformation is called rigid when only rotation and translation are allowed. A rigid transformation is represented by the equation
	\begin{align}
	\phi(x) \: &= \: R(\alpha) \cdot x + G 
	\end{align}
where $R(\alpha) \in \mathbb{R}^{d \times d}$ is an orthogonal matrix and $G \in \mathbb{R}^{d}$.
\end{definition}
The name refers to a movement of a rigid body which can not be deformed through a shear, scaling or any other non--affine transformation. In the two--dimensional case, the deformation is given by a translation vector $b$ of length two and a single rotation angle $\alpha$. Then the transformation shall take the form
\begin{align}
\phi(x) \: &= \: \begin{pmatrix}
\cos\alpha & -\sin\alpha \\
\sin\alpha & \cos\alpha
\end{pmatrix} \cdot
\begin{pmatrix}
x_{0} \\
x_{1}
\end{pmatrix}
+ \begin{pmatrix}
t_{0} \\
t_{1}
\end{pmatrix}
\end{align}
In practice, however, it is usually impossible to obtain a deformation providing an ideal match for images. Therefore, a proper metric has to be defined to estimate the similarity between images and serve the role of an objective functional in the minimization process. We intend to find a transformation such that given a metric, the distance between a reference $\mathcal{R}$ and a transformed template $\mathcal{T} \circ \phi$ is minimal. Finally, we define the image registration problem
\begin{definition}{\textbf{Image registration problem}}\label{def:image_reg} Given a metric $\mathcal{M}: \longrightarrow \mathbb{R}$ and two images $\mathcal{R}, \mathcal{T}$, find a transformation $\phi$ such that
\begin{align}
\mathcal{M}(\mathcal{R}, \phi \circ \mathcal{T})
\end{align}
is minimized.
\end{definition}
An example of a trivial transformation is presented in Figure~\ref{fig:chapter_img_rect_deformation}. For both images, $\Omega = [0, 1]^{2}$. In the image $f_{1}$, the rectangle has been translated but not rotated. We intend to find a deformation $\phi_{0, 1}$ such that
\begin{align*}
f_{1} \circ \phi_{0, 1} &= f_{0} \\
\forall x\in\Omega \: f_{1}( \phi_{0, 1}(x) ) &= f_{0}(x)
\end{align*}
To find a transformation, we observe that no shift is performed on vertical axis and we solve for the lower--left corner of the rectangle $(0.25, 0.25)$ and $(0.5, 0.25)$
\begin{align*}
f_{0}(x) &= f_{1}( \phi_{0, 1}(x))\\
f_{0}(x) &= f_{1}( x + G )\\
f_{0}((0.25, 0.25)^{T}) &= f_{1}((0.25, 0.25)^{T} + (g_{0}, 0)^{T}) \\
&\iff \\
(0.5, 0.25)  &= (0.25, 0.25)^{T} + (g_{0}, 0)^{T}\\
G &= (0.25, 0)^{T}
\end{align*}
Thus, the transformation does not deform the template image $f_{1}$ to match $f_{0}$. This operation would require a translation vector $-G$. The deformation is applied to coordinates before computing an image value at given position. $\phi$ represents a geometrical change from frame $f_{0}$ to $f_{1}$, not the other way around. With both images representing exactly the same object, the transformation encodes a correspondence between coordinate systems of two images.
\begin{figure}
	\centering
	\includegraphics[width=0.75\textwidth]{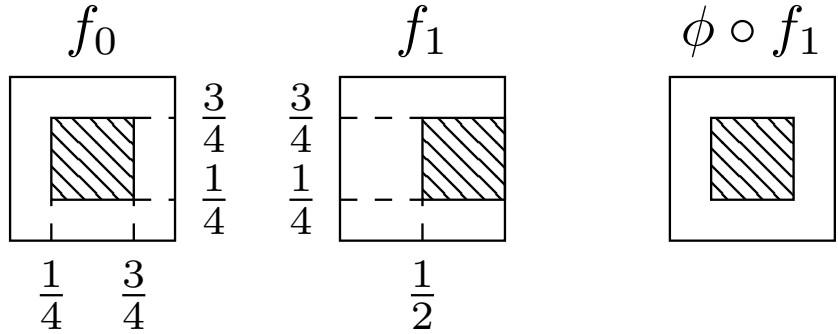}
	\caption{An example of two images containing a rectangle of size $(0.5, 0.5)$, located in the center of image $f_{0}$ and on right side of image $f_{1}$. Deformation $\phi$ produces an ideal match of $f_{1}$ to $f_{0}$.}
	\label{fig:chapter_img_rect_deformation}
\end{figure}

\section{Electron microscopy data}

For the image registration, we consider data from an experiment where ultrahigh vacuum high--resolution transmission electron microscopy (UVH HRTEM) has been applied to capture the process of aluminum oxidation\cite{thesisManchester}. The images have been acquired at the rate of 400 frames per second, and each experiment has lasted for up to 4 minutes, producing up to 96,000 frames. The quality of images is lowered by \textit{sample drift}, a movement of the aluminum sample between taking consecutive frames, and the presence of low--contrast frames.

An example of an electron microscopy image is presented on Figure~\ref{fig:chapter_img_electron_data}. The regular structure on the left represents an atomic grid of aluminum. The approximated deformation for this pair of images is
\begin{align}
\phi(x) \: &= \: \begin{pmatrix}
0.99 & -4.18 \cdot 10^{-4} \\
4.18 \cdot 10^{-4} & 0.99
\end{pmatrix} \cdot
\begin{pmatrix}
x_{0} \\
x_{1}
\end{pmatrix}
+ \begin{pmatrix}
-4.53 \cdot 10^{-4} \\
-0.004909
\end{pmatrix}
\end{align}
A key feature of these pictures is a very low variability between two images with neighbor indices, presented in Figures~\ref{fig:chapter_img_electron_data_ref} and~\ref{fig:chapter_img_electron_data_temp}. It is clearly seen in the deformation, the estimated angle of rotation is approximately equal $4.18 \cdot 10^{-4}$ radians and the translation on the horizontal axis is insignificant as well. Not surprisingly, it is a hard task to notice a change in the deformed frame presented on Figure~\ref{fig:chapter_img_electron_data_def}. A magnification of the upper-left corner of this frame is presented on Figure~\ref{fig:chapter_img_electron_data_def_crop}. There, a movement of the frame along the vertical axis can be spotted. 
\begin{figure}[htb]
	\centering
	\subfloat[Frame 25]{%
		\includegraphics[width=\dimexpr0.45\textwidth-2\fboxrule,frame={\fboxrule}]{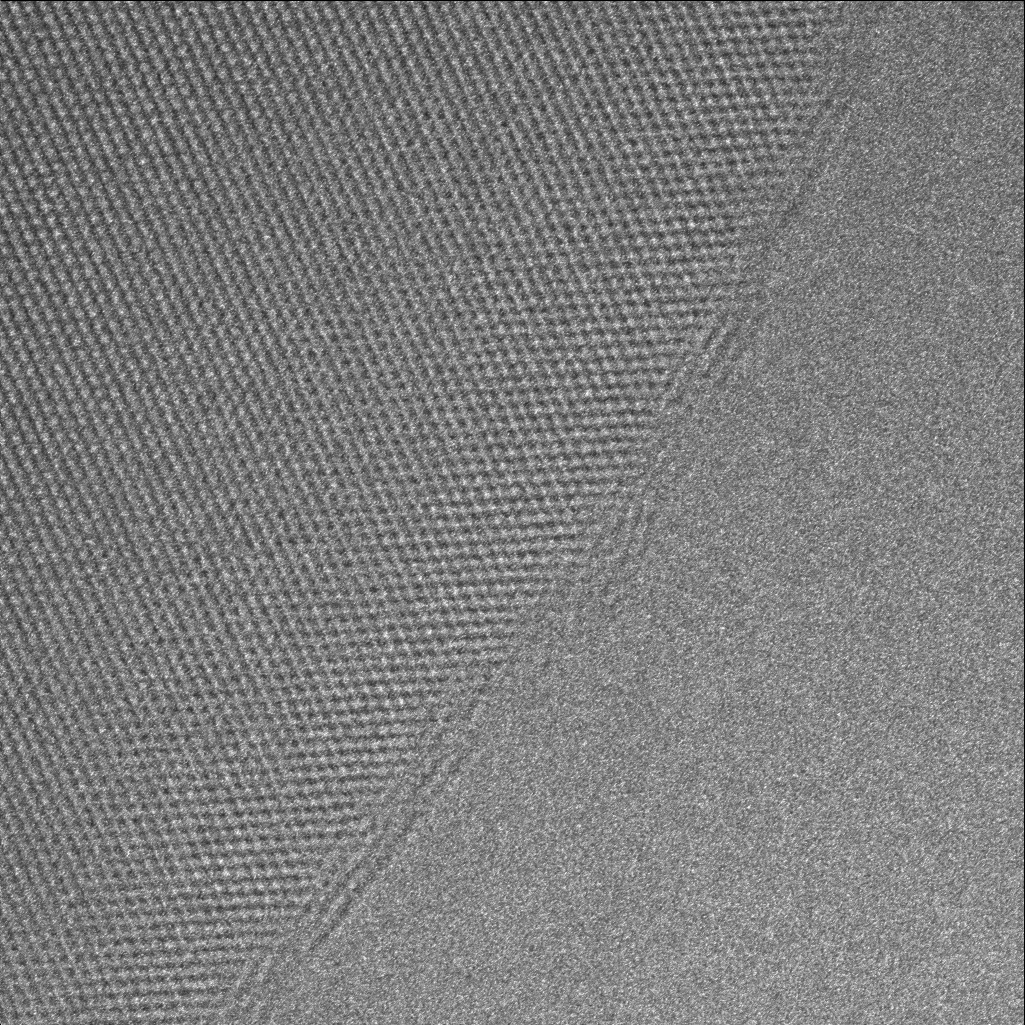}\label{fig:chapter_img_electron_data_ref}}\hfill
	\subfloat[Frame 26]{%
		\includegraphics[width=\dimexpr0.45\textwidth-2\fboxrule,frame={\fboxrule}]{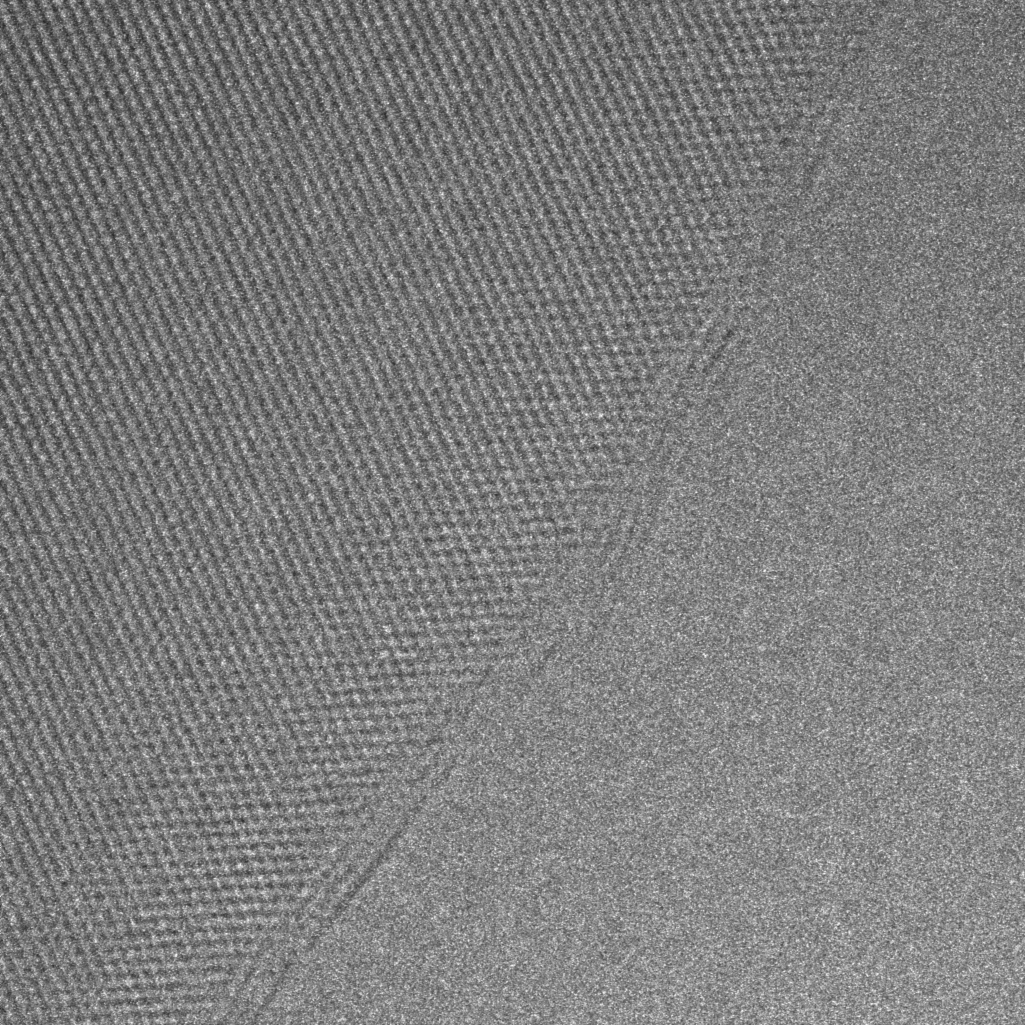}\label{fig:chapter_img_electron_data_temp}}\\
	\subfloat[Deformed frame 26]{%
		\includegraphics[width=\dimexpr0.45\textwidth-2\fboxrule,frame={\fboxrule}]{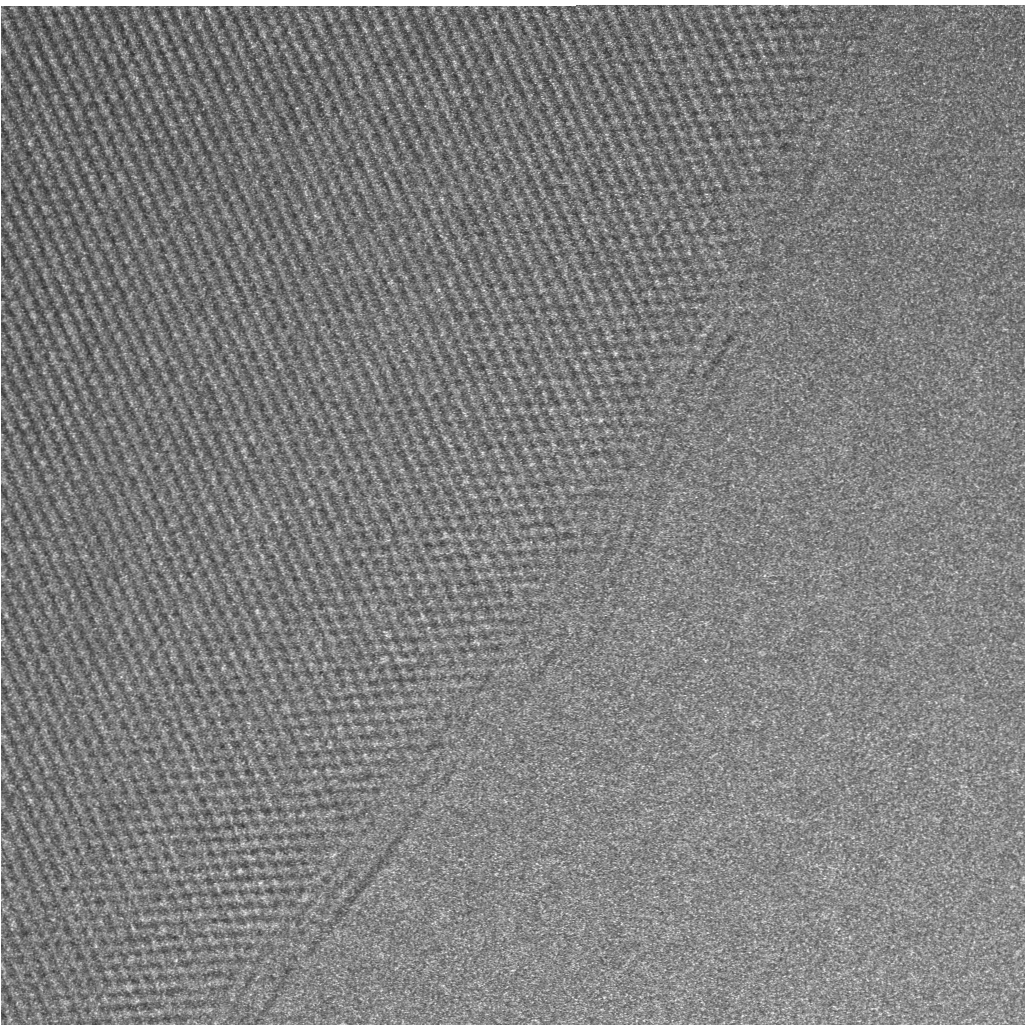}\label{fig:chapter_img_electron_data_def}}\hfill
	\subfloat[An upper--left corner of deformed frame 26]{%
		\includegraphics[width=\dimexpr0.45\textwidth-2\fboxrule,frame={\fboxrule}]{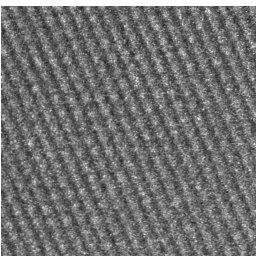}\label{fig:chapter_img_electron_data_def_crop}
	}\hfill
	\caption{The first two figures of the top row present frames 25 and 26, respectively,  from twentieth--first second of first minute of the acquisition. The registration has been performed and the two figures of the bottom row depict the deformed frame 26 and a magnification of an upper--left corner of this frame. The movement of frame is visible on the last image.}  \label{fig:chapter_img_electron_data}
\end{figure}
\pagebreak
\section{Image registration for electron microscopy}
This section provides a brief description of the selected method for image registration. This technique has been proposed for electron microscopy images by Berkels et al. in \cite{Berkels201446}. For more details on the method, please refer to the mentioned paper or to \cite{Clarenz02towardsfast}. \\
We begin the description with a definition of selected metric. The mean value of an image $\mean{f}$ is defined as
\begin{align}
\mean{f} &= \frac{1}{|\Omega|} \int_{\Omega} f \mathrm{d}x 
\end{align}
And the standard deviation is defined as follows
\begin{align}
\sigma_{f} &= \sqrt{\frac{1}{|\Omega|} \int_{\Omega} (f - \mean{f})^{2} \mathrm{d}x} 
\end{align}
The \textit{normalized cross--correlation} of reference $\mathcal{R}$ and template $\mathcal{T}$ is given as
\begin{align}
\mathrm{NCC}[\mathcal{R}, \mathcal{T}] &= \frac{1}{|\Omega|} \int_{\Omega} \frac{\mathcal{R} - \mean{\mathcal{R}}}{\sigma_{\mathcal{R}}} \frac{\mathcal{T} - \mean{\mathcal{T}}}{\sigma_{\mathcal{T}}} \mathrm{d}x
\end{align}
The value of a normalized cross--correlation is bounded from below by $-1$ and from above by $1$. Obviously, for a perfect transformation $\mathrm{NCC}[\mathcal{R}, \phi \circ \mathcal{T}] = 1$. This function is combined with a regularization term consisting of Dirichlet energy of the displacement $\phi(x) - x$ to form the objective functional for the optimization process
\begin{align}
\mathrm{E}[\phi] \: &= \: -\mathrm{NCC}[\mathcal{R}, \phi \circ \mathcal{T}] + \lambda \cdot \frac{1}{2} \int_{\Omega} \norm{D(\phi(x) - x)}^{2} \mathrm{d}x 
\end{align}
The additional term, controlled by a regularization parameter $\lambda > 0$, has been introduced because of an ill--posedness of the problem. This function is characterized by a presence of multiple local minima which make finding a unique solution a very non--trivial problem. A key requirement for well--posedness of a problem is an existence of a unique solution, and the additional regularization is expected to make the problem well--posed by leading to a more convex functional with a single minimum\cite{modersitzki2009fair}. As we are going to see later, it is usually not the case, and the computed deformation may vary not only between different starting points for the minimization but also among various implementations of the same algorithm. \\

The proposed approach for minimization of the functional is a hybrid one, based on a combination of a multilevel scheme with a gradient flow minimization process. The multilevel process and the gradient flow with its spatial discretization are introduced in the following sections.
\subsection{Multilevel}

The idea of a multilevel algorithm comes from multigrid\cite{Trottenberg:2000:MUL:374106}, a major technique developed for solving partial differential equations. There, a scheme consisting of multiple grids is used to reduce high--frequency errors by applying a smoother at a fine grid and iteratively solving the problem on coarser grids. Each step down to a coarser grid requires \textit{restricting} the residual, solving problem there and \textit{prolongating} the solution to a finer grid.

A multilevel scheme applies those concepts in image registration to minimize the likelihood of a gradient solver stopping at local minima. The minimization process operates on grid levels $m_{0}, m_{0} + 1, \dots, m_{1}$ and $m_{0} < m_{1}$. For each level $l$, a grid of size $2^{l} \times 2^{l}$ or $(2^{l} + 1) \times (2^{l} + 1)$ is created. For this implementation, we focus on the latter. \\
A coarser grid $\mathcal{G}_{l}$ of size $(2^{l} + 1) \times (2^{l} + 1)$ is extended with a new node between each pair of nodes. The new finer grid $\mathcal{G}_{l+1}$ has $(2^{l+1} + 1) \times (2^{l+1} +1)$ nodes. An example is presented in Figure~\ref{fig:chapter_img_multigrid}. \\
The \textit{prolongation} operator $\mathcal{I}_{l}^{2l}$ is responsible for copying values from a coarser grid to corresponding nodes in the finer grid and for computing a bilinear interpolation for each new node, as defined below
\begin{equation}
\mathcal{I}_{l}^{2l} f(x, y)=
\begin{cases}
\frac{1}{4}(f(x, y - l) + f(x, y + l) & \for x \bmod 2 = 0 \\ \: \: + f(x - l, y) + f(x + l, y))  &  \ \quad \land \:  y \bmod 2 = 0 \\
\frac{1}{2}(f(x - l, y) + f(x + l, y)) & \for x \bmod 2 = 0\\ 
\frac{1}{2}(f(x, y - l) + f(x, y + l))& \for y \bmod 2 = 0 \\
f(x, y) & \text{otherwise} 
\end{cases}
\end{equation}
Since the direction is from a coarse to fine grid, there is no procedure of going back to the starting level like in the V--cycle in multigrid. However, the \textit{restriction} operator is required in the initialization to represent images and the initial guess of deformation on the coarse grid. This operation is performed by a scaled average of neighbors around the coarse node and a stencil representation of this operator is
\begin{equation}
\frac{1}{16}
\begin{pmatrix}
1 & 2 & 1 \\
2 & 4 & 2 \\
1 & 2 & 1
\end{pmatrix}
\end{equation}
\begin{figure}[htb]
	\centering
	\includegraphics[width=0.6\textwidth]{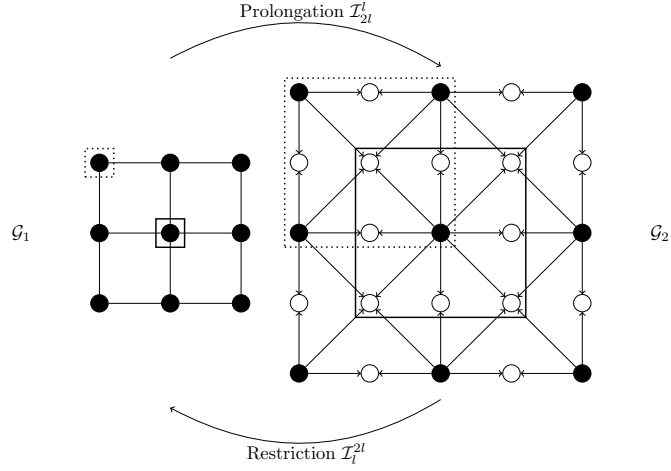}
	\caption{An example of a coarse grid $\mathcal{G}_{1}$ and a finer grid $\mathcal{G}_{3}$. On a fine grid, nodes copied from a coarse grid are depicted as black, and new nodes are white. For each new node, arrows depict nodes contributing to the bilinear interpolation $\mathcal{I}_{2l}^{l}$. Dashed and solid frames symbolize the restriction process. On the finer grid, a frame identifies nodes to which the stencil operator $\mathcal{I}_{l}^{2l}$ is applied. The result of a restriction is stored in a node indicated by a corresponding frame on the coarse grid.}
	\label{fig:chapter_img_multigrid}
\end{figure}
On coarser levels, fewer features of an image are preserved which eradicates local minima created by small structures in the image. Furthermore, the multigrid strategy has been known as an efficient iterative solver because of less computationally intensive computations on coarse levels. 

\subsection{Gradient flow}

A gradient flow solver is a generalization of the gradient descent. The idea behind this optimization technique is the same - moving in the direction of the negative gradient - but the gradient is computed with respect to a scalar product $G$. The update is defined as an ordinary differential equation
\begin{align}
\frac{\partial \phi}{\partial t} &= -\mathrm{grad}_{G} \mathrm{E}[\phi]
\end{align}
The scalar product $G$ is selected in a way to help the minimization process avoid local minima. It has been shown that this ordinary differential equation can be reformulated as
\begin{align}
\label{eq:gradient}
\frac{\partial \phi}{\partial t} &= -A^{-1} \mathrm{E}'[\phi]
\end{align}
$\mathrm{E}'$ refers to the first variation of the functional $\mathrm{E}$ i.e. the generalization of the first derivative of a function of one variable to functionals. The other component $A^{-1}$ is selected to smooth the functional. The smoother is applied only in the non--rigid registration. \\
The equation is discretized in both spatial and time domain. For the latter, a forward Euler scheme is applied
\begin{align}
\frac{\partial \phi}{\partial t} &= \frac{\phi^{k+1} - \phi^{k}}{\tau}
\end{align}
Where $\tau$ is a step size and $\phi^{k}$ is an approximation of deformation from $k$-th iteration. The discrete form of equation~\ref{eq:gradient} is obtained
\begin{align}
\phi^{k+1} &= \phi^{k} -\tau A^{-1} \mathrm{E}'[\phi^{k}]
\end{align}
The step size is selected to ensure the convergence and make the update process more efficient by choosing in each iteration the largest $\tau$ still guaranteeing a decrease of energy. For this problem, an Armijo rule with widening is applied to the energy function of a new deformation $\mathrm{E}(\phi^{k+1})$
\begin{align}
\Phi(\tau) &= \mathrm{E}[ \phi^{k} -\tau A^{-1} \mathrm{E}'[\phi^{k}] ]
\end{align}
A $\tau$ value is selected such that all conditions below are satisfied for $0 < \sigma < 1$
\begin{align}
\left\{
\begin{array}{ll}
	\frac{ \Phi(\tau) - \Phi(0) }{ \Phi'(\tau) \tau} > \sigma \\
	\tau \leq \tau_{max}
	\end{array}
\right.
\end{align}
This ensures that the decay of energy $\Phi$ is at least $\sigma$ times larger than the expected decrease in energy, given by the derivative $\Phi'$. In the implementation employed for this problem, $\sigma = 0.5$.

\subsection{Spatial discretization}

The image domain $\Omega$ is mapped into a uniform rectangular mesh i.e. a mesh consisting of $M$ equal nodes $N_{i}$ in shape of a rectangle. Then, a canonical basis needs to be constructed. For this problem, piecewise bilinear have been selected as basis functions $\varphi_{i}$.
The image function $f$ expressed in the nodal basis is presented in the next equation
\begin{align}
f &= \sum_{j=1}^{M} f_{j} \varphi_{j}
\end{align}
An FE discretization allows introducing the mass matrix $M$ to discretize the integration over domain $\Omega$
\begin{align}
M_{i, j} &= \int_{\Omega} \varphi_{i} \varphi_{j} \mathrm{d}x 
\end{align}
As a result, a new equation for a mean of an image can be derived
\begin{equation}
\begin{split}
\mean{f} &= \frac{1}{\Omega} \int_{\Omega} f \mathrm{d}x \\
&= \sum_{i=1}^{M} \sum_{j=1}^{M} \int_{\Omega} f_{i} \varphi_{i} \varphi_{j} \mathrm{d}x \\
&= M F \mathbbm{1}
\end{split}
\end{equation}
where $F$ denotes a vector representation of $f$ in the nodal basis and $\mathbbm{1}$ is a vector of ones. Furthermore, the standard deviation becomes
\begin{equation}
\begin{split}
\sigma_{f} &= \sqrt{\frac{1}{|\Omega|} \int_{\Omega} (f - \mean{f})^{2} \mathrm{d}x}  \\
&= \sqrt{M(F - \mean{F})^2}
\end{split}
\end{equation}
And the normalized cross--correlation can be formulated as
\begin{equation}
\begin{split}
\mathrm{NCC}[\mathcal{R}, \mathcal{T}] &= \frac{1}{|\Omega|} \int_{\Omega} \frac{\mathcal{R} - \mean{\mathcal{R}}}{\sigma_{\mathcal{R}}} \frac{\mathcal{T} - \mean{\mathcal{T}}}{\sigma_{\mathcal{T}}} \mathrm{d}x \\
&= M \tilde{\mathcal{R}} \tilde{\mathcal{T}}
\end{split}
\end{equation}
where $\tilde{\mathcal{R}}$ and $\tilde{\mathcal{T}}$ are FE representations of normalized reference and template images.

\subsection{Summary}

\begin{algorithm}[H]
	\caption{Multilevel gradient flow for image registration.}
	\begin{algorithmic}[1]
		\ForModS{$i$}{$m_{1} - 1$}{$m_{0}$}{-1} \Comment{Initialize coarse grids}
		\State $\mathcal{R}^{i} \gets restriction(\mathcal{R}^{i+1})$
		\State $\mathcal{T}^{i} \gets restriction(\mathcal{T}^{i+1})$
		\State $\phi^{i} \gets restriction(\phi^{i+1})$ \Comment{Only for a non--rigid deformations}
		\EndForModS
		\ForMod{$i$}{$m_{0}$}{$m_{1}$}
		\State $\phi^{i} \gets solve(\mathcal{R}^{i}, \mathcal{T}^{i}, \phi^{i})$ \Comment{Gradient flow}
		\If{$i < m_{1}$} \Comment{Prolongate deformation to a finer grid}
		\State $\phi^{i+1} \gets prolongate(\phi^{i})$ \Comment{Only for a non--rigid deformations}
		\EndIf
		\EndForMod
	\end{algorithmic}
	\label{alg:multigrid_gradient}
\end{algorithm}

The method is summarized on Algorithm~\ref{alg:multigrid_gradient}. The deformation $\phi^{i}$ is prolongated and restricted only in the non--rigid case. A rigid transformation consists only of three parameters, and it does not depend on grid size. \\
Given relatively low differences between two consecutive frames, setting an identity transformation as the initial guess should not prevent the algorithm from finding a decent solution. Multilevel start and end levels are selected by the user, and they provide a firm boundary on the outer loop. On the other hand, the inner loop requires proper stopping criteria. We employ two criteria:
\begin{itemize}
	\item convergence $\epsilon$ \\
	stop the computation if a difference in energy $\mathrm{E}[\phi^{k+1}] - \mathrm{E}[\phi^{k}]$ is less than the $\epsilon$
	\item maximum number of iterations \\
	perform at most $iter\_max$ iterations
\end{itemize}

In the following chapters, we refer to an implementation of this technique as the function \textbf{A}. It accepts two images with subsequent indices, $f_{i}$ and $f_{i+1}$, and the initial guess for deformation $\phi_{0}$ with a default value of this parameter equal to an identity transformation $I_{\phi}$. The function applies the proposed algorithm to estimate a deformation $\phi_{i, i+1}$
\begin{align}
\forall i \in \mathbb{N} \: \phi_{i, i + 1} &= \mathbf{A}(f_{i}, f_{i+1}, \phi_{0} = I_{\phi})
\end{align}

\section{Registration for series of images}

The previous section introduced a multilevel gradient flow for registration of two consecutive frames. However , we operate on a sequence of $n + 1$ images $f_{0}, f_{1}, \dots, f_{n}$. Therefore, a strategy for registration of series of frames is necessary. \\
Given deformation $\phi_{0, 1}$ which estimates $f_{1} \circ \phi_{0, 1} \approx f_{0}$,  and deformation $\phi_{1, 2}$ providing an approximation $f_{2} \circ \phi_{1, 2} \approx f_{1}$, we can safely assume that a composition of deformations $\phi_{1,2} \circ \phi_{0, 1} $ is a decent guess of deformation registering $f_{0}$ and $f_{2}$, since
\begin{equation}
\begin{split}
f_{2} \circ (\phi_{1,2} \circ \phi_{0, 1}) &= (f_{2} \circ \phi_{1, 2}) \circ \phi_{0, 1} \\ &\approx f_{1} \circ \phi_{0, 1} \\ &\approx f_{0} 
\end{split}
\end{equation}
Thus, we can approximate the deformation for two non--consecutive frames by using a specific initial guess. We reuse the function $\mathbf{A}$ defined in the previous section to define a new function $\mathbf{B}$ such that
\begin{equation}
\begin{split}
\forall i, k \in \mathbb{N}, |k - i| > 1\ \forall j \in \mathbb{N}, i < j < k\ \phi_{i, k}  \: &= \: \mathbf{B}(\phi_{i, j}, \; \phi_{j, k}) \\
&= \mathbf{A}(f_{i}, f_{k}, \phi_{j, k} \circ \phi_{i, j})
\end{split}
\end{equation}
In particular, if we iterate consecutively from the first image
\begin{align}
\label{eq:chapter_img_consecutive_b}
\forall i \in \mathbb{N}, i > 1 \: \phi_{0, i} &= \mathbf{B}(\phi_{0, i - 1}, \phi_{i-1, i})
\end{align}
\begin{figure}[htb]
	\centering
	\includegraphics[width=\textwidth]{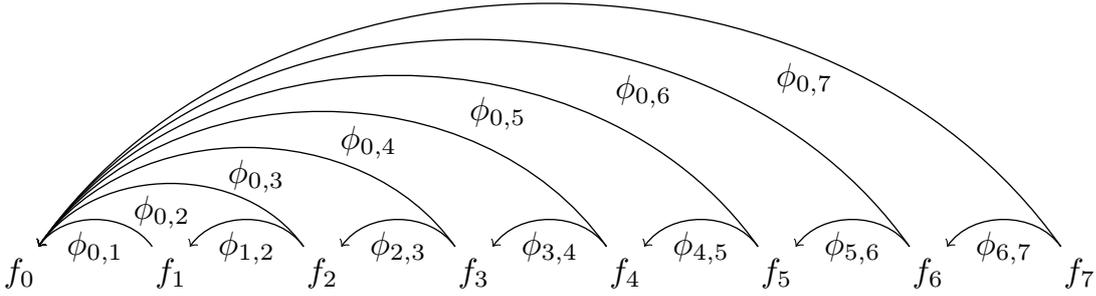}
	\caption{An image registration process for a series of frames. For an image $f_{i}$, the partial result from its predecessor $\phi_{0,i-1}$ is combined with a neighbor deformation $\phi_{i-1, i}$ to register the image to the reference frame $f_{0}$.}
	\label{fig:chapter_img_series_matching}
\end{figure}
The algorithm (\ref{eq:chapter_img_consecutive_b}) is depicted in Figure~\ref{fig:chapter_img_series_matching}. A series of aligned images may be averaged to obtain a single image storing all information acquired in the experiment. Later, the \textit{series averaging procedure} may be used to obtain results of a better quality by performing multiple iterations of the algorithm (\ref{eq:chapter_img_consecutive_b}). For details, please refer to \cite{Berkels201446}.

We formally define the problem as consisting of two steps - a preprocessing stage to generate neighbor transformations and the general registration for non--consecutive frames. In next sections, we refer to the second stage as the \textit{image registration problem}.
\begin{definition}{\textbf{Image series registration}}\label{def:img_reg} Given a sequence of images $f_{0}, f_{1}, \dots, f_{n}$, perform a \textit{preprocessing} step to register each pair of frames
\begin{align*}
\phi_{i, i + 1} &= \mathbf{A}(f_{i}, f_{i+1}, I_{\phi})
\end{align*}
and the \textit{series registration} step to align each image $f_{i}$ to the reference $f_{0}$.
\end{definition}

\section{A note on the associativity}

The standard iterative algorithm outlined in the previous section would require three applications of function $\mathbf{B}$
\begin{equation}
\label{eq:chapter_img_def_str}
\begin{split}
\phi_{0, 2} &= \mathbf{B}(\phi_{0, 1}, \phi_{1, 2}) \\
\phi_{0, 3} &= \mathbf{B}(\phi_{0, 2}, \phi_{2, 3}) \\
\phi_{0, 4} &= \mathbf{B}(\phi_{0, 3}, \phi_{3, 4})
\end{split}
\end{equation}
It is not the only way of computing $\phi_{0,4}$. The process can be split between two processors computing independently $\phi_{0, 2}$ and $\phi_{2, 4}$. Then, results are merged to estimate $\phi_{0,4}$
\begin{equation}
\label{eq:chapter_img_alter_str}
\begin{split}
\phi_{0, 2} &= \mathbf{B}(\phi_{0, 1}, \phi_{1, 2}) \\
\phi_{2, 4} &= \mathbf{B}(\phi_{2, 3}, \phi_{3, 4}) \\
\phi_{0, 4} &= \mathbf{B}(\phi_{0, 2}, \phi_{2, 4})
\end{split}
\end{equation}
The first strategy~(\ref{eq:chapter_img_def_str}) returns such deformation
\begin{align}
\phi_{0,4}(x) \: &= \: \begin{pmatrix}
0.999 & 1.074  \cdot 10^{-5} \\
-1.074 \cdot 10^{-5} & 0.999
\end{pmatrix}
\begin{pmatrix}
x_{0} \\
x_{1}
\end{pmatrix}
+ \begin{pmatrix}
4.743 \cdot 10^{-4} \\
3.431 \cdot 10^{-3}
\end{pmatrix}
\end{align}
The alternative approach~(\ref{eq:chapter_img_alter_str}) produces a result $\phi'_{0, 4}$ which is clearly different from the previous one
\begin{align}
\phi'_{0,4}(x) \: &= \: \begin{pmatrix}
0.999 & -1.424  \cdot 10^{-4} \\
1.424 \cdot 10^{-4} & 0.999
\end{pmatrix}
\begin{pmatrix}
x_{0} \\
x_{1}
\end{pmatrix}
+ \begin{pmatrix}
5.158 \cdot 10^{-4} \\
4.325 \cdot 10^{-3}
\end{pmatrix}
\end{align}
Deformed images have been verified to represent images indistinguishable by a human operator. For further analysis, we compare these deformations by computing energy along a line going through both solutions. We evaluate the energy functional for different deformations
\begin{align}
\mathrm{E}[ t\phi_{0, 4}  + (1-t)\phi'_{0,4}]
\end{align}
Obviously, for $t = 0$ we have $\mathrm{E}[\phi'_{0,4}]$ and for $t = 1$ the computed functional is $\mathrm{E}[\phi_{0,4}]$. The function is plotted against different values of $t$ in Figure~\ref{fig:chapter_img_associativity}. In the second case, the gradient solver has not been able to reach $\phi_{0,4}$ because it got stuck at a local minimum $\phi'_{0,4}$.
An important conclusion here is that the difference between deformations is small enough to not have any impact on the image. Thus, we consider two solutions to be identical if the difference between them appears only at the sub--pixel level. We formally define it by introducing the concept of \textit{approximate associativity}, a weaker form of the associativity
\begin{definition}{\textbf{Approximate associativity}}\label{def:appr_ass} A binary operation $\odot$ defined on a set $S$ is called \textit{approximately associative} if
\begin{align}
\forall a, b, c \in S\ (a \odot b) \odot c \approx a \odot (b \odot c)
\end{align}
where $\approx$ defines two objects as equal if they are indistinguishable in the context of the operation represented by $\odot$ i.e. they represent the same final result.
\end{definition}
Thus, changing the order of operations during evaluation may affect the exact representation of the result. Obviously, each associative operator is approximately associative at the same time. We conclude this chapter with a final remark.
\begin{remark}
Function $\mathbf{B}$ is approximately associative.
\end{remark}
\begin{figure}[htb]
	\centering
	\includegraphics[width=\textwidth]{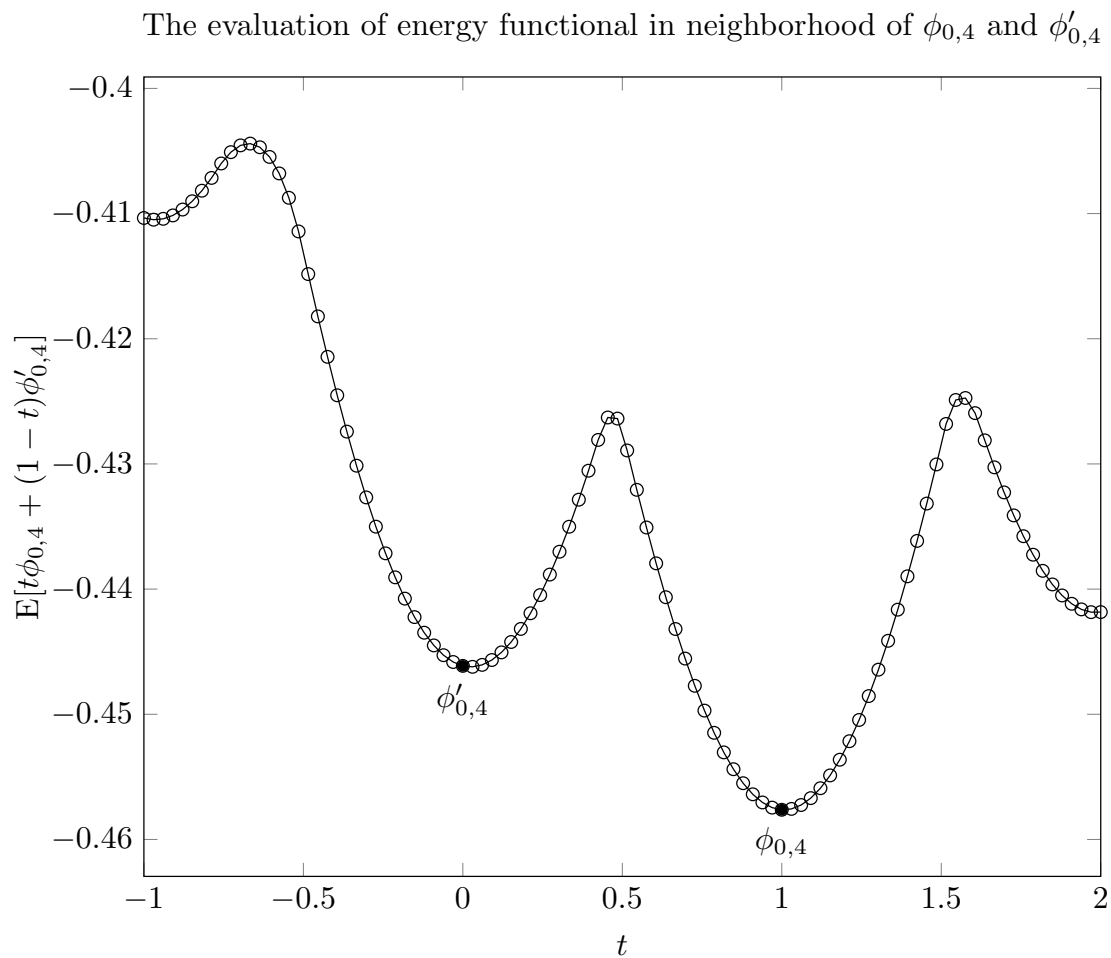}
	\caption{An evaluation of the energy function along a line passing through $\phi_{0,4}$ and $\phi'_{0,4}$. Both solutions appear to be a local minimum and the problem is ill--posed.}
	\label{fig:chapter_img_associativity}
\end{figure}





\chapter{Prefix sum}
\label{chap:prefix_sum}

In this chapter, we introduce the prefix sum and discuss parallelization strategies known from the literature. We use names \textit{prefix sum} and \textit{scan} interchangeably.

\section{Introduction}

A prefix sum is an operation accepting a sequence of elements and generating a new sequence of partial sums. The operation has been described under many names in literature, including names such as inclusive prefix sum, scan, cumulative sum. The definition is as follows:
\begin{definition}{\textbf{Prefix sum}}\label{def:inclusive_scan} A prefix sum operation applies a binary approximately associative operator $\odot$ to a sequence $(x_{i})$
\begin{align*}
x_{1}, x_{2}, x_{3}, \dots , x_{n-1}
\end{align*}
generating a new sequence $(y_{i})$
\begin{align*}
y_{1} \: &= \: x_{1} \\
y_{2} \: &= \: x_{1} \odot x_{2} \\
y_{3} \: &= \: x_{1} \odot x_{2} \odot x_{3} \\
& \dots \\
y_{n-1} \: &= \: x_{1} \odot x_{2} \odot \dots \odot x_{n-1} = \odot_{i=1}^{n-1} x_{i}
\end{align*}
\end{definition}
Our definition differs from the one commonly used in the literature in that we permit approximately associative operators. With such operator, the result may be affected by changes in the order of operator evaluations, but it is guaranteed to be correct. This extension is required to treat the image registration problem as a prefix sum.\\
We use the notation $\odot_{i=1}^{n-1}$ to describe an iterative application of the binary operator $\odot$ and $x_{1, n-1}$ to denote the product of such application. The definition above presents an \textit{inclusive} prefix sum where each new element with index $i$ is a sum of first $i$ elements of input sequence. In an alternative approach, the new $i$-th value is a sum of first $i-1$ input values. The $i$-th value is \textit{excluded}, therefore this approach is known as an \textit{exclusive} prefix sum. The method is also known under a name prescan.
\begin{definition}{\textbf{Exclusive prefix sum}} An exclusive prefix sum operation applies a binary approximately associative operator $\odot$, with an identity element $I_{\odot}$, to a sequence $(x_{i})$
	\begin{align*}
	x_{1}, x_{2}, x_{3}, \dots , x_{n-1}
	\end{align*}
	generating a new sequence $(y_{i})$
	\begin{align*}
	y_{1} \: &= \: I_{\odot} \\
	y_{2} \: &= \: x_{1} \\
	y_{3} \: &= \: x_{1} \odot x_{2}\\
	& \dots \\
	y_{n-1} \: &= \: x_{1} \odot x_{2} \odot \dots \odot x_{n-2} = \odot_{i=1}^{n-2} x_{i}
	\end{align*}
\end{definition}

Inclusive and exclusive prefix sums are closely connected with each other. After all, for an input sequence of length $n$, $n-1$ output elements are exactly the same in both scans, only placed in different positions. Obtaining an exclusive result from an inclusive prefix sum is trivial because all necessary values are already computed, and it is sufficient to shift results by one position to the right and place the identity element $I_{\odot}$ on the very first position. On the other hand, computing an inclusive prefix sum from an exclusive one may require additional computation. Results are shifted by one position to the left, and the binary operator is applied to $y_{n-2}$ and $x_{n-1}$ to compute the last reduction value $y_{n-1}$. In some algorithms, such as Blelloch parallel prefix sum described in section~\ref{subsection:blelloch}, this step is unnecessary because the full reduction has already been computed.

The scan algorithm has been originally proposed for APL programming language\cite{Iverson:1962:PL:1098666}. Prefix sum is a natural representation for a binary addition of two numbers in digital circuits\cite{Ladner:1980:PPC:322217.322232} and a significant progress have been achieved to improve the performance and internal design of arithmetical circuits known as parallel prefix adders. Modern parallel prefix algorithms are usually based on circuit design research. \\
The parallel prefix sum has been described as a common pattern and fundamental building block in parallel applications\cite{scanEncyclopedia}\cite{McCool:2012:SPP:2385466}. It has been found to be helpful and useful for implementation and parallelization of multiple computer science problems with serial dependencies, including, but not limited to, polynomial evaluation, various sorting algorithms, solving recurrence equations, graph and tree algorithms\cite{Chatterjee:1990:SPV:110382.110597}\cite{BlellochTR90}\cite{Cormen:2001:IA:580470}. The usefulness of prefix sum has motivated a proposal of \textit{scan vector model} for parallel computations, where prefix sums are given as a unit time primitive\cite{Blelloch:1990:VMD:91254}.

A proof for the importance of this algorithm is the prevalence of various implementations of the prefix sum in major programming languages. C++ Standard Template Library\cite{ISO14882} includes a sequential implementation of an inclusive prefix sum \code{std::partial\_\\sum}. The order of summation is explicitly defined by the standard and associativity of a binary operator is not required. The most recent update of standard\footnote{At the time of writing, C++17 was feature-complete but an official ISO standard has not been published.}\cite{N4659} provides several overloads of \code{std::inclusive\_scan} and \code{std::exclusive\_scan}, both sequential and parallel via appropriate execution policies. Associativity of the binary operator is required, otherwise the behavior is nondeterministic. A similar set of overloads \code{std::transform\_inclusive\_scan} and \code{std::transform\_exclusive\_scan} transforms input range with an unary operator before a prefix sum is computed. Other parallel implementations are provided with an \code{inclusive\_scan} and \code{exclusive\_scan} in OpenCL-based Boost.Compute\cite{BoostCompute} and and CUDA-based Thrust\cite{hoberock2010thrust}, a multithreaded implementation in Intel TBB\cite{Reinders:2007:ITB:1461409} and a distributed implementation of \code{MPI\_Scan} and \code{MPI\_Exscan}\cite{MPISpec}. \\
The prefix sum has been generalized to perform independently multiple scans on \textit{segments}, disjoint subsequences of input data, therefore this operation is known as a segmented prefix sum\cite{Blelloch:1990:VMD:91254}. An additional input is required to distinguish how segments are located in the input sequence $(x_{i})$. It could be a bit sequence of the same length as input data for prefix sum, indicating where new segment starts, or a shorter sequence of integers containing lengths of consecutive segments. We provide a formal definition for the former case.
\begin{definition}{\textbf{Segmented prefix sum}} A segmented prefix sum operation takes as an input two sequences of the same length, a data sequence $(x_{i})$ and a flag sequence of bits $(b_{i})$, and applies a binary approximately associative operator $\odot$, generating a new sequence $(y_{i})$
	\begin{align*}
	y_{1} \: &= \: x_{1} \\
	y_{2} \: &= \: (y_{1} \otimes b_{2}) \odot x_{2} \\
	& \dots \\
	y_{n-1} \: &= \: (y_{n-2} \otimes b_{i-1}) \odot x_{n-1}
	\end{align*}
where binary operator $\otimes$ is defined as follows
	\begin{equation*}
	x \otimes b =\begin{cases}
	I_{\odot}, & \text{if $b = 1$}.\\
	x, & \text{otherwise}.
	\end{cases}
	\end{equation*}
\end{definition}
Segmented prefix sum has been found useful in parallelization of the quicksort algorithm\cite{Blelloch:1990:VMD:91254}. An example of an implementation is \code{MPI\_Scan} which computes $k$ independent prefix sums for input array of length $k$. \\
A further variation of this technique is known as multiprefix\cite{Blelloch:2010:PA:1882723.1882748}. There, multiple exclusive scans are computed for data represented by pairs $(k, a)$, where $k$ is a key encoding in which subsequence is the value $a$ located. Segments are required to be neither contiguous nor disjoint. The multiprefix operation can be performed with work--efficient algorithm of $O(\sqrt{n})$ span\cite{Sheffler:1993:IMO:165231.166115}.

\section{Prefix sum for registration problem}

In the previous chapter, we have defined the process of applying distinct operators $\mathbf{A}$ and  $\mathbf{B}$ to register a sequence of images. In this chapter, we focus on the second phase of the process where neighbor deformations are processed to align all images to the first one. 
The initial application of function $\mathbf{A}$ to images can be performed independently and it is not relevant to the analysis. \\
A short investigation of the approach reveals a striking similarity between registration process and prefix sum. Indeed, for any final deformation $\phi_{0, i}$ we observe
\begin{equation}
\label{eq:chapter_prefix_sum_b}
\begin{split}
\phi_{0, i} &= \mathbf{B}(\phi_{0, i - 1}, \; \phi_{i - 1, i}) \\ &= \mathbf{B}( \mathbf{B}(\phi_{0, i - 2}, \phi_{i-2, i-1}), \phi_{i-1,i}) \\ &= \mathbf{B}( \mathbf{B}( \mathbf{B}(\dots) , \phi_{i-2, i-1}), \phi_{i-1,i}) \\ &= \: \phi_{0, 1} \odot_{B} \phi_{1, 2} \odot_{B} \dots \odot_{B} \phi_{i-2, i-1}
\end{split}
\end{equation}
where binary operator $\odot_{B}$ is defined as follows
\begin{equation}
\phi_{i, j} \odot_{B} \phi_{j, k} \: = \: \mathbf{B}(\phi_{i, j}, \phi_{j, k})
\end{equation}
The new operator inherits approximate associativity from function $\mathbf{B}$ and by the definition~\ref{def:inclusive_scan}, we prove that the problem ~(\ref{eq:chapter_prefix_sum_b}) may be represented in terms of a prefix sum. A stronger formulation of prefix sum with a regular associativity would not allow using function $\mathbf{B}$ as an operator for prefix sum. This result allows to express the parallelization of image registration as a parallelization of prefix sum, and the parallel prefix pattern has been found to be applicable in yet another problem.

\section{Parallel prefix sum}

Upon initial inspection, the inevitable sequential nature of prefix sum is a bad indicator for finding an efficient parallelization strategy. Undoubtedly, in this case, it is not possible to achieve a perfect linear scaling, but several decades of research have produced numerous procedures with varying efficiency. \\
Below we describe distinct approaches for parallelization of prefix sum. A lot of recent work has been done on researching and optimizing different algorithms for SIMD architectures such as GPGPU\cite{Sengupta:2007:SPG:1280094.1280110}\cite{cudaBook}\cite{CS2009PSSA}. New, hybrid strategies have been developed to fit their execution model\cite{Ha:2013:SWD:2553646.2553648}\cite{Dotsenko:2008:FSA:1375527.1375559}. Design goals and improvements are, however, related to the specific execution and memory models, such as removing bank conflicts or proper utilization of hierarchical memory.

For the simplicity of analysis, we assume that there are exactly as many workers as data elements and that the length of input data is a power of two. The Chapter~\ref{chap:distr_prefix_sum} presents the general scheme for prefix sum without those assumptions. For simplicity's sake, for all discussed algorithms we assume a constant running time $C_{\odot}$ of the binary operator $\odot$. \\
Our comparative analysis is based on the \textit{span}\cite{McCool:2012:SPP:2385466} or \textit{depth}\cite{Blelloch:1996:PPA:227234.227246} of an algorithm i.e. length of the critical path determined by the longest sequence of computation performed during the execution. A comparison of a span between serial and parallel algorithm gives an upper bound on attainable parallelism.
We provide work complexity as well, which estimates the span in a case of a fully serialized execution. An algorithm is considered to be \textit{work--efficient} if its work complexity scales linearly with the size of input data.
In our PRAM algorithms, we assume the input data to be allocated in a single block of memory with zero-based indexing. The for loop has an inclusive upper bound in our notation. As an example, a for loop iterating from $1$ to $N$ executes $N$ iterations.

We begin the description by introducing the serial algorithm for prefix sum computation on Listing~\ref{alg:serial}. Figure~\ref{fig:chapter_scan_serial_scheme} presents an example of serial prefix sum. Each layer of nodes depicts a single iteration of the algorithm and filled nodes represent an application of binary operator. The lines joining nodes represent communication between workers. The algorithm is a direct mapping from definition~\ref{def:inclusive_scan}. It is worth noting that this solution is the most optimal in terms of work--efficiency.
\begin{figure}
	\centering
	\includegraphics[width=0.8\textwidth]{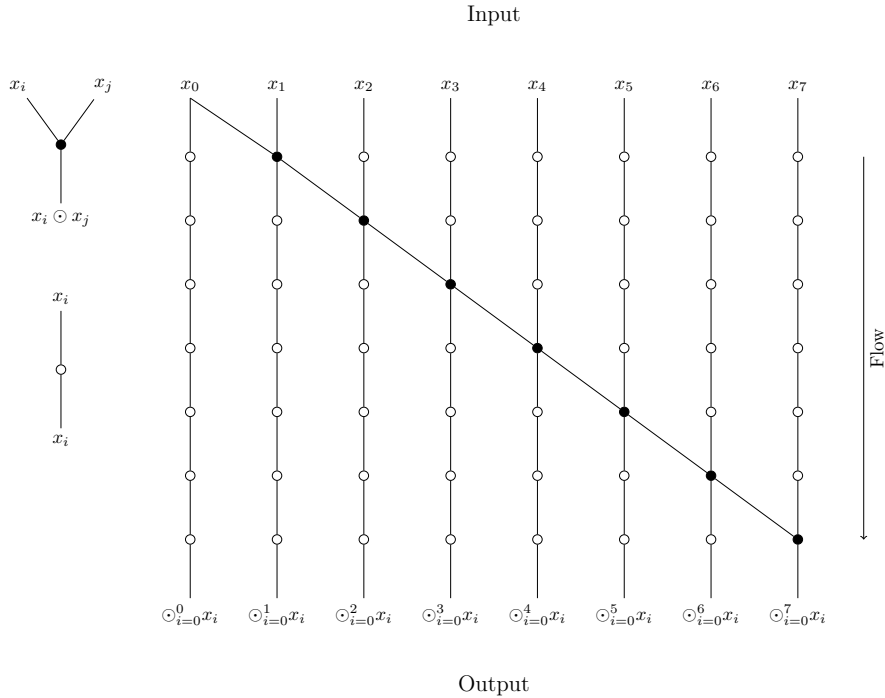}
	\caption{An example of serial prefix sum on 8 data elements. Results are produced in $7$ steps and $7$ applications of binary operator.}
	\label{fig:chapter_scan_serial_scheme}
\end{figure}
Total span of the algorithms is simply equal to $N - 1$ applications of the operator
\begin{equation}
\begin{split}
S_{S}(N) \: &= \: N - 1\\
W_{S}(N) \: &= \: N - 1
\end{split}
\end{equation}
\alglanguage{pseudocode}
\begin{algorithm}[H]
	\caption{A pseudocode for serial prefix sum of $N$ deformations.}
	\begin{algorithmic}[1]
		\ForMod{$i$}{$1$}{$N-1$}
		\State $data[i] = data[i - 1] \odot data[i]$
		\EndForMod
	\end{algorithmic}
	\label{alg:serial}
\end{algorithm}
\subsection{Blelloch scan}
\label{subsection:blelloch}
\begin{figure}
	\centering
	\includegraphics[width=0.8\textwidth]{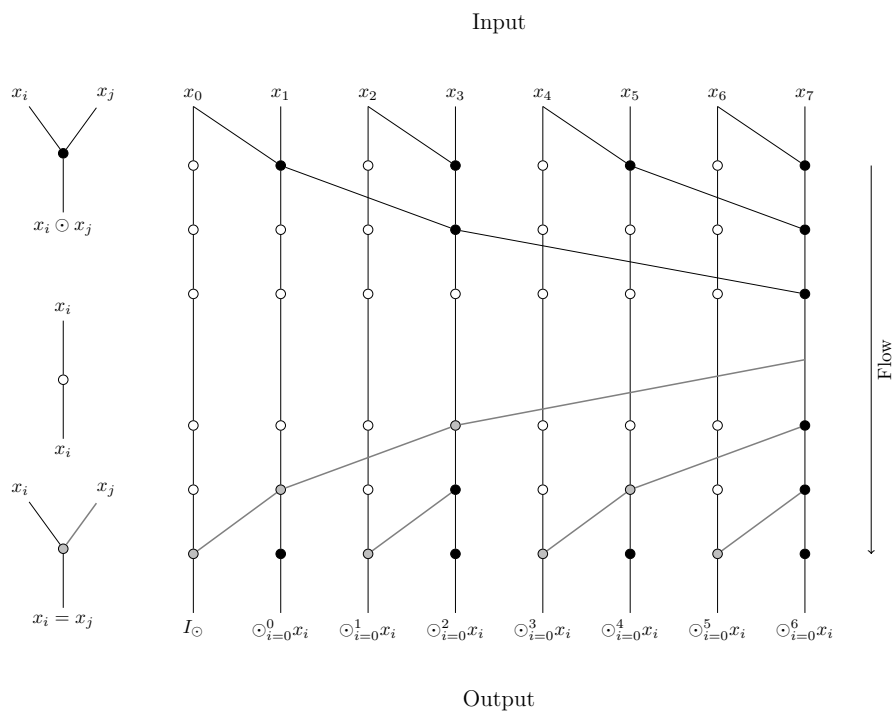}
	\caption{An example of Blelloch parallel prefix sum on 8 data elements. Results are produced in $6$ steps and $14$ applications of the binary operator. A black node represents an application of the binary operator during either up--sweep or down--sweep phase (lines 3 and 14 of Algorithm~\ref{alg:blelloch}, respectively). 	A gray node represents a part of down--sweep where a left child receives a value from its parent and no computation is performed. This node corresponds to line 12 of Algorithm~\ref{alg:blelloch}.}
	\label{fig:chapter_scan_blelloch_scheme}
\end{figure}
One of the most popular parallel prefix sum strategies has been researched and presented by Guy Blelloch\cite{Blelloch:1989:SPP:76108.76113}. The tree--based approach correlates with Brent--Kung parallel prefix adder\cite{Brent:1980:CCB:800141.804666}, described separately later. Figure~\ref{fig:chapter_scan_blelloch_scheme} depicts an example of Blelloch prefix sum on eight image deformations. \\
The algorithm is usually defined as two sweeps on a binary tree. An up-sweep, from leaves to the root, produces a reduction of all data elements, in our case $x_{0, 7}$. This procedure requires no more steps than a height of three which is equal to $\log_{2}{N}$. Then, a down--sweep is performed, from the root to leaves, and in each iteration workers proceed in a triple of a parent, left and right child. The left child, illustrated in the figure with a gray node, receives a partial result from its parent and the right child, depicted in the figure with a black node, computes another partial result. In practice, parent and right child refer to the same worker in different iterations. Once again, a full sweep requires  $\log_{2}{N}$ steps to finish. Listing~\ref{alg:blelloch} presents an example of the algorithm. 
\alglanguage{pseudocode}
\begin{algorithm}[H]
	\caption{A pseudocode for Blelloch parallel prefix sum.}
	\begin{algorithmic}[1]
		\ForMod{$i$}{0}{$\log_{2}{N}-1$} \Comment{Up-sweep traversal of the tree.}
		\ForModP{$j$}{$2^{i}$}{$N$}{$2^{i+1}$}
		\State $data[j] = data[j - 2^{i+1}] \odot data[j]$
		\EndForModP
		\EndForMod
		\State
		\State $data[N-1] = I_{\odot}$ \Comment{$data[N-1]$ stored the final reduction}
		\State
		\ForMod{$i$}{$\log_{2}{N} - 1$}{0} \Comment{Down-sweep traversal of the tree.}
		\ForModP{$j$}{$0$}{$N - 1$}{$2^{i+1}$}
		\State $temp = data[j + 2^{i} - 1]$ \Comment{Save value of left child.}
		\State $data[j+2^{i} - 1] \: = \: data[j + 2^{i+1} - 1]$ \Comment{Copy value to left child.}
		\State $data[j+2^{i+1}-1] = temp \odot data[j + 2^{i+1} - 1]$ \Comment{Apply left} \State \Comment{child's value to right child.}
		\EndForModP
		\EndForMod
	\end{algorithmic}
	\label{alg:blelloch}
\end{algorithm}
The algorithm computes an exclusive prefix sum, but an inclusive prefix sum may be computed after a small modification. For a $i$-th worker, the inclusive value may be obtained either through one additional application of operator or by receiving the value of exclusive scan from $i+1$-th worker. The former approach is preferable for applications where it is less expensive to apply the operator rather than communicate with other workers. This is especially significant on message-passing systems. On the other hand, the latter approach is desirable for examples where the binary operator is computationally intensive. In this scenario, the missing value for last worker $x_{0, N-1} $ is already computed by the same worker at the end of an up-sweep. \\
The span of Blelloch prefix sum is bounded by a double traversal of binary tree which scales logarithmically with the number of workers
\begin{equation}
S_{B}(N) \: = \: 2 \cdot \log_{2}{N}
\end{equation}
A number of applications of the operator scales linearly with the number of data elements. This makes the Blelloch algorithm work--efficient
\begin{equation}
\begin{split}
W_{B}(N) &= 2 \cdot \sum_{i=0}^{\log_{2}{N} - 1} \frac{N}{2^{i+1}} \\ &= 2 \cdot \frac{N}{2}\frac{1 - (\frac{1}{2})^{\log_{2}{N}}}{1 - \frac{1}{2}} \\ &= 2 \cdot N \cdot (1 - \frac{1}{N}) \\ &= 2 \cdot (N - 1)
\end{split}
\end{equation}
\subsection{Brent--Kung}
Brent--Kung adder\cite{Brent:1980:CCB:800141.804666} is a work-efficient circuit design for parallel prefix sum. Figure~\ref{fig:chapter_scan_brent_kung_scheme} presents an example of Brent--Kung strategy on eight image deformations.
The up--sweep traversal is exactly the same as in Blelloch prefix sum and produces a reduction of all data elements. In the down--sweep, a breadth--first search (BFS) is performed where only right child adds partial result from its parent. An implementation may not always synchronize between workers and given lack of data dependencies, some parts of second tree traversal may be performed earlier, even during the first tree visit. This improvement, however, won't improve algorithm's span and a less greedy version simplifies the description.
\\ The strategy allows computing results in $2 \cdot \log_{2}{N} - 1$ steps, making it slightly more efficient than Blelloch method. Listing~\ref{alg:brent_kung} presents the algorithm. \\
Contrary to the Blelloch algorithm, Brent--Kung produces an inclusive prefix sum. An exclusive version of Brent--Kung\cite{CS2009PSSA}, with an integrated rejection of final reduction and propagation of operator's identity, exhibits a design very similar to the Blelloch algorithm with a slightly better depth of Brent--Kung strategy.
\begin{figure}
	\centering
	\includegraphics[width=\textwidth]{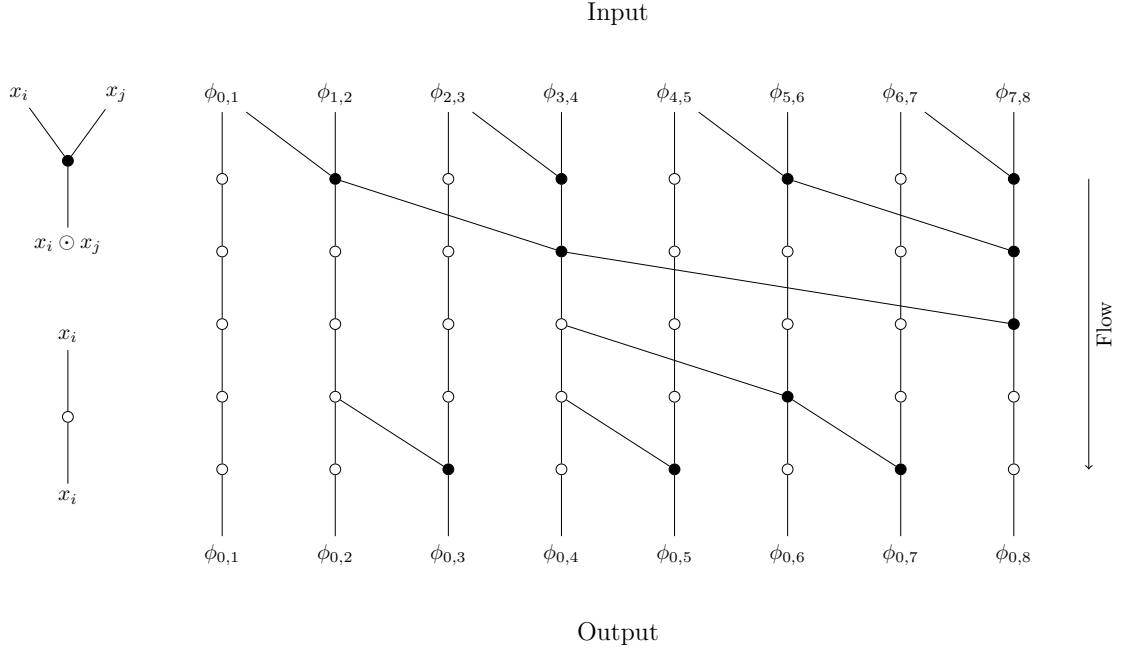}
	\caption{An example of Brent--Kung parallel prefix sum on 8 data elements. Results are produced in $5$ steps and $11$ applications of the binary operator.}
	\label{fig:chapter_scan_brent_kung_scheme}
\end{figure}
\alglanguage{pseudocode}
\begin{algorithm}[H]
	\caption{A pseudocode for Brent--Kung parallel prefix sum.}
	\begin{algorithmic}[1]
		\ForMod{$i$}{0}{$\log_{2}{N}-1$} \Comment{Up-sweep traversal of the tree.}
		\ForModP{$j$}{$2^{i}$}{$N$}{$2^{i+1}$}
		\State $data[j] = data[j - 2^{i+1}] \odot data[j]$
		\EndForModP
		\EndForMod
		\State \Comment{Iterations of an inner loop may be executed concurrently with the loop above}
		\ForMod{$i$}{$\log_{2}{N} - 2$}{0} \Comment{Down-sweep traversal of the tree.}
		\ForModP{$j$}{$2^{i+1}$}{$N - 1$}{$2^{i+1}$}
		\State \Comment{Left child is visited without any computation.}
		\State $data[j+2^{i} - 1] = data[j - 1] \odot data[j+2^{i} - 1]$
		\EndForModP
		\EndForMod
	\end{algorithmic}
	\label{alg:brent_kung}
\end{algorithm}
The span of Brent--Kung prefix sum is bounded by a full traversal of the binary tree and a breadth--first search where root does not perform any computation
\begin{equation}
S_{B}(N) \: = \: 2\log_{2}{N} - 1
\end{equation}
Estimation of performed work is similar to Blelloch case. Brent--Kung strategy is work--efficient due to linear scaling of performed work
\begin{equation}
\begin{split}
W_{B}(N) &= \sum_{i=0}^{\log_{2}{N} - 1} \frac{N}{2^{i+1}} + \sum_{i=0}^{\log_{2}{N} - 2} (\frac{N}{2^{i+1}} - 1) \\
&= N - 1 + \sum_{i=0}^{\log_{2}{N} - 2} \frac{N}{2^{i+1}} - (\log_{2}{N} - 1) \\
&= N - 1 + \frac{N}{2}\frac{1 - (\frac{1}{2})^{\log_{2}{N} - 1}}{1 - \frac{1}{2}} - (\log_{2}{N} - 1) \\
&= N - 1 + N - 2 - (\log_{2}{N} - 1) \\ &= 2 \cdot N - \log_{2}{N} - 2
\end{split}
\end{equation}
\pagebreak
\subsection{Kogge--Stone}
This parallelization strategy is based on a Kogge--Stone parallel prefix adder, firstly proposed by Peter Kogge and Harold Stone in 1973\cite{Kogge:1973:PAE:1638607.1639095}. In 1986, Hillis and Steele\cite{Hillis:1986:DPA:7902.7903} described an application of the adder to PRAM model and the algorithm has been known under a different name \textit{Hillis--Steele parallel prefix sum}. This algorithm has also been discussed under a name \textit{recursive doubling algorithm}\cite{EGECIOGLU198995}. \\ Figure~\ref{fig:chapter_scan_kogge_stone_scheme} presents an example of parallel prefix sum with Kogge--Stone approach on eight image deformations. The strategy allows to compute results in $\log_{2}{N}$ steps and each step produces $2^{N} - 1$ partial results. In each iteration, workers apply the binary operator to its own result and a partial result obtained from one of the workers on their left. Listing~\ref{alg:kogge_stone} presents the algorithm. \\
Kogge--Stone produces an inclusive prefix sum. An obvious difference with the Blelloch prefix sum is a much larger amount of work performed in all steps. Furthermore, the huge work intensity requires excessive communication. Internally, the algorithm is synchronous and each level depends on results from the previous iteration. This particular feature makes it sensitive to deviations in execution time between different applications of the binary operator $\odot$. \\
Another characteristic feature of the algorithm is the presence of Write-After-Read anti-dependencies, where a worker should not overwrite the previous result until it has been accessed. A common technique to resolve this problem is \textit{double-buffering} where two distinct arrays are used to store partial results\cite{cudaBook}.
\begin{figure}
	\centering
	\includegraphics[width=\textwidth]{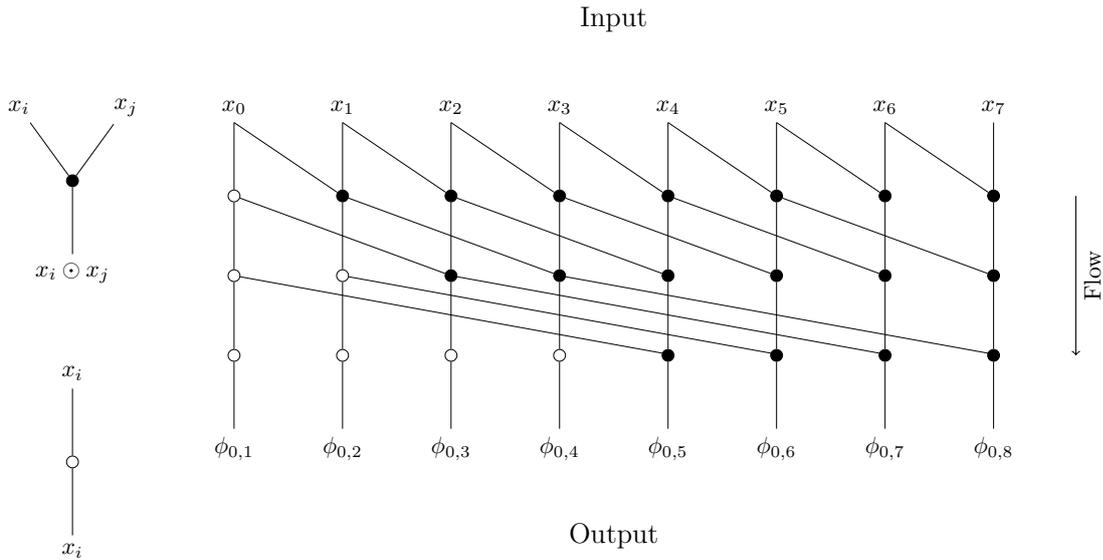}
	\caption{An example of Kogge--Stone a.k.a. Hillis--Steele parallel prefix sum on $8$ data elements. Results are produced in $3$ steps and $17$ applications of the binary operator.}
	\label{fig:chapter_scan_kogge_stone_scheme}
\end{figure}
\alglanguage{pseudocode}
\begin{algorithm}[H]
	\caption{A pseudocode for Kogge--Stone parallel prefix sum.}
	\begin{algorithmic}[1]
		\ForMod{$i$}{0}{$\log_{2}{N}-1$}
		\ForModPS{$j$}{$2(i + 1)$}{$N$}
		\State $data[j] = data[j - 2^{i}] \odot data[j]$
		\EndForModPS
		\EndForMod
	\end{algorithmic}
	\label{alg:kogge_stone}
\end{algorithm}
An estimation of time complexity for Kogge--Stone parallel prefix sum is trivial and the span is simply determined by the outer loop
\begin{equation}
S_{KS}(N) \: = \: \log_{2}{N}
\end{equation}
Work complexity is assessed by multiplying the number of steps with a number of active workers in each step
\begin{equation}
\begin{split}
W_{KS}(N) &=  \sum_{i=0}^{\log_{2}{N} - 1} N - 2^{i} \\ &= N \cdot \log_{2}{N} - \sum_{i=0}^{\log_{2}{N} - 1} 2^{i} \\ &=  N \cdot \log_{2}{N} - \frac{1 - 2^{log_{2}{N}}}{1 - 2} \\ &= N \cdot \log_{2}{N} - N + 1
\end{split}
\end{equation}
\subsection{Sklansky}
\label{chap:prefix_sum_sklansky}
This strategy is based on a first parallel prefix adder described in 1960 by Sklansky\cite{5219822}. This inclusive parallel prefix sum is similar to Kogge--Stone in work inefficiency and purely logarithmic span. The recursive nature of algorithm is visible on Figure~\ref{fig:chapter_scan_sklansky_scheme}. A divide--and--conquer approach splits the problem in half at each step, instating twice the same task for two halves of input data. \\
The strategy generates results in $\log_{2}{N}$ steps. According to a recent report\cite{CS2009PSSA}, their proposed algorithmic description for Sklansky prefix sum is the first iterative version of Sklansky prefix adder. We present a more verbose version of the algorithm on Listing~\ref{alg:sklansky}. In our opinion, it is easier to follow the flow of execution with a triple--nested loop. \\
Comparing to previously introduced prefix adders, Sklansky is the only one to have a non-constant number of \textit{fan-outs} i.e. outbound wires in a logical gate applying the operator. The example on Figure~\ref{fig:chapter_scan_sklansky_scheme} shows how a number of outputs in a node changes from two to four. Furthermore, the algorithm involves a constant number of tasks per each iteration which simplifies the mapping of work to hardware in implementations such as SIMD architectures.
\begin{figure}
	\centering
	\includegraphics[width=\textwidth]{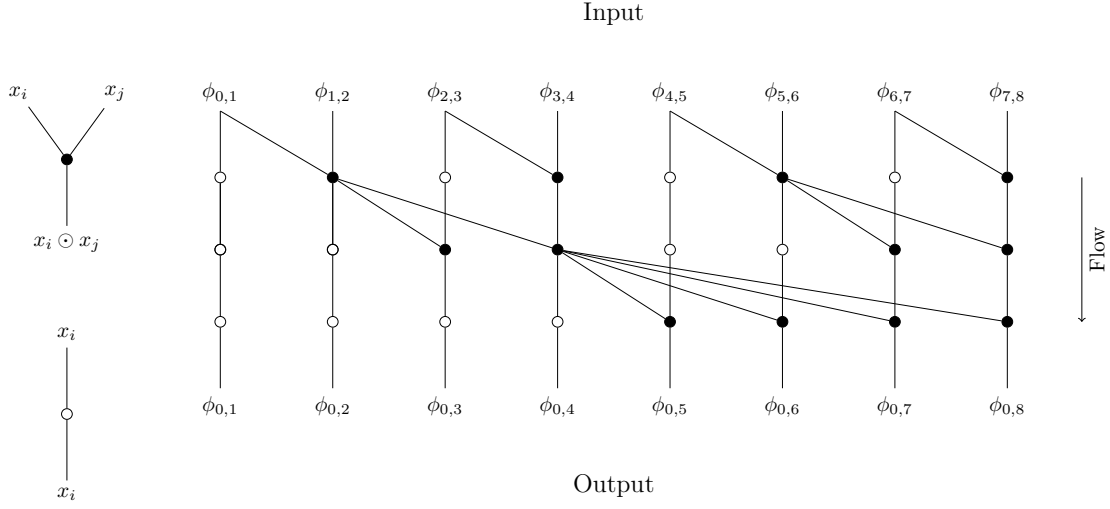}
	\caption{An example of Sklansky parallel prefix sum on $8$ image deformations. Results are produced in $3$ steps and $12$ applications of the binary operator.}
	\label{fig:chapter_scan_sklansky_scheme}
\end{figure}
\alglanguage{pseudocode}
\begin{algorithm}[H]
	\caption{A pseudocode for Sklansky parallel prefix sum.}
	\begin{algorithmic}[1]
		\ForMod{$i$}{0}{$\log_{2}{N}-1$}
			\LineComment{Iterate over sources of result from previous step.}
			\ForModP{$j$}{$2^{i} - 1$}{$N$}{$2^{i+1}$} 
			\LineComment{Iterate over destinations for result from previous step.}
				\ForModPS{$k$}{$0$}{$2^{i}$}
					\State $data[j+k+1] = data[j] \odot data[j+k+1]$
				\EndForModPS
			\EndForModP
		\EndForMod
	\end{algorithmic}
	\label{alg:sklansky}
\end{algorithm}
The span of the Sklansky prefix sum is given by a divide--and--conquer method needing $\log_2{N}$ steps for $N$ input elements
\begin{equation}
S_{SK}(N) \: = \: \log_{2}{N}
\end{equation}
Work complexity is straightforward as well. Each recursive call creates twice as many problems with a half of the original size, and therefore each level processes the same number of tasks
\begin{equation}
\begin{split}
W_{SK}(N) &= \frac{N}{2} \cdot \log_{2}{N}
\end{split}
\end{equation}
\section{Relation between span and work}

Introduced parallel prefix sum algorithms have different work complexities, but neither of them can match the purely linear complexity of a serial application. The intuition suggests that the sequential nature of prefix sum allows for parallelization only by performing more work but in parallel. This intuition has been formalized and certain bounds for the relation between span and work of a parallel prefix sum has been proven. These results have been described with terminology appropriate for prefix adder circuits, but semantics stay the same and we can apply directly these results to PRAM algorithms.

An important theorem has been proposed and proved for prefix circuits by Snir\cite{Snir:1986:DTP:8088.8091} in 1986. The theorem introduces a relation between size and depth of a prefix circuit. The former property describes the number of nodes inside the circuit and corresponds to work complexity $W(N)$ in a parallel algorithm. The latter represents a delay introduced by the circuit and it primarily depends on the critical path. Therefore, this attribute corresponds to span in parallel computation model.
\begin{theorem}
Let $x_{0}, x_{1}, \dots, x_{N-1}$ be a sequence of inputs and $f_{i}$ be prefix sums computed with $i$ first elements of input sequence. Let G be a prefix circuit that computes $f_{1}, \dots, f_{N-1}$, with a size $s(G)$ and depth $d(G)$. Then
\[ s(G) + d(G) \ge 2N -2 \]
\end{theorem}
\begin{proof}
	The original proof can be found in\cite{Snir:1986:DTP:8088.8091}. Alternative proof not requiring induction can be found in\cite{Zhu:2006:CZP:1142155.1142162}.
\end{proof}
An immediate corollary of this theorem is that each gain in improving parallelism, which reduces critical path of the algorithm, has to be compensated by performing more work. The question remains whether the excessive work is justified by an improvement in depth. The concept of zero--deficiency measures if there exists a linear trade-off between size and depth:
\begin{definition}
	The deficiency of a prefix circuit is defined as
	\[ \mathrm{def}(G) = 2N - 2 - s(G) - d(G) \]
	A parallel prefix circuit is said to be of zero--deficiency if  $\mathrm{def}(G) = 0$
\end{definition}
A method for constructing zero--deficiency prefix sums has been proposed for depths in the range $2\log_{2}{N} - 2 \le d(G) \le N - 1$. In 2006 a lower bound for depth of zero--deficiency prefix circuit for a given $N$ was proven\cite{Zhu:2006:CZP:1142155.1142162}. Zero--deficiency prefix circuits do not exist below this boundary, including prefix circuits of minimal depth $\log_{2}{N}$, such as Sklansky or Kogge--Stone. Therefore, the most span--optimal parallel algorithm for prefix sum can not achieve a linear work complexity. 

A trivial example of zero-deficiency prefix sum is the serial algorithm with exactly $N - 1$ applications of a binary operator and $N - 1$ span. 

\section{Other work}
A hybrid parallel prefix strategy has been proposed by Han and Carlson\cite{6158699}, where Brent--Kung and Kogge--Stone prefix adders are merged into a single algorithm. The goal of a new layout is to leverage an optimal span of Kogge--Stone and a linear work complexity of Brent--Kung. The algorithm is parameterized by a non-zero constant $k$ which controls the balance between span and work complexity. Span is equal to $k + \log_{2}{N}$ and it has been proved\cite{Ha:2013:SWD:2553646.2553648} that a proper choice of $k$ allows to bound work complexity by $\mathcal{O}(N \cdot \log_{2}{N})$ or even $\mathcal{O}(N)$. Thus, Han--Carlson strategy achieves asymptotically optimal span and work complexity.

Ladner and Fischer\cite{Ladner:1980:PPC:322217.322232} proposed a general recurrence method of designing circuits for prefix adders. An interesting application of their method involves a prefix sum using Sklansky and Brent--Kung strategies to attain a minimum span and a slightly better work complexity\cite{1292373}\cite{Hinze04analgebra}. However, some literature describes Sklansky prefix sum under the name of Ladner--Fischer scan\cite{6970664}.



\section{Summary}

We have described several different algorithms that have been developed and researched for parallelization of prefix sum. The next chapter defines and explains the methodology for choosing a right parallel prefix strategy for the problem of image registration.

Table~\ref{table:chap_prefix_sum_summary} presents a comparison of discussed parallel prefix sum algorithms. Kogge--Stone with Hills--Steele remain a widely-used strategy due to its minimal span. On the other hand, Sklansky prefix adder does not seem to be as popular. Blelloch and Brent--Kung strategies are still highly influential and prevalent. Less popular prefix adders, such as Han--Carlson, Ladner--Fischer or Hockney--Jesshope have been recently researched for GPGPU architectures.
\begin{table}[h!]
\centering
	\begin{tabular}{cccc}
		\hline
		Name 		& Type		& Span						& Work \\ \hline \noalign{\smallskip}
		Sequential  & Inclusive & $N - 1$					& $N-1$ \\
		Blelloch 	& Exclusive	& $2 \cdot \log_{2}{N}$		& $2(N-1)$ \\
		Brent--Kung	& Inclusive	& $2 \cdot \log_{2}{N} - 1$	& $2 \cdot N - \log_{2}{N} - 2$ \\
		Kogge--Stone& Inclusive	& $\log_{2}{N}$				& $N \cdot \log_{2}{N} - N + 1$ \\
		Sklansky	& Inclusive	& $\log_{2}{N}$				& $\frac{N}{2} \cdot 		\log_{2}{N}$ \\ 
		&  &  &  \\ 
		&  &  &  \\ 
		&  &  &  \\
	\end{tabular}
	\caption{Comparison of major strategies for parallel prefix sum.}
	\label{table:chap_prefix_sum_summary}
\end{table}

\chapter{Distributed prefix sum}
\label{chap:distr_prefix_sum}


In this chapter, we consider a distributed implementation of an inclusive parallel prefix sum described in chapter~\ref{chap:prefix_sum}. Presented strategies have been developed to attain a best theoretical speedup when executing in a cluster environment. Furthermore, we demonstrate how the distributed prefix sum can be implemented within MPI programming model. \\
Algorithms accept input data of length $N$ and operate on $P$ workers with separate address spaces. A worker is expected to obtain multiple data elements. We use the terms \textit{worker} and \textit{process} interchangeably. For MPI--based implementations we use the term rank as well. Workers are allocated in a one--dimensional grid and for a worker with index $I$, we use terms \textit{successor} and \textit{right neighbor} for a worker with index $I+1$, if it has one. Similarly, names \textit{predecessor} and \textit{left neighbor} are used interchangeably to refer to a worker with index $I-1$.

\section{General strategy}

We have seen in the previous chapter that the span--optimal solution scales logarithmically with the number of workers when each one is responsible for one data element. Logarithmic complexity may be desired for a time complexity of a serial algorithm, but in a parallel algorithm it prevents any major improvements by spawning more workers on a larger set of cluster nodes. In a distributed setting, we expect to have significantly fewer workers than data. Hence, it is preferable to reduce the intra--process part to $\log_{2}{P}$. Thus, this distributed stage requires one input value per a worker and a reduction step is required to transform the input data from length $N$ to length $P$.\\
\begin{figure}
	\centering
	\includegraphics[width=\textwidth]{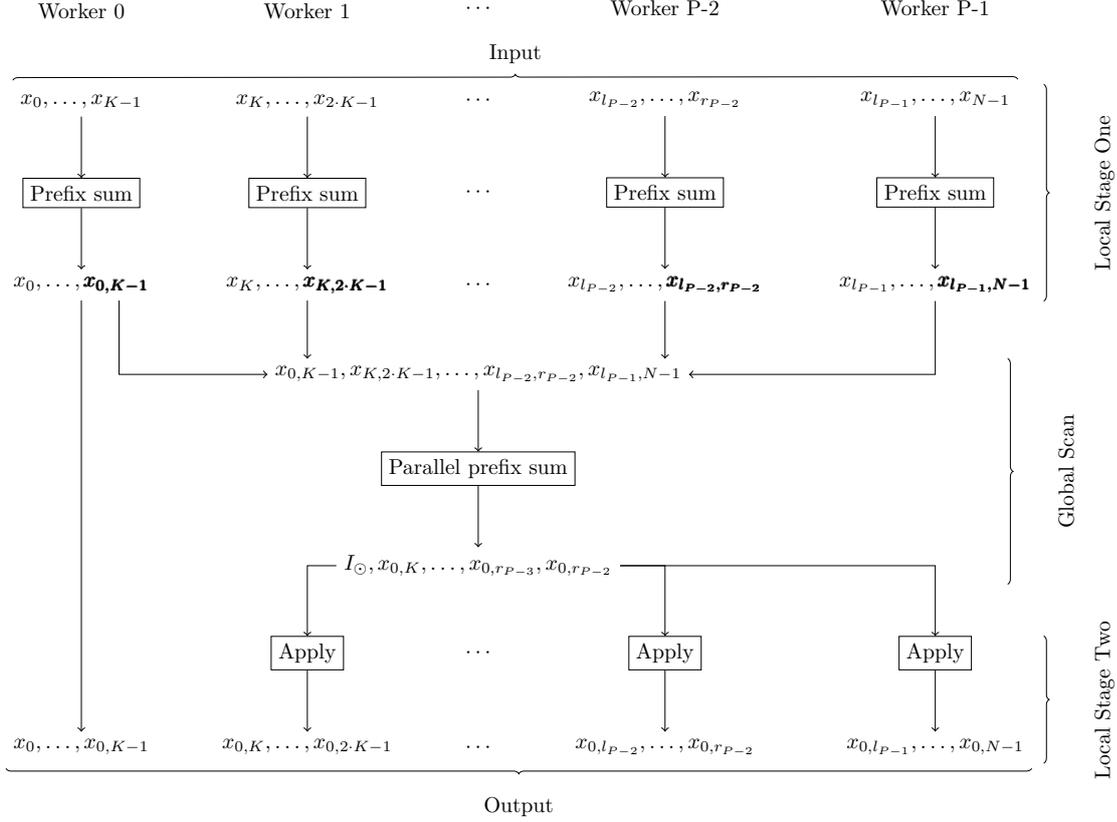}
	\caption{The general strategy for distributed prefix sum of $N$ elements on $P$ workers. An evenly distributed data across workers is passed to a local sequential prefix sum and the last computed value becomes an input to the global stage. Results from a global parallel prefix sum are applied on each worker except the first one.}
	\label{fig:chapter_distr_scan_general_scheme}
\end{figure}
We name the first part of the general strategy a \textit{local stage one}. It should accept the whole sequence of input data and end with a one value per process. For the purpose of this description, we do not assume any a priori knowledge which could suggest a specific data distribution policy. Without any hints on the actual running time of the binary operator on different operands, the safest choice is to split data as equally as possible over all workers. Each process is assigned $K \: = \: \frac{N}{P}$ input elements, if $P$ evenly divides $N$. In other case, $K$ shall be equal to $\left \lfloor \frac{N}{P} \right \rfloor + 1$, first $N \bmod P$ workers are assigned $K$ elements, and the rest obtains $K - 1$ data elements. For the sake of simplicity, we assume an evenly distributed data. \\
The logical distribution of work across processors follows a 1D grid where $I$-th worker is responsible for $K$ input elements from $x_{K \cdot I}$ to $x_{K \cdot (I + 1) - 1}$. To simplify notation, helper variables $l_{I}$ and $r_{I}$ are introduced to store left and right boundary for a worker, with $l_{I}$ equal to the index of first data element $K \cdot i$ and $r_{I}$ equal to the index of last data element $K \cdot (i + I) - 1$. Data layout is depicted on Figure~\ref{fig:chapter_distr_scan_general_scheme}.

The first local stage is presented on lines 1--3 in Listing~\ref{alg:distributed_parallel}. Each worker performs performs independently a sequential prefix sum on the assigned chunk of data $x_{l_{i}}, x_{l_{i} + 1}, \dots, x_{r_{i}}$. Local data is overwritten with partial results and the last item is a reduction of all $K$ data items $\odot_{i=0}^{K-1} x_{ l_{I} + i}$. \\
Span is estimated as for a serial prefix sum
\begin{align}
S_{LS1}(N, P) &= S_{S}(\frac{N}{P}) = \frac{N}{P} - 1
\end{align}
Work complexity is estimated as $P$ workers performing a serial prefix sum
\begin{equation}
\begin{split}
W_{LS1}(N, P) &= P \cdot S_{S}(\frac{N}{P}) \\ &= N - P
\end{split}
\end{equation}
\alglanguage{pseudocode}
\begin{algorithm}[H]
	\caption{Distributed parallel prefix sum of $N$ data elements on worker $I$.}
	\begin{algorithmic}[1]
		\ForMod{$i$}{$1$}{$K - 1$}
		\State $data[i] \gets data[i - 1] \odot data[i]$ \Comment{Local Stage One}
		\EndForMod
		
		\State $excl\_scan \gets parallel\_scan(data[K - 1])$ \Comment{An exclusive prefix sum}
		
		\If{$I > 0$}
			\ForMod{$i$}{$0$}{$K - 1$}\Comment{Local Stage Two}
			\State $data[i] = excl\_scan \odot data[i]$ 
			\EndForMod
		\EndIf
	\end{algorithmic}
	\label{alg:distributed_parallel}
\end{algorithm}


The second stage, a global parallel scan, computes a prefix sum over local reductions. After the parallel scan, processor $I$ should receive a value which allows combining its local results $\odot_{i=0}^{j} x_{l_{I}+j}$ with a reduction of all values assigned to preceding workers $0, 1, \dots, I - 1$. Consequently, the global prefix sum has to produce value $\odot_{i=0}^{r_{I-1}} x_{i}$ for worker $I$ and the scan should be exclusive. Some of the proposed algorithms are by default inclusive, but they can be applied here without modifications changing their behavior, as described in section~\ref{chap:distr_prefix_sum_incl_excl}. \\
The last stage, presented on lines 5--9 in algorithm listing, operates again locally and independently from other processes. As soon as a result from the parallel stage has arrived, it is applied to each data item, and local results are transformed into partial results for a global prefix sum. The corner case here is the first worker who does not have any predecessors, and it does not perform any computation after first local stage.\\
This stage is a single loop with $K$ iterations, and therefore span and work analysis are trivial:
\begin{align}
S_{LS2}(N, P) &= \frac{N}{P} = S_{LS1}(N, P) + 1
\end{align}
\begin{equation}
\begin{split}
W_{LS2}(N, P) &= (P - 1) \cdot \frac{N}{P} \\ &= N - \frac{N}{P}
\end{split}
\end{equation}
We summarize the strategy by combining estimations for local stages and an unknown span $S_{GS}$ and work $W_{GS}$ of a global scan which depends only on the number of workers, not on the input size. The span and work for a distributed scan is given as follows
\begin{equation}
\begin{split}
S_{DS}(N, P) &= S_{LS1}(N, P) + S_{GS}(P) + S_{LS2}(N,P) \\
&= 2 \cdot \frac{N}{P} - 1 + S_{GS}(N, P)
\end{split}
\end{equation}
\begin{equation}
\begin{split}
W_{DS}(N, P) &= W_{LS1}(N, P) + W_{GS}(P) + \cdot W_{LS2}(N,P) \\
&= 2 \cdot N - P - \frac{N}{P} + W_{GS}(N, P)
\end{split}
\end{equation}
We must remark that this analysis of critical path is possible only in the condition of an even distribution of data. In the opposite case, a simple summation of the span for two local stages might yield an incorrect result if critical paths for those stages are provided by different workers. An example of such situation may be a prefix sum where $N \bmod P = 1$. There, the span of the first stage is given by the first worker who has one more data element than other processes, but it is inactive in the second local stage.

This general strategy attempts to minimize the asymptotically logarithmic global stage and perform locally as much computation as possible. The speedup of a local stage should scale linearly with an increase in a number of workers, and it is expected that for large values of $P$ and small values of $K$, the global part is going to dominate the runtime of application because of its poor scaling. \\
Moving work from global to local stages may improve the runtime not only because of a better scalability of local stage. For all parallel prefix sum algorithms introduced in chapter~\ref{chap:prefix_sum}, it holds that for the first step, each worker has to receive a value from its left neighbor. In this strategy, this value is the last result computed in the local stage which creates a Read--After--Write (RAW) dependency of the global scan on the local reduction stage. Any computational imbalance in the latter may influence the former, and the flow dependency increases the negative influence of time deviations on total runtime. Besides that, global scan requires sending partial results after each iteration which makes it even more sensitive to variations in execution time of the binary operator. Local stages are free of those dependencies. 

\subsection{Scan versus reduce}
\label{chap:distr_prefix_sum_alternative}

The general strategy presented above is not the only possible way of organizing work in a distributed prefix sum. One can notice that the global stage requires only the last value of local prefix sum which is also a result of performing a reduction on input data. Hence, preparing input for global stage does not require storing prefix sum, and those intermediate results can be recomputed in second local stage with a little cost. \\
After the global stage, a worker operates on received result from an exclusive scan and unmodified input data. To transform an input value $x_{l_{I} + j}$ to $x_{0} \odot x_{1} \odot \dots \odot x_{l_{I} + j}$, scan result is merged with the first element $x_{l_{I}}$ and a new partial result $x_{0, l_{I}} = x_{0, r_{I-1}} \odot x_{l_{I}}$ is computed. Then, a sequential prefix sum would propagate changes from global scan and desired results are computed as $x_{0, l_{I} + j} = x_{0, l_{I}} \odot (x_{l_{I} + 1} \odot x_{l_{I}+2} \odot \dots \odot x_{l_{I} + j})$. Alternative strategy is presented in Figure~\ref{fig:chapter_distr_scan_alternative_scheme} and in Listing~\ref{alg:distributed_alternative_parallel}.
\begin{figure}
	\centering
	\includegraphics[width=\textwidth]{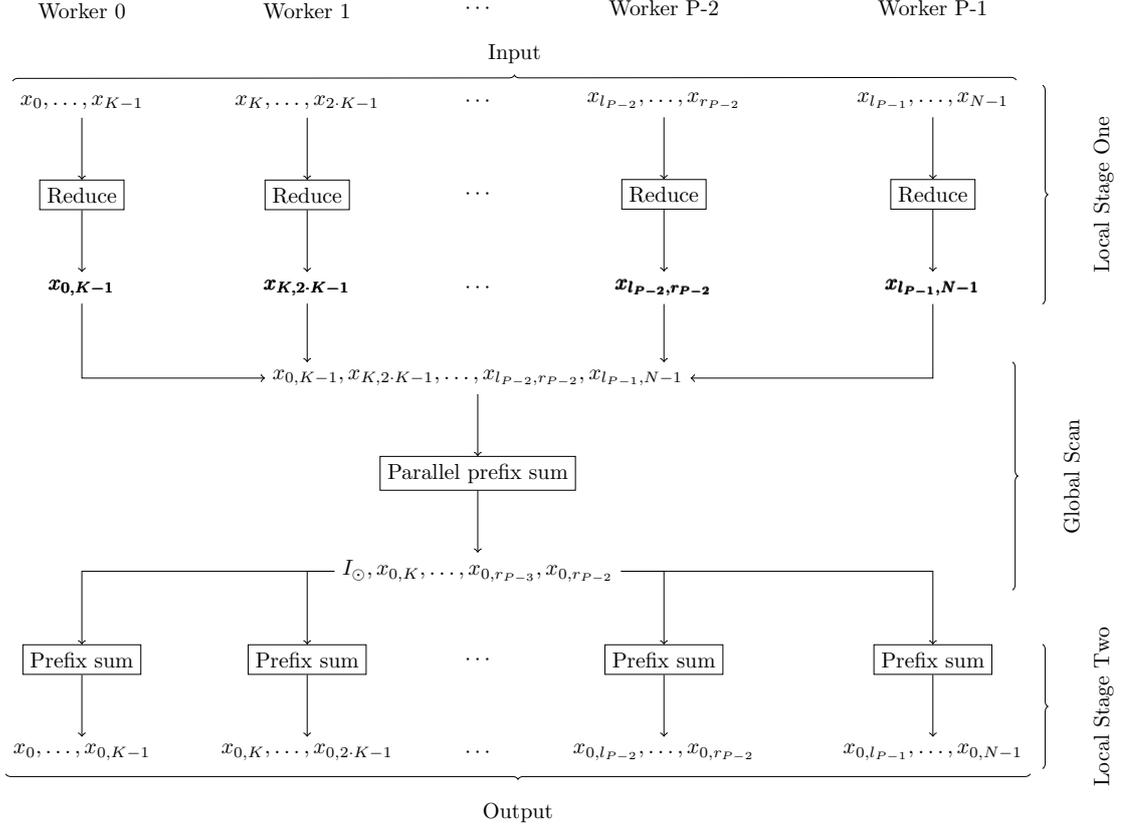}
	\caption{An alternative strategy of a distributed prefix sum of $N$ elements on $P$ workers. A local prefix sum is replaced by a reduction generating just one value as a result. In the second stage, the global scan value is applied to the first element of input data and a sequential prefix sum is performed.}
	\label{fig:chapter_distr_scan_alternative_scheme}
\end{figure}
An alternative second local stage performs the same number of loop iterations as the old one but one additional application of binary operator is necessary
\begin{align}
S_{ALS2}(N, P) &= \frac{N}{P} + 1
\end{align}
More work is performed because the worker with index $0$ is active in all stages, hence 
\begin{equation}
\begin{split}
W_{ALS2}(N, P) &= P \cdot (\frac{N}{P} + 1) \\ &= N + P
\end{split}
\end{equation}
And the span and work for an alternative distributed scan is as follows
\begin{equation}
\begin{split}
S_{ADS}(N, P) &= S_{LS1}(N, P) + S_{GS}(P) + S_{ALS2}(N,P) \\
&= 2 \cdot \frac{N}{P} + S_{GS}(N, P)
\end{split}
\end{equation}
\begin{equation}
\begin{split}
W_{ADS}(N, P) &= W_{LS1}(N, P) + W_{GS}(P) + W_{ALS2}(N,P) \\
&= 2 \cdot N + W_{GS}(N, P)
\end{split}
\end{equation}
\alglanguage{pseudocode}
\begin{algorithm}[H]
	\caption{An alternative distributed parallel prefix sum of $N$ data elements on worker $I$.}
	\begin{algorithmic}[1]
		\State $red \gets data[0]$
		\ForMod{$i$}{$1$}{$K - 1$}
		\State $red \gets red \odot data[i]$ \Comment{Local Stage One}
		\EndForMod
		
		\State $excl\_scan \gets parallel\_scan(red)$ \Comment{An exclusive prefix sum}
		
		\State $data[0] = excl\_scan \odot data[0]$ \Comment{Apply scan result to first value}
		\ForMod{$i$}{$1$}{$K - 1$}\Comment{Local Stage Two}
		\State $data[i] \gets data[i - 1] \odot data[i]$
		\EndForMod
	\end{algorithmic}
	\label{alg:distributed_alternative_parallel}
\end{algorithm}
\subsection{Inclusive global scan}
\label{chap:distr_prefix_sum_incl_excl}

In both strategies, an exclusive partial sum is required to start processing the second local stage. Hence, an exclusive prefix sum is preferred as an algorithm for the global scan. This, however, does not exclude the possibility of using an inclusive scan. All exclusive results are already computed, but they are not placed correctly on workers. Thus, an additional round of communication is necessary. A worker $I$ sends its result to its successor $I+1$, if it has one, and receives a new result from the predecessor $I-1$. \\
A downside of this solution is that additional communication forces each worker to wait for its left neighbor and adds another flow dependency. Sadly, many prefix sum algorithms are inclusive by default. However, it is possible to benefit from this situation by reducing the computation cost in last local stage. An inclusive result for $x_{0, r_{I}}$ is exactly the result of last loop iteration of the second local stages. Thus, one can directly assign this value and skip one loop iteration. Algorithm~\ref{alg:general_strategy_inclusive_scan} formally defines the case for an inclusive global scan. 
\alglanguage{pseudocode}
\begin{algorithm}[H]
	\caption{The general strategy with an inclusive scan. The inclusive result replaces the last loop iteration in the second local stage.}
	\begin{algorithmic}[1]
		\ForMod{$i$}{$1$}{$K - 1$}
		\State $data[i] \gets data[i - 1] \odot data[i]$ \Comment{Local Stage One}
		\EndForMod
		
		\State $incl\_scan \gets parallel\_scan(data[K - 1])$ \Comment{An exclusive prefix sum}
		\State $excl\_scan \gets blocking\_receive(I - 1)$ \Comment{Receive from left neighbor}
		\If{$I > 0$}
		\ForMod{$i$}{$1$}{$K - 2$}\Comment{Local Stage Two}
		\State $data[i] = excl\_scan \odot data[i]$ 
		\EndForMod
		\State $data[K - 1] \gets incl\_scan$
		\EndIf
	\end{algorithmic}
	\label{alg:general_strategy_inclusive_scan}
\end{algorithm}
\pagebreak
The span of a second local stage combined with an inclusive global scan is reduced by a constant factor. Nevertheless, this small improvement may be reduced or even eliminated by an additional synchronization after the global scan
\begin{equation}
\begin{split}
S_{ILS2}(N, P) &= S_{LS2}(N, P) - 1 \\
&= \frac{N}{P} - 1
\end{split}
\end{equation}
This optimization may be explored for exclusive prefix sums as well, especially if non--blocking communication is available. Scan result from a successor may reduce the computation in the second local stage by one loop iteration. Non--blocking communication can be applied to avoid synchronization because this value is not required until the second local stage is finished. Algorithm~\ref{alg:general_strategy_exclusive_scan} presents the application of optimization to an exclusive scan.
\begin{algorithm}[H]
	\caption{An optimized general strategy with an exclusive scan. Function $probe\_receive$ is a simplified notation for a function intended to return a boolean true value only if a specific message has been received. In an actual MPI implementation, \code{MPI\_Test} or \code{MPI\_Wait} may be used.}
	\begin{algorithmic}[1]
		\ForMod{$i$}{$1$}{$K - 1$}
		\State $data[i] \gets data[i - 1] \odot data[i]$ \Comment{Local Stage One}
		\EndForMod
		
		\State $excl\_scan \gets parallel\_scan(data[K - 1])$ \Comment{An exclusive prefix sum}
		\If{$I > 0 \And I < P - 1$} \Comment{Last worker does not have a successor}
		\State $incl\_scan \gets non\_blocking\_receive(I + 1)$ \Comment{Receive from neighbor}
		\EndIf
		\If{$I > 0$}
		\ForMod{$i$}{$1$}{$K - 2$}\Comment{Local Stage Two}
		\State $data[i] = excl\_scan \odot data[i]$ 
		\EndForMod
		\If{$(I > 0 \And I < P - 1) \And probe\_receive(incl\_scan)$}
		\State $data[K - 1] \gets incl\_scan$
		\Else
		\State $data[K - 1] \gets excl\_scan \odot data[K-1]$
		\EndIf
		\EndIf
	\end{algorithmic}
	\label{alg:general_strategy_exclusive_scan}
\end{algorithm}
We do not expect this change to reduce the theoretical span because it does not apply to the last worker which may leave the critical path unaffected. However, it is worth considering this small improvement in applications where differences in computation time lead to large imbalances. There, it is not uncommon for processes to finish much later than the last one, even if lengths of their computation paths are exactly the same. 

\section{Related work}

Literature review reveals that an approach similar to the general strategy has been evaluated for GPGPU architectures. Our general strategy is known there under the name \textit{scan--then--propagate}\cite{nvidiaScan}. This strategy has also been presented as an algorithm for prefix sum where a limited number of processors is available\cite{EGECIOGLU198995}. \\
The idea behind an alternative strategy has been presented in the algorithm for vectorization of prefix sum on CRAY Y-MP by Chatterjee et. al.\cite{Chatterjee:1990:SPV:110382.110597}. In addition, this approach has been described as the \textit{reduce--then--scan} strategy for GPGPU architectures\cite{Dotsenko:2008:FSA:1375527.1375559}. There, it may outperform the general strategy because of a smaller global memory footprint in reduction phase. \\
\textit{Pipelined} binary trees have been proposed for a distributed implementation of MPI scan collective\cite{Sanders2006}. Later, the performance of prefix sum in message--passing systems has been improved by exploiting a bidirectional communication\cite{Sanders2007}\cite{SANDERS2009581}. This research has been focused on reducing the communication cost and improving bandwidth. Furthermore, only simple memory--bound operators have been evaluated. As explained in the next section, the requirements for a distributed prefix sum in image registration problem are quite different.

\section{Registration problem}

In the previous chapter, we have proved that the problem of image registration can be represented as a prefix sum with the function \textbf{B} as a binary operator. This function has several important properties which make our problem much more specific and different from applications and case studies analyzed in the literature. \\
\begin{figure}
	\centering
	\includegraphics[width=\textwidth]{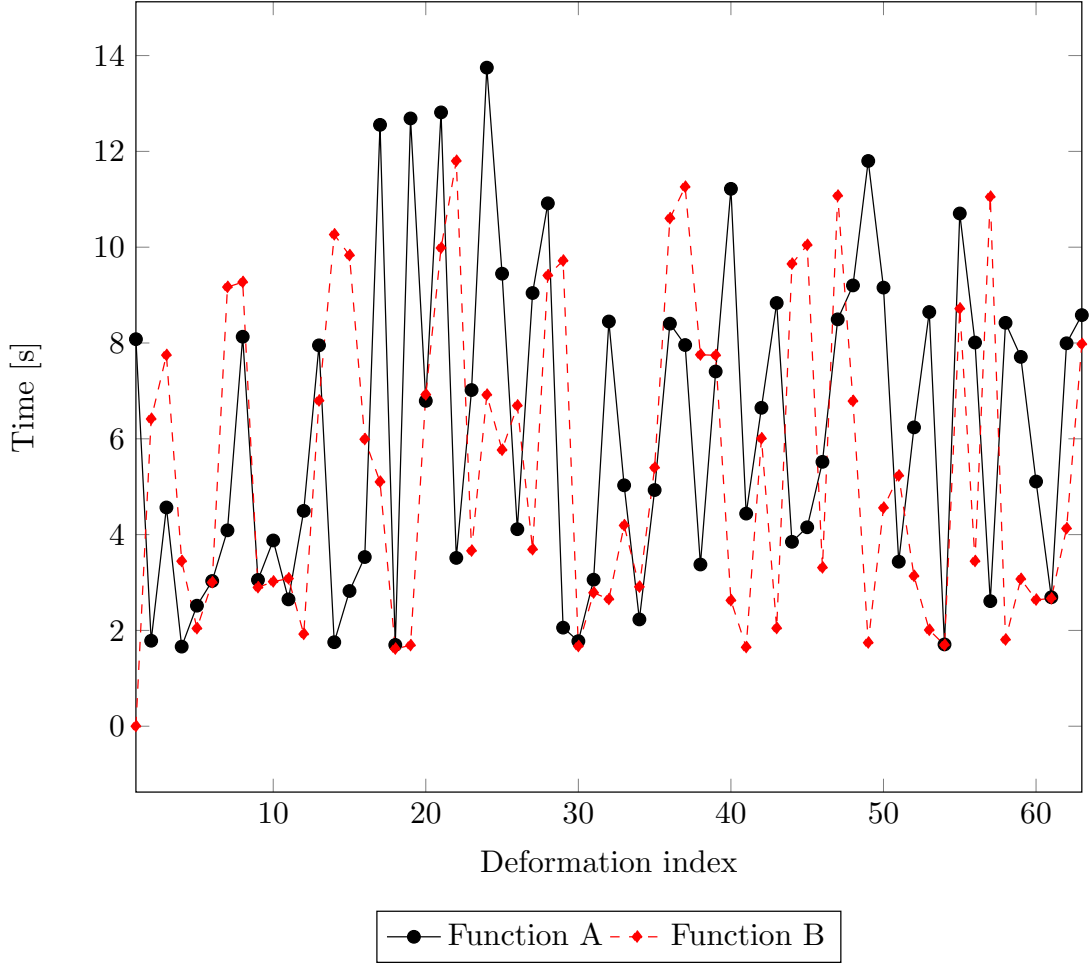}
	\caption{An example of serial registration for 64 images. The function \textbf{A} has been applied 64 times to create neighbor deformations $\phi_{i, i + 1}$ and a serial prefix sum applied function \textbf{B} 63 times to generate final deformations $\phi_{0, i}$.}
	\label{fig:chapter_distr_scan_execution_time}
\end{figure}
We begin with the actual cost of applying the operator. The simplest case found in the literature, which happens to be the one most frequently evaluated, is an integer addition which should not take more than one CPU cycle on modern processors. More complex examples still involve relatively cheap operations, such as polynomial evaluation with floating--point multiplication and addition or summed area table where the binary operator performs multiple additions. \\
As a result, parallel prefix sum algorithms tend to be optimized for memory--bound applications with a rather low execution time of the operator. Image registration does not fit into this category. Figure~\ref{fig:chapter_distr_scan_execution_time} presents the execution time for a serial registration process. Each single image registration is much more computationally expensive and the actual execution time is of a completely different order of magnitude than a simple integer of floating--point number operation. It is very likely that different parallelization strategies may be required for a prefix sum with an operator taking several seconds to compute a single result.

Each non--rigid deformation stores only three floating--point values and the cost of sending this amount of data between cluster nodes is dominated by the latency, not the bandwidth. As a matter of fact, each execution of the registration function requires corresponding image data which is stored on the disk. We assume that each worker has access to each image file through usual and serial I/O operations. \\
The cost of sharing data between workers is negligible when compared with the computation time. Because of this difference, our problem will not benefit from prefix sum algorithms developed for message--passing systems which concentrated on minimizing the number of communication cycles.

In contrast to operations with a deterministic execution time, here in both functions \textbf{A} and \textbf{B} the actual computation cost is not only unpredictable but highly variant. Due to the iterative nature of problem solved by the two functions, as explained in Chapter~\ref{chap:image_registration}, we can not foresee for a given input data how many iterations are necessary to reach a stopping criterion. Figure~\ref{fig:chapter_distr_scan_preprocessing_imbalance} presents execution times of the first local stage for different workers. The black line represents an average computation time across workers. This scenario is not feasible, but it allows us estimating the delay introduced by an unequal load balance. In the first iteration of the global stage, workers are required to wait for a result from its predecessor, and it is not unusual that this idle time can be as large as the total runtime of local stage.\\
\begin{figure}
	\centering
	\includegraphics[width=\textwidth]{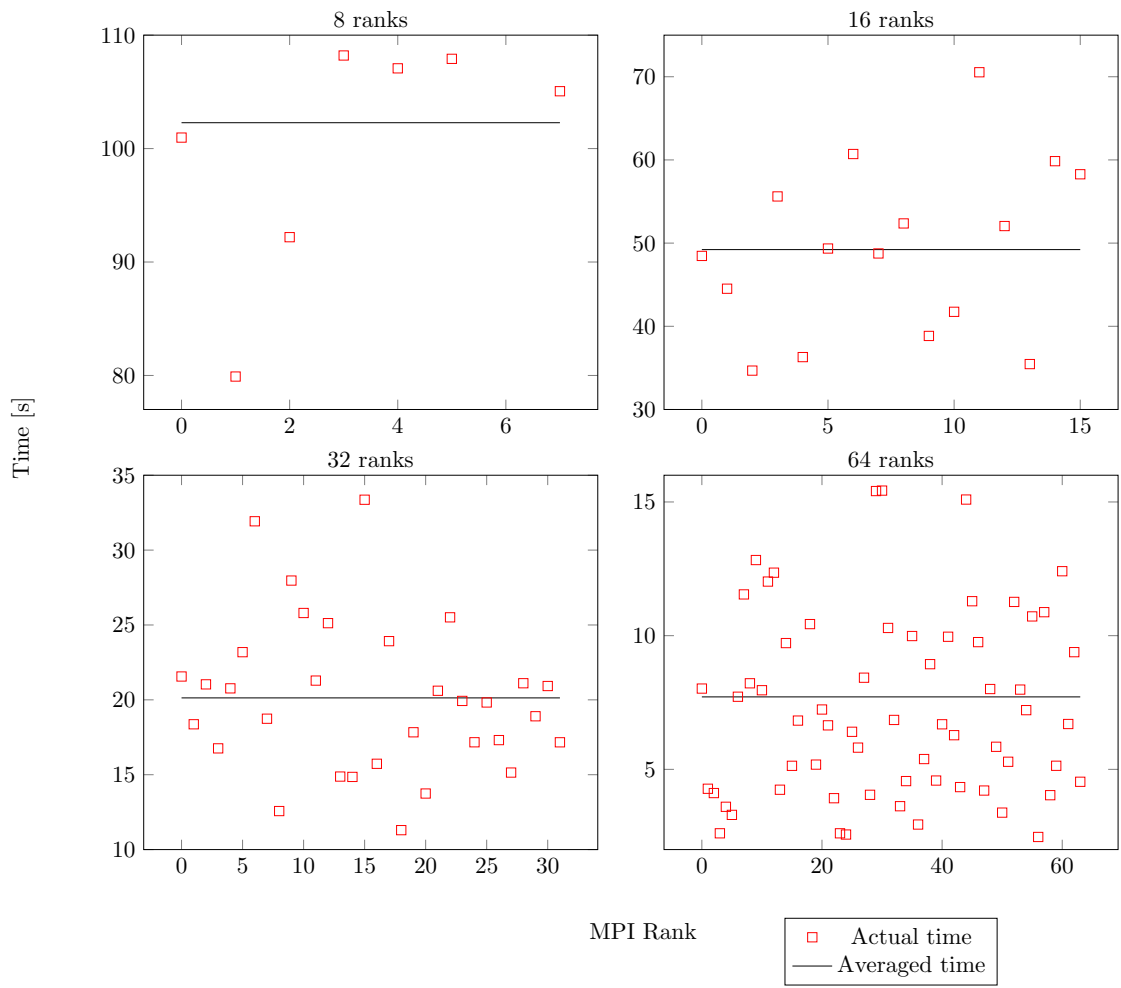}
	\caption{A comparison of the execution time of the first local stage for $N = 128$ image registrations within an MPI implementation. Four cases are presented, with $P$ values varying from 8 to 64. Red squares represent execution time for each worker and the black line represents a theoretical mean, an ideal distribution of work in an imaginary case where all operator applications take the same amount of time.}
	\label{fig:chapter_distr_scan_preprocessing_imbalance}
\end{figure}
\newpage
Regarding actual implementation, our operator depends on external objects supplying data structures such as multiple levels of a multigrid. Since constructing these objects in each application of the operator adds an undesired overhead, it is preferred to implement an operator able to capture external objects. 
We summarize the characteristics of our problem in several suggestions. We have used these pieces of advice to choose which parallel prefix algorithms, strategies, and optimizations may be interesting to our problem.
\begin{itemize}
	\item \textbf{Prefer replacing computation with communication} \\
	In contrast to usual design principles for prefix sum in a distributed environment, we do not optimize to reduce the latency of data transmission -- quite the opposite, computational intensiveness of our application suggest that we should prefer sharing partial results between workers if it allows reducing the amount of work to perform.
	\item \textbf{Do more work but in parallel} \\
	A lot of research have been done to design prefix sums with a logarithmic span and optimal, linear work complexity. For image registration, the most promising algorithms are the one who attains a minimum span with a non--linear work complexity.
	\item \textbf{Reduce dependencies when computation time is not predictable} \\
	There is a very little we could do with work distribution when no a priori knowledge on time imbalance is present. We can, however, look for parallel algorithms with the minimum number of flow dependencies. 
\end{itemize}

\section{MPI implementation}

In this section, we present an MPI implementation of the general strategy combined with various algorithms for a global parallel prefix sum. Work, span and speedup estimations are provided.

\subsection{MPI scan}

MPI defines two functions to perform parallel prefix sum, \code{MPI\_Scan} for an inclusive and \code{MPI\_Exscan} for an exclusive prefix sum. Besides that, non--blocking alternatives \code{MPI\_Iscan} and \code{MPI\_Iexscan} have been added in MPI 3.0. The C declaration of \code{MPI\_Scan} function is presented on Listing~\ref{lst:scan_interface}.
\begin{lstlisting}[
	caption={An interface of MPI function for an inclusive parallel prefix sum.},
	label={lst:scan_interface}
]
int MPI_Scan(const void *sendbuf, void *recvbuf, int count, MPI_Datatype datatype, MPI_Op op, MPI_Comm comm);
\end{lstlisting}
The operation requires all ranks to pass the same arguments for \textit{count}, \textit{datatype} and \textit{comm}. Input data is provided in \textit{sendbuf} and results are written in \textit{recvbuf}. A user--defined operator is assumed to be associative and can be declared as commutative to enable selection of more optimized algorithms. A C prototype of an operator is presented on Listing~\ref{lst:user_defined_op}. 
\begin{lstlisting}[
caption={An interface of MPI user--defined operator for reducion function.},
label={lst:user_defined_op}
]
void MPI_User_function(void* invec, void* inoutvec, int *len, MPI_Datatype *datatype);
\end{lstlisting}
\begin{figure}
	\centering
	\includegraphics[width=\textwidth]{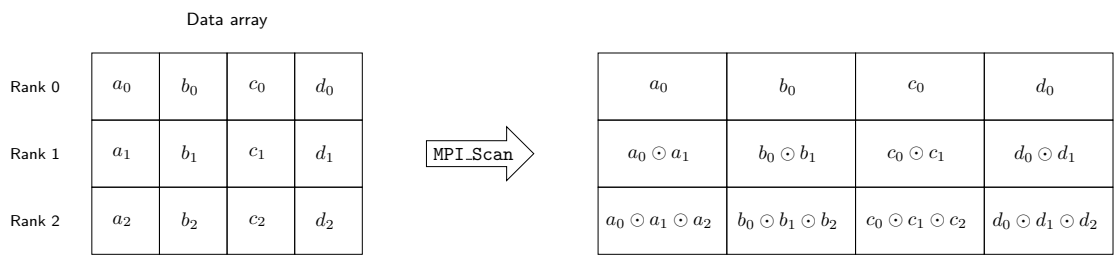}
	\caption{An example of executing \code{MPI\_Scan} on three ranks with four data elements per rank. After the execution, each array element stores a result corresponding to one of four different scans.}
	\label{fig:chapter_distr_scan_mpi_scan}
\end{figure}
The function is expected to apply $len$ times the operation $inoutvec[i] = invec[i] \odot inoutvec[i]$. The scan operator is applied elementwise to each element of the input buffer. Consequently, for $count$ data elements passed to an MPI scan on each rank, $count$ different prefix sums are computed. Given this limitation, existing MPI algorithms can not be used outside the general strategy for parallel prefix sum unless there is exactly one data item per rank. Otherwise, a specific type of segmented prefix sum is computed with only one data item per scan allowed on each rank. Figure~\ref{fig:chapter_distr_scan_mpi_scan} demonstrates this problem. MPI literature does not describe the problem in detail, nor it proposes any idea alternative to already described general scheme.\\
The MPI specification does not put performance requirements for an implementation. We provide a brief overview of algorithms present in two evaluated MPI implementations.

\subsubsection{OpenMPI}

OpenMPI documentation does no provide any information on actual implementation. An investigation of current source code\cite{OpenMPI} revealed an \code{MPI\_Scan} implementation in \textit{ompi/mca/coll/basic/coll\_basic\_scan.c} and an \code{MPI\_Exscan} in \textit{ompi/mca/coll/basic/coll\_basic\_exscan.c}. Both have been implemented with a fully serial algorithm where an MPI rank receives a value from its predecessor, computes a new value and sends it to its successor.\\
\newpage
The same algorithm can be found in an implementation of \code{MPI\_Iscan} in \textit{ompi/mca/coll/libnbc/nbc\_iscan.c} and \code{MPI\_Iexscan} in \textit{ompi/mca/coll/libnbc/nbc\_iexscan.c}.

\subsubsection{IntelMPI}

IntelMPI documentation shortly mentions types of algorithms implemented for the prefix sum\cite{IntelMPIRef}. The selection of algorithm can be manipulated by setting environment variable \code{I\_MPI\_ADJUST\_SCAN} or \code{I\_MPI\_ADJUST\_EXSCAN} with an integer value corresponding to algorithm selection. \\
For an inclusive scan, two algorithms are offered to the user: \textit{Partial results gathering} and \textit{Topology aware partial results gathering}. An exclusive scan can be computed with two algorithms, the first one is the same as in inclusive scan and the other one is called \textit{Partial results gathering regarding layout of processes}.

Unfortunately, the documentation does not describe algorithms in detail and the available information is limited only to these names.

\subsection{Serial}

With an exclusive sequential global prefix sum in the global stage, the span of a distributed scan becomes
\begin{equation}
\begin{split}
S_{DS}(N, P) &= S_{LS1}(N, P) + S_{GS}(P) + S_{LS2}(N,P) \\
&= \frac{N}{P} - 1 + P - 2 + \frac{N}{P} \\
&= 2 \cdot \frac{N}{P} + P - 3
\end{split}
\end{equation}
Then, the theoretical speedup is given as
\begin{equation}
\begin{split}
SP_{Serial}(N,P) &= \frac{N - 1}{2 \cdot \frac{N}{P} + P - 3}
\end{split}
\end{equation}
This term reaches its maximum value when denominator reaches its minimum. Removing constants simplifies the term to minimize to
\begin{align}
\frac{2N}{P} + P
\end{align}
and as the literature suggests\cite{McCool:2012:SPP:2385466}, it may be interpreted as doubled arithmetic mean of $\frac{2N}{P}$ and $P$. Arithmetic mean is bounded from below by the geometric mean
\begin{equation}
\begin{split}
\frac{\frac{2N}{P} + P}{2} &\geq \sqrt{\frac{2N}{P} \cdot P} \\
\end{split}
\end{equation}
Arithmetic and geometric means are equal if and only if both operands are equal which implies that for a fixed N the minimum is obtained for
\begin{equation}
\begin{split}
P &= \frac{2N}{P} \\
P &= \sqrt{2N}
\end{split}
\end{equation}
This result proves that for any $N$ the best speedup is obtained for a number of processors much smaller than $N$. Thus, the maximum speedup is
\begin{equation}
\begin{split}
\mathop{\boldsymbol\max}_{P=1}^{P=N} SP_{Serial}(N, P) &= \frac{N - 1}{\frac{2N}{\sqrt{2N}} + \sqrt{2N} - 3} \\
&= \frac{ \sqrt{2N} \cdot (N - 1) }{4N - 3\sqrt{2N}}
\end{split}
\end{equation}
As we can see, even for quite large input the attainable speedup is very low, no matter how many processor cores are available. The limit on scalability and very poor upper bound on speedup make this solution not attractive for a distributed implementation.

\subsection{Blelloch}

With a double sweep of a tree, the span is equal to
\begin{equation}
\begin{split}
S_{DS}(N, P) &= S_{LS1}(N, P) + S_{GS}(P) + S_{LS2}(N,P) \\
&= \frac{N}{P} - 1 + 2 \cdot \log_{2}{P} + \frac{N}{P} \\
&= \frac{2N}{P} - 1 + 2 \cdot \log_{2}{P}
\end{split}
\end{equation}
and the theoretical speedup of distributed scan becomes
\begin{equation}
\begin{split}
SP_{Blelloch}(N, P) &= \frac{N - 1}{\frac{2N}{P} - 1 + 2 \cdot \log_{2}{P}}
\end{split}
\end{equation}
The maximum of this function is found by minimizing the denominator
\begin{equation}
\begin{split}
f(P) &= \frac{2N}{P} - 1 + 2 \cdot \log_{2}{P} \\
\frac{df}{dP} &= - \frac{2N}{P^2} + \frac{2}{P\ln{2}} \\ &= \frac{2P - 2N\ln{2}}{P^2\ln{2}}
\end{split}
\end{equation}
The derivative reaches zero when
\begin{align}
P_{0} &= N\ln{2}
\end{align}
To find out whether it is a maximum or minimum, second derivative is tested
\begin{equation}
\begin{split}
\frac{d^2f}{dP^2} |_{P=N\ln2} &= \frac{4N}{P^3} - \frac{2}{P^2\ln2}|_{P=N\ln2} \\ &= \frac{4N\ln2 - 2P}{P^3\ln2}|_{P=N\ln2} \\ &= \frac{2N\ln2}{N^3\ln^42}
\end{split}
\end{equation}
Obviously, the second derivative is greater than zero for any positive value of $N$. Hence, $P_{0}$ is a local minimum of $f(P)$ and a local maximum of $S$.
The important result here is that critical point is always greater than N, implying that for practical values of $P$ it is always possible to improve the result by adding more processors until the number of processors reaches the number of data elements.

Work performed at all stages is equal to
\begin{equation}
\begin{split}
W_{DS}(N, P) &= 2\cdot N - P - \frac{N}{P} + W_{GS}(P) \\
&= 2 \cdot N - \frac{N}{P} - P + 2\cdot(P - 1) \\
&= 2 \cdot N - \frac{N}{P} + P - 2
\end{split}
\end{equation}
which proves that work complexity is linear in both input size and number of workers.

\subsection{Kogge--Stone}

With an inclusive scan, we apply the optimization and the span of distributed scan becomes
\begin{equation}
\begin{split}
S_{KS}(N, P) &= S_{LS1}(N, P) + S_{GS}(P) + S_{ILS2}(N,P) \\
&= \frac{N}{P} - 1 + P - 2 + \frac{N}{P} - 1\\
&= \frac{2N}{P} - 2 + \log_{2}{N}
\end{split}
\end{equation}
The analysis of speedup is similar to Blelloch algorithm, with the derivative
\begin{align}
\frac{df}{dP} &= \frac{2P - N\ln{2}}{P^2\ln{2}}
\end{align}
Here the local maximum is reached for
\begin{align}
P_{0} &= 2 \cdot N \ln{2}
\end{align}
Which is, again, an impractical value. Work performed at all stages is
\begin{equation}
\begin{split}
W_{DS}(N, P) &= 2\cdot N - P - \frac{N}{P} + W_{GS}(P) \\
&= 2 \cdot N - \frac{N}{P} - P + P \cdot \log_{2}{P} - P + 1 \\
&= 2 \cdot N - \frac{N}{P} + P \cdot (\log_{2}{P} - 2) + 1
\end{split}
\end{equation}
which proves that work complexity is linear in input size but $\mathcal{O}(n\log{}n)$ in number of workers.  

\subsection{Sklansky}

Sklansky parallel prefix adder has the same span as Kogge--Stone algorithm. Therefore, attainable speedup is exactly the same. The work complexity is slightly different but asymptotically equal to work performed by Kogge--Stone prefix sum
\begin{equation}
\begin{split}
W_{DS}(N, P) &= 2\cdot N - P - \frac{N}{P} + W_{GS}(P) \\
&= 2 \cdot N - \frac{N}{P} - P + \frac{P}{2} \cdot \log_{2}{P}  \\
&= 2 \cdot N - \frac{N}{P} + P \cdot (\frac{\log_{2}{P}}{2} - 1)
\end{split}
\end{equation}

\subsection{Summary}

We select for evaluation several instantiations of our general strategy
\begin{itemize}
	\item \textbf{MPI--based} \\
	Global scan implementation provided by the MPI library. Evaluate inclusive and exclusive variants against OpenMPI and different implementations of the scan in IntelMPI.
	\item \textbf{Blelloch}
	\item \textbf{Kogge--Stone, Sklansky}
\end{itemize}
\begin{figure}
	\centering
		\includegraphics[width=0.6\textwidth]{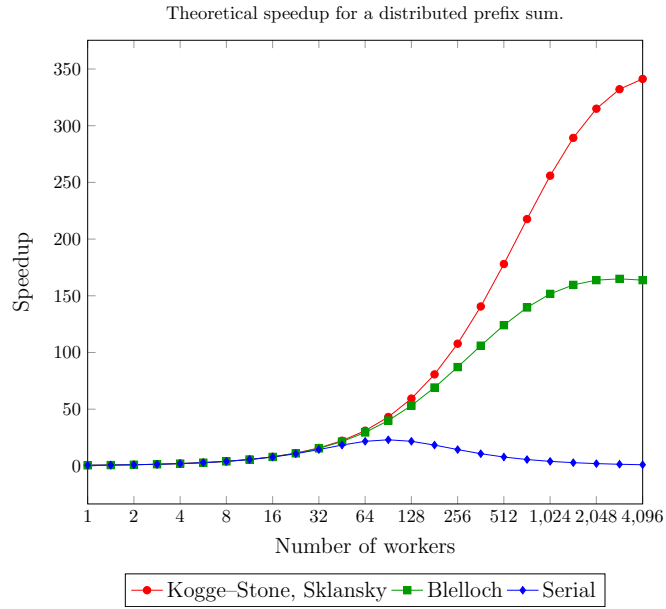}
	\caption{Theoretical speedup attainable by the general strategy with various implementation of the global stage. The horizontal axis is logarithmic with a base of two.}
	\label{fig:chapter_distr_scan_variants_comparison}
\end{figure}\begin{figure}
\centering
\includegraphics[width=0.6\textwidth]{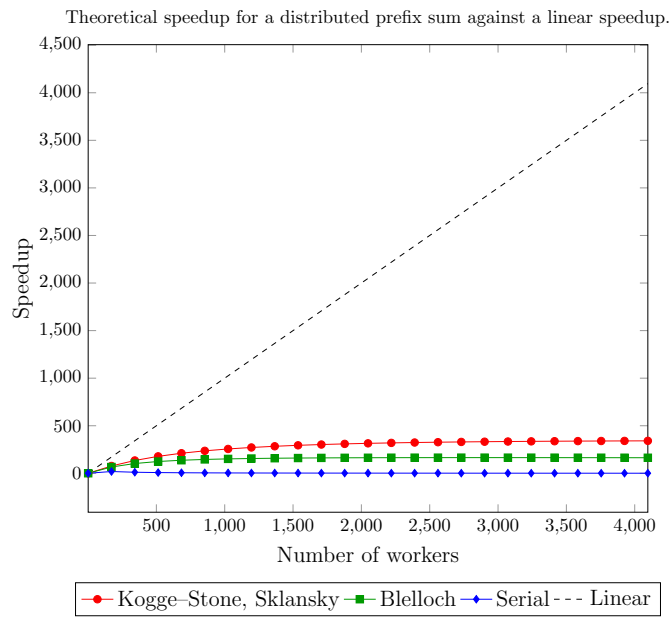}
\caption{A comparison of an ideal linear speedup against the theoretical speedup attainable by the general strategy with various implementation of the global stage.}
\label{fig:chapter_distr_scan_variants_comparison_linear}
\end{figure}
Figure~\ref{fig:chapter_distr_scan_variants_comparison} presents a theoretical prediction of attained speedup for selected variants. Clearly, the small constant factor in a span between Blelloch and Kogge--Stone prefix sums may lead to a large difference. Figure~\ref{fig:chapter_distr_scan_variants_comparison_linear} compares selected strategies against a linear scaling to demonstrate that even with an ideal implementation, the parallel prefix sum is a problem which can not be efficiently parallelized on a large number of processor cores.

We expect to find answers for key questions:
\begin{itemize}
	\item How these algorithms perform with a computationally intensive operator?
	\item How load imbalance influences different algorithms? Does the theoretical promise of scalability holds up? 
	\item Is there a gain in selecting an algorithm with a slightly better span but noticeably worse work complexity?
	\item What is the quality of prefix sum provided by MPI implementations?
	\item How does an inclusive and exclusive MPI algorithm perform in the general strategy?
\end{itemize}


\chapter{Results}
\label{chap:results}

This chapter presents and analyses evaluation of parallel prefix sum algorithms for image registration. Experiments have been performed on cluster nodes consisting of Intel Xeon E5-2680 v2 CPUs. Each node contains two ten--core processors with frequency of 2.80 \si{\GHz} and 3.60 \si{\GHz} in TurboBoost mode. Our compiler of choice is GCC in version 5.3.0. MPI libraries selected for evaluations are OpenMPI in version 1.10.4 and IntelMPI in version 2017.1. Each measurement has been performed with the help of \code{MPI\_Wtime} function, preceded by an \code{MPI\_Barrier} as it is advised in the MPI standard\cite{MPISpec}. Timings have been averaged over five repetitions and the standard deviation has been calculated to ensure the quality of results. \\
A comparison of different parallel prefix strategies in terms of strong and weak scalability is presented in sections~\ref{chap:results_strong} and~\ref{chap:results_weak}, respectively. The section~\ref{chap:results_mpi} provides a detailed view of MPI performance in the parallel prefix sum problem. In the section~\ref{chap:results_alternative}, a brief comparison of the general and alternative strategy is provided. Results of experimental shared--memory parallelization of the operator are presented in section~\ref{chap:results_hybrid}.

\section{Strong scaling}
\label{chap:results_strong}
In this section, we investigate the parallel performance of a distributed prefix sum when the number of processor cores is increased from 1 to 512 and the problem size is fixed. Thus, the size of work chunk for each worker is decreased. An ideal linear speedup is not possible to obtain in our problem because of an increased workload in the parallel execution. For the comparison, we select results from following parallel prefix sums
\begin{itemize}
	\item \textbf{Kogge--Stone, Sklansky} \\
	our implementation of work--inefficient prefix adders with a logarithmic span
	\item \textbf{Blelloch} \\
	our implementation of a classical work--efficient scan
	\item \textbf{OpenMPI} \\
	an inclusive MPI collective function \code{MPI\_Scan}
	\item \textbf{IntelMPI} \\
	an exclusive MPI collective \code{MPI\_Exscan} with the algorithm \textit{Partial results gathering regarding layout of processes}
\end{itemize}
For MPI libraries, we have selected the best performing implementation for this comparison. A detailed overview is presented in section~\ref{chap:results_mpi}. \\
The speedup is measured as a ratio of serial and parallel execution time. For all results, the standard deviation of measured speedup $\mathcal{SP}$ has not exceeded $0.1$. Thus, we can safely compare various algorithms even when the difference in speedup is rather small. A standard deviation of speedup $\mathcal{SP}$ is given as a 
\begin{align}
\sigma_{\mathcal{SP}} &= \mathcal{SP} \sqrt{ (\dfrac{\sigma_{t_{1}}}{t_{1}})^2 + (\dfrac{\sigma_{t_{2}}}{t_{2}})^2} 
\end{align}
where times $t_{1}$ and $t_{2}$ with their respective deviations $\sigma_{t_{1}}$ and $\sigma_{t_{2}}$ are measured runtime of a serial and parallel execution.  
\begin{figure}[htb]
	\centering
	\includegraphics[width=\textwidth]{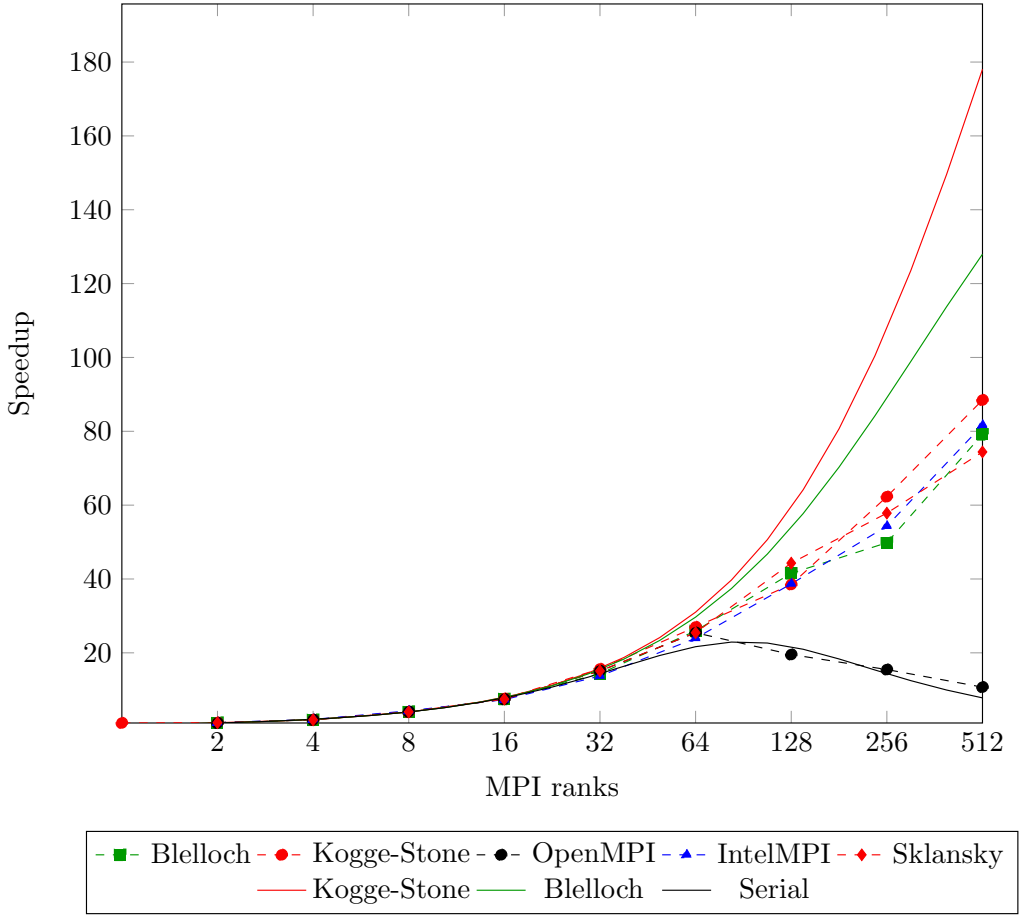}
	\caption{Speedups of different implementations of a distributed prefix sum, plotted as dashed lines, are compared with the theoretical upper bounds, presented with solid lines. The horizontal axis is logarithmic with a base of two.}
	\label{fig:chapter_results_strong_4096}
\end{figure}

Figure~\ref{fig:chapter_results_strong_4096} presents the image registration of 4096 frames. We compare our implementations against the theoretical upper bound, with Sklansky and Kogge--Stone implementations having the highest bound due to the lowest span of the	 global stage. We expect that an ideal implementation should be able to achieve such speedup for problems with a low variation in execution time for the binary operator of prefix sum.\\
We do not observe any significant difference in experiments with less than $32$ MPI ranks, where each process has a large work chunk and the cost of the global stage is relatively low. For larger runs, the best performance is attained by a Kogge--Stone implementation. Interestingly enough, the Sklansky implementation performs worse even though both algorithms have the same span. Moreover, the Sklansky prefix sum has a lower work complexity which should indicate a better performance. A possible explanation of this phenomena may be the non--constant \textit{fan--out} of Sklansky algorithm, as explained in section~\ref{chap:prefix_sum_sklansky}. In both algorithms, $\frac{N}{2}$ workers are active in the last iteration but their dependencies are different. In the Kogge--Stone prefix sum, each active worker depends on a result from a different worker. In the Sklansky algorithm, all active workers have to wait for a result from the same worker with index $\frac{N}{2} - 1$. Therefore, a significant delay on this worker has a much serious influence on the total performance. The work--efficient Blelloch scan performs slightly worse than a Kogge--Stone approach. A large work complexity seems to not be a problem in the image registration, as long as the span is minimal. \\
We do not have any expectations about the performance of an exclusive IntelMPI implementation because of a lack of details about the algorithm in the documentation. It behaves similarly to the Blelloch prefix sum, which might suggest that the tree--based scan has been implemented as \code{MPI\_Exscan}. At the same time, the source code analysis of OpenMPI suggested a serial implementation of the scan operation and measurements confirm this hypothesis. This implementation performs very poorly and it is simply unable to scale beyond a certain upper limit. \\
\begin{figure}[htb]
	\centering
	\includegraphics[width=\textwidth]{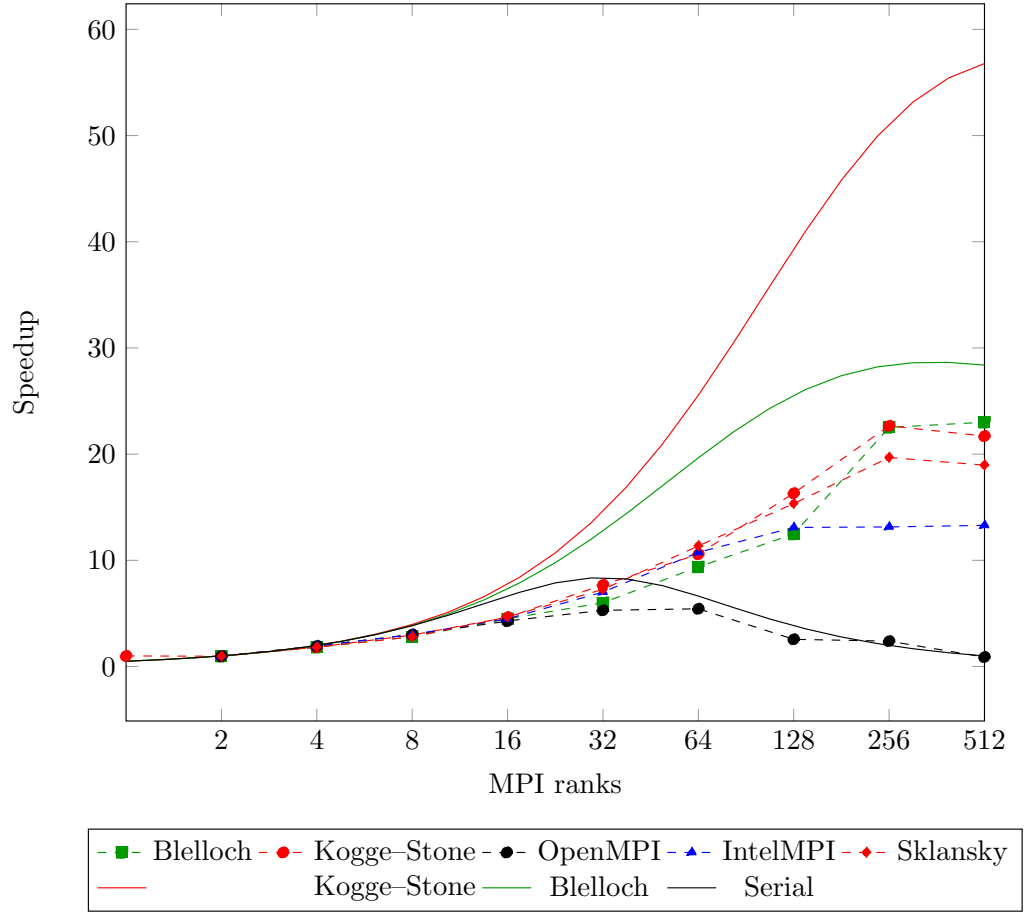}
	\caption{Speedups of different implementations of a distributed prefix sum, plotted as dashed lines, are compared with the theoretical upper bounds, presented with solid lines. The horizontal axis is logarithmic with a base of two.}
	\label{fig:chapter_results_strong_512}
\end{figure}
Results of a smaller experiment are presented in Figure~\ref{fig:chapter_results_strong_512}. Here, the execution time is even more dominated by the global stage, because of a much smaller workload in local stages. A corner case is $P = 512$ where each worker obtains only one data element, and no computation is done in local phases. Three implementations - Kogge--Stone, Sklansky and IntelMPI - attain the best speedup for $P = 256$. Surprisingly, the Blelloch algorithm improves even for $P = 512$. Although this result proves that the work--efficient Blelloch algorithm can perform better in some scenarios than a span--optimal Kogge--Stone algorithm, we note that this is a pathological case of $N = P$. Thus, Kogge--Stone is a preferable choice for real--world applications. \\

It is not surprising that the efficiency of all algorithms drops quickly on a larger number of MPI ranks. An increased number of workers means that more iterations are performed in the global stage and each one involves a blocking receive of results from another worker. This synchronization increases the likelihood of a rank idle waiting due to an unequal distribution of workload.
\newpage
\section{Weak scaling}
\label{chap:results_weak}
In the case of weak scaling, the performance is measured for a fixed amount of work performed by each worker. In contrary to strong scaling, where the algorithm is stressed to utilize all available parallelism, weak scaling answers the question: is the algorithm able to solve a bigger problem without a decrease in efficiency, when more hardware is available? For an ideal linear scaling, the algorithm should be able to match an available parallelism of more processors with an increased size of the problem and the execution should stay constant. \\
In a distributed prefix sum, the span is given in a general form
\begin{align}
S(N, P) &= \dfrac{2N}{P} + C_{1}\log_{2}{P} + C_{2}
\end{align}
for some constants $C_{1}, C_{2} \in \mathbb{N}$. For an equal increase in the problem size and the number of processing elements, the span of a prefix sum is given as
\begin{align}
S(2N, 2P) &= \dfrac{4N}{2P} + C_{1}\log_{2}{2P} + C_{2} \\ &= \dfrac{2N}{P} + \log_{2}{P} + C_{2} + C_{1}
\end{align}
Thus, we can not expect that the execution time stays constant. Nevertheless, as long as the amount of work per worker is relatively high, we should observe that the increase in span does not cause a significant growth of execution time. \\
\begin{figure}[htb]
	\centering
	\includegraphics[width=\textwidth]{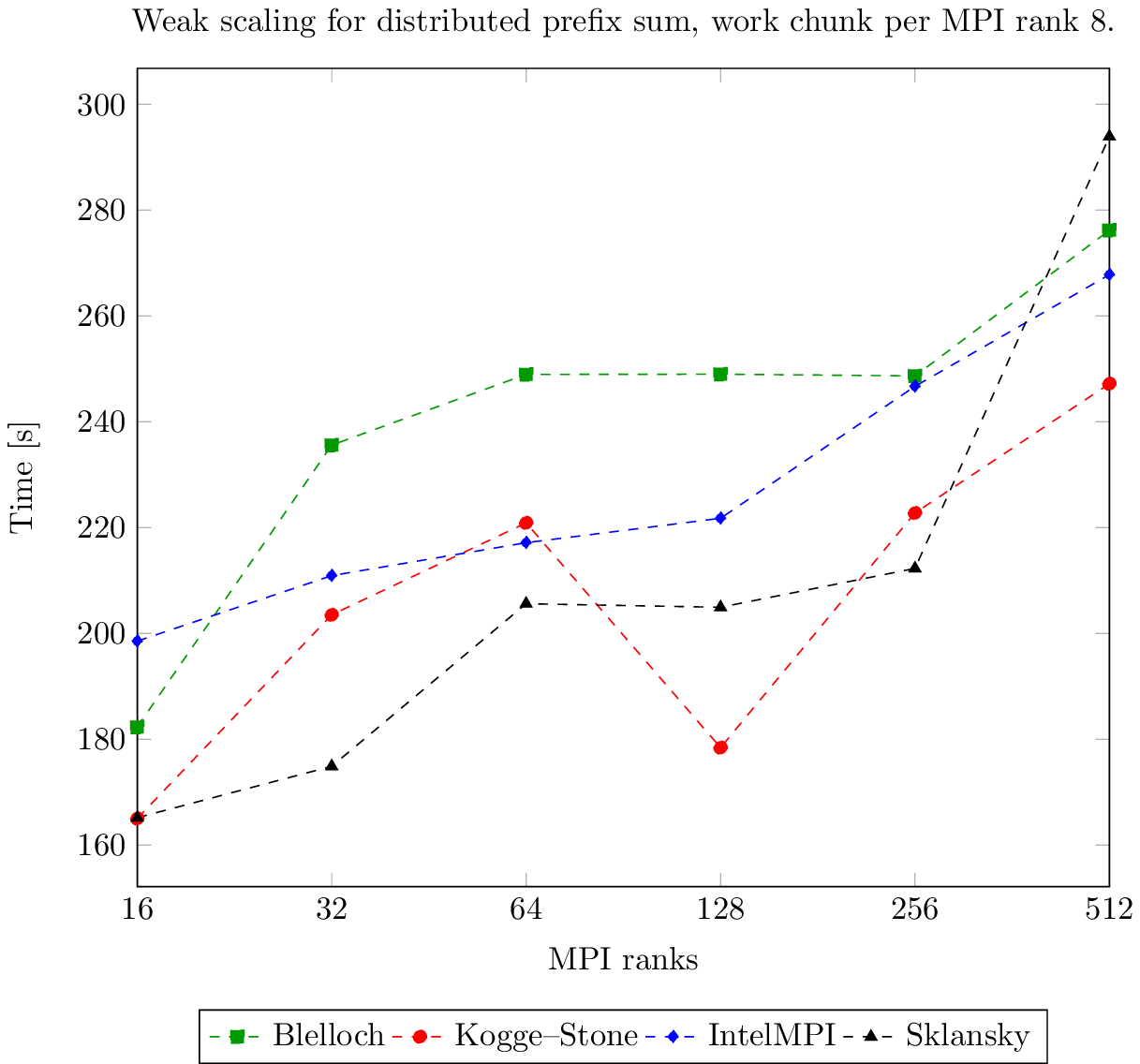}
	\caption{The smallest and largest experiments have been performed with 128 and 4096 images, respectively. The OpenMPI solution is excluded from this plot due to comparatively large values of execution time. The horizontal axis is logarithmic with a base of two.}
	\label{fig:chapter_results_weak_small}
\end{figure}
\begin{figure}[htb]
	\centering
	\includegraphics[width=\textwidth]{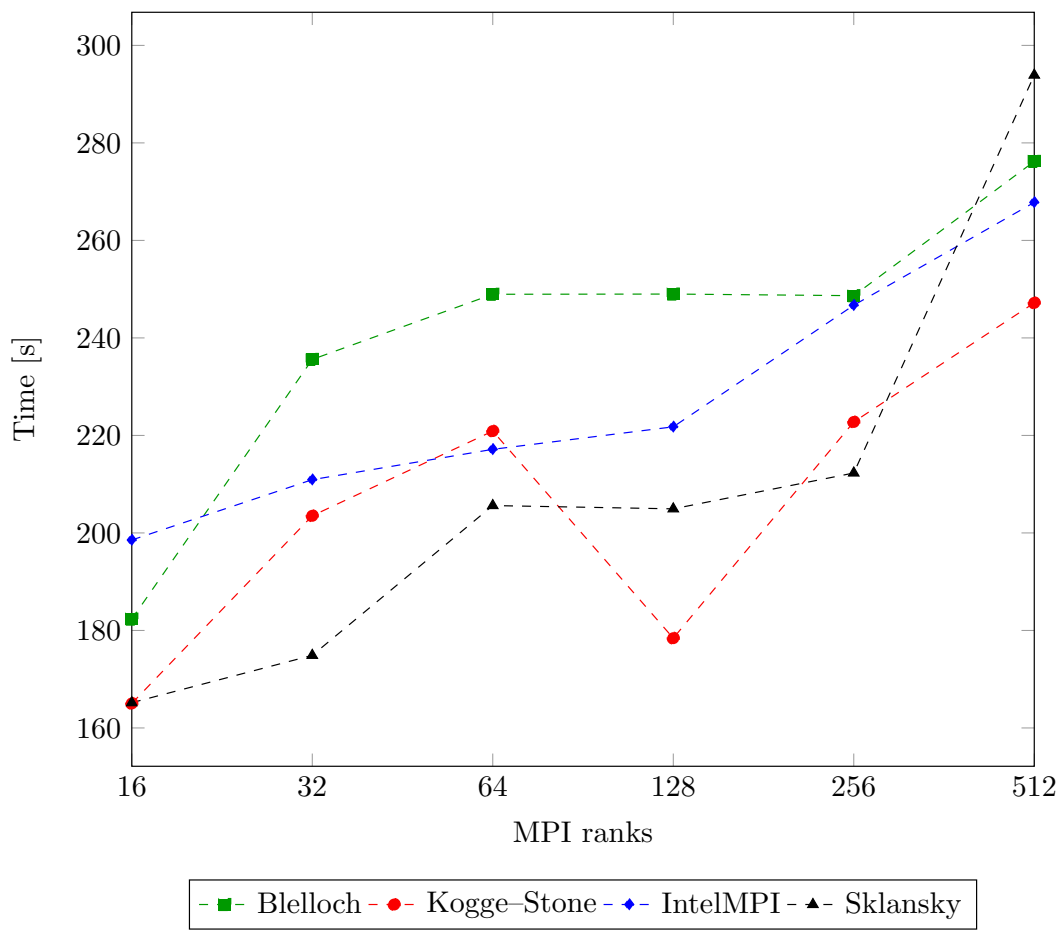}
	\caption{The smallest and largest experiments have been performed with 128 and 4096 images, respectively. The OpenMPI solution is excluded from this plot due to comparatively large values of execution time. The horizontal axis is logarithmic with a base of two.}
	\label{fig:chapter_results_weak_large}
\end{figure}
We begin with considering the case $\frac{N}{P} = 8$. The smallest and largest experiments have been performed with 128 and 4096 images, respectively. Results are presented on Figure~\ref{fig:chapter_results_weak_small}. The IntelMPI implementation has the lowest increase in time of $34\%$ but it has been the slowest solution at the beginning. Kogge--Stone and Blelloch perform similarly with $50\%$ and $51.6\%$ increase, respectively. Sklansky algorithm outperforms all competitors most of the time but the sudden decrease in the least measurement gives it the worst result - $77\%$ increase in execution time. On average, the increase in execution time is approximately equal to $42.8\%$. The OpenMPI solution is excluded from this analysis due to a large difference between its execution time and all other algorithms. It is sufficient to say that the runtime increases from 196 to 2041 seconds i.e. by $941\%$. \\ The analysis reveals more interesting information about the algorithms. Although Kogge--Stone provided the best performance in large runs, it is the only algorithm to not have a monotonic increase in execution time\footnote{The Sklansky algorithm exhibits a small decrease between $64$ and $128$ ranks but the decrease is lower than the standard deviation of measurements.}. The significant increase from $16$ to $64$ ranks, followed by a sharp decrease on $128$ ranks is a very unusual behavior. The variance of measurements is relatively low and we have found no reason to doubt the quality of conducted experiments. A possible explanation may be a huge sensitivity to load imbalance, causing a serious performance degradation for specific distributions of work.

Another results are presented on Figure~\ref{fig:chapter_results_weak_large}. An experiment with the larger amount of work per rank $\frac{N}{P} = 32$ is expected to perform better because a constant increase should have less effect when the span is dominated by $\frac{2N}{P}$. All algorithms show a high increase in execution time from $4$ to $8$ ranks and a slow and rather stable increase for other values. Although Blelloch is the slowest algorithm on $128$ images, its $14.3\%$ increase in execution time is the lowest one, followed by $27.4\%$ increase for Sklansky, $29.3\%$ for the IntelMPI implementation and $34.4\%$ for the Kogge--Stone algorithm. The average $28\%$ increase can not be compared with a $165\%$ increase for OpenMPI.

We observe that some image registration problem algorithms are not able to match theoretical bounds. For Kogge--Stone and Sklansky algorithms, the increase is much higher than theoretical limits of $27.7\%$ and $7.8\%$ increase, obtained by dividing the total increase with a span for the first measurement. This time, the Kogge--Stone algorithm exhibits a sudden increase in execution time at the last stage. On the other hand, Blelloch performs better than theoretical limits of $52.6\%$ and $14.9\%$.
\section{MPI implementation}
\label{chap:results_mpi}
In this section, we compare the performance of a distributed image registration with the global stage wrapped over a built--in MPI parallel scan. First, we present various implementations in the IntelMPI library. Then, we measure the performance of OpenMPI scan and compare it against the best available IntelMPI scan.
\subsection{IntelMPI}
\begin{figure}[htb]
	\centering
	\subfloat[Inclusive]{%
		\includegraphics[width=\dimexpr0.5\textwidth]{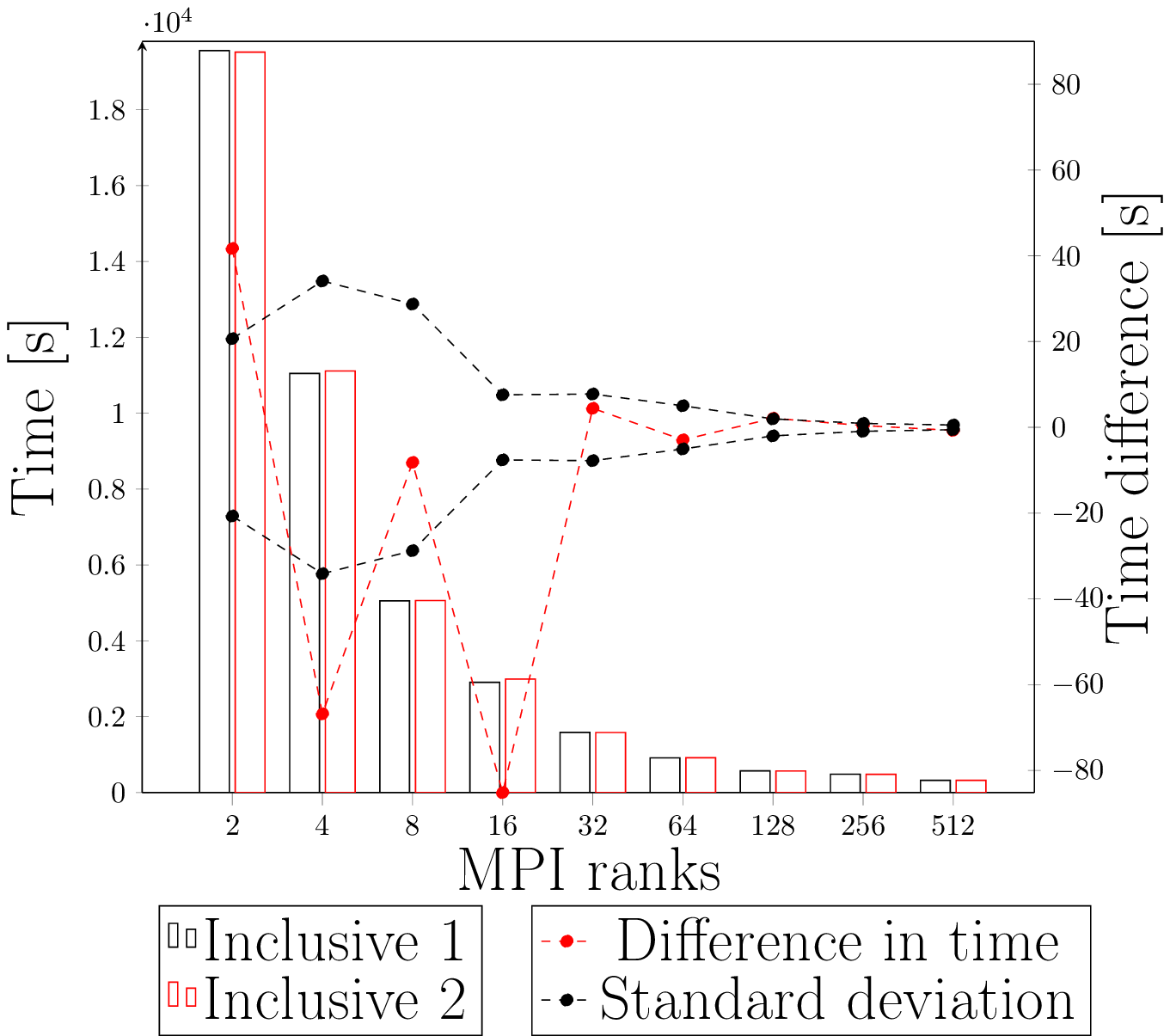}\label{fig:chapter_results_weak_large_intelmpi_incl}}\hfill
	\subfloat[Exclusive]{%
		\includegraphics[width=\dimexpr0.5\textwidth]{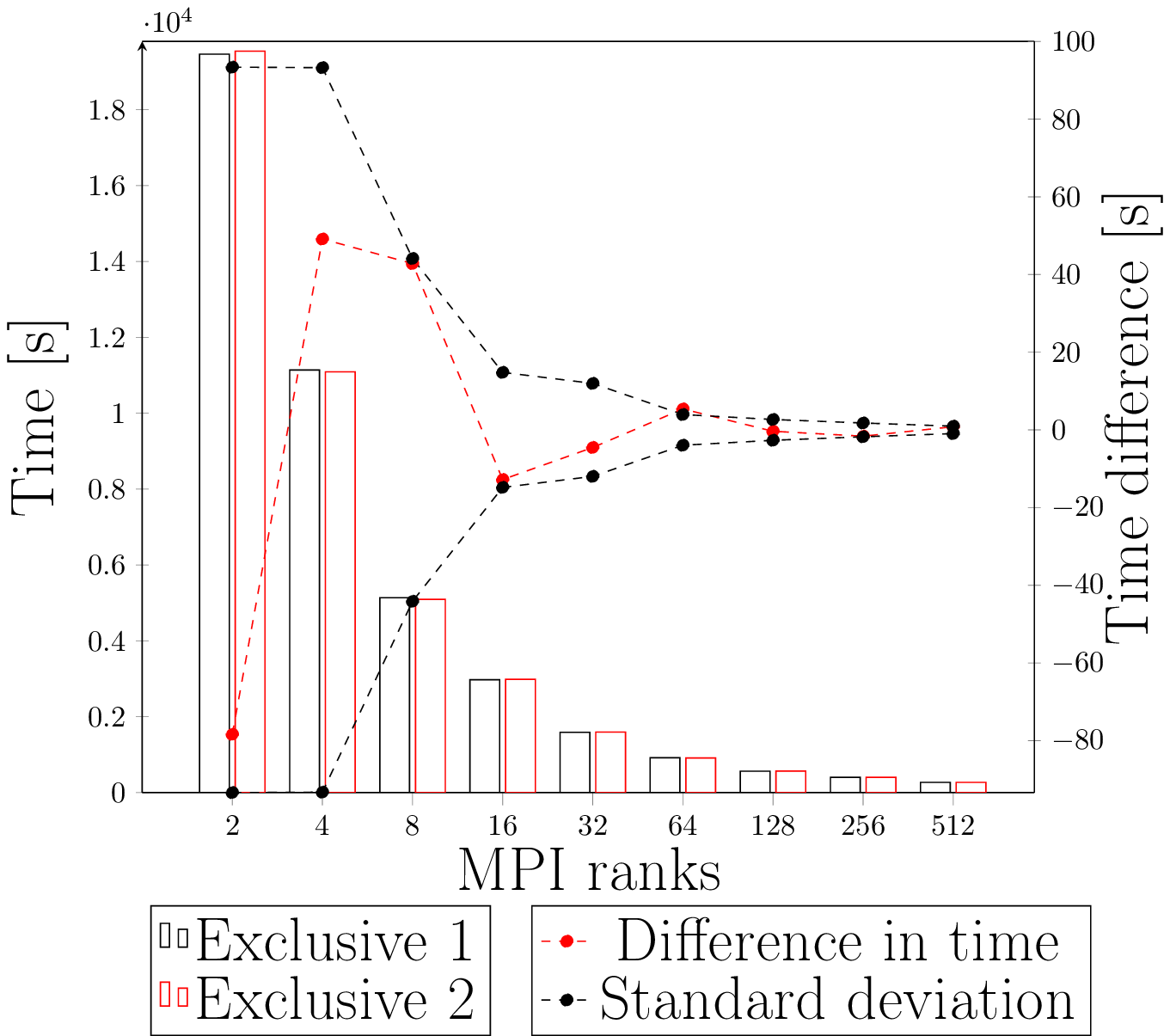}\label{fig:chapter_results_weak_large_intelmpi_excl}}\\
	\subfloat[Inclusive vs exclusive]{%
		\includegraphics[width=\dimexpr0.5\textwidth]{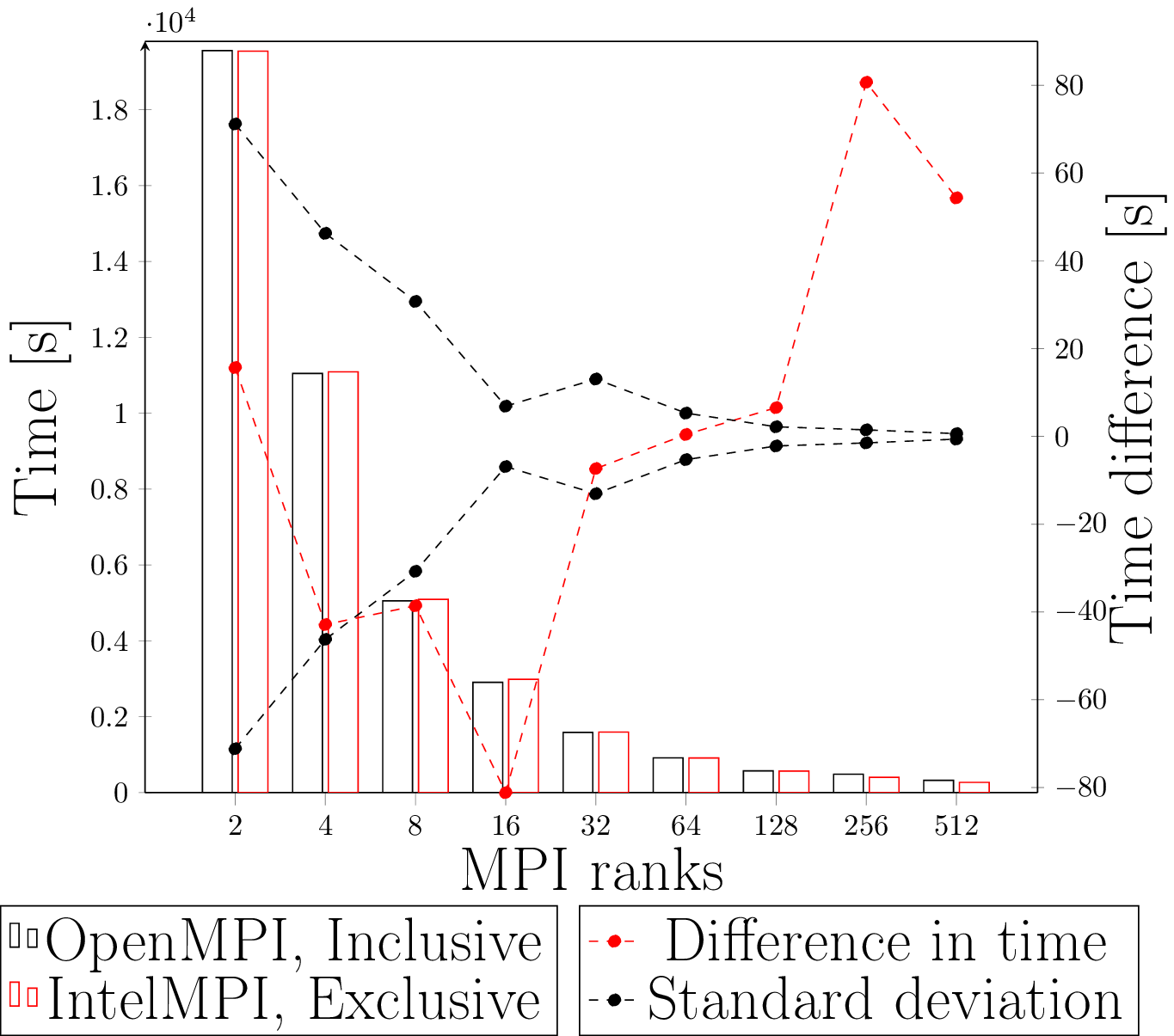}\label{fig:chapter_results_weak_large_intelmpi_comparison}}
	\caption{An IntelMPI--based implementation of a distributed prefix sum for $N = 4096$. The left vertical axis corresponds to the bar plot presenting the execution time for  algorithms. The right vertical axis corresponds to the plot of a difference in execution time between the first and second algorithm. The standard deviation is plotted with a reflection over the horizontal axis.}  \label{fig:chapter_results_weak_large_intelmpi}
\end{figure}
Inclusive and exclusive IntelMPI algorithms are compared on Figures~\ref{fig:chapter_results_weak_large_intelmpi_incl} and~\ref{fig:chapter_results_weak_large_intelmpi_excl}, respectively. Interestingly, the inclusive \textit{partial results gathering} performs better than its \textit{topology--aware} counterpart for runs which do not span across nodes. For computations requiring intra--node communication, the difference is lower than a combined standard deviation. On the other hand, results with an exclusive MPI algorithm are much less clear, since inter--node results are not consistent, and they are too close to each other to justify a verdict. However, we prefer the second algorithm because of a slightly better stability of measurements. \\
The figure~\ref{fig:chapter_results_weak_large_intelmpi_comparison} compares a default implementation of an inclusive scan and the \textit{topology--aware} implementation of the exclusive scan. The inclusive implementation provides a lower runtime when the number of MPI ranks is low. There, the difference in execution time is only an insignificant percentage of the total execution time. The exclusive scan surpasses the other algorithm on $128$, $256$ and $512$ processor cores where this improvement has a noteworthy influence on the performance.
\subsection{OpenMPI vs IntelMPI}
\begin{figure}[htb]
	\centering
	\subfloat[OpenMPI]{%
		\includegraphics[width=\dimexpr0.5\textwidth]{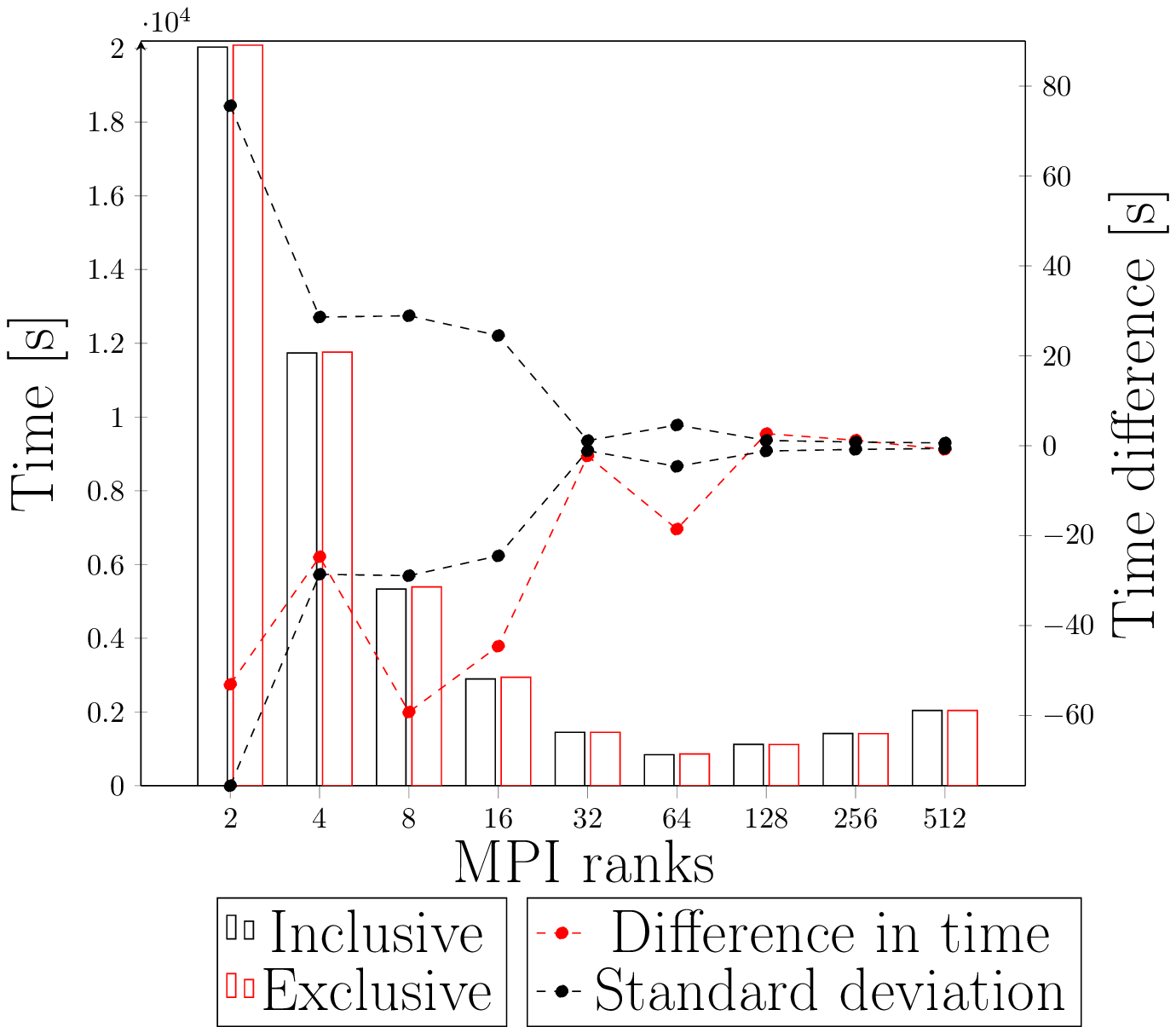}\label{fig:chapter_results_openmpi}}\hfill
	\subfloat[Comparison]{%
		\includegraphics[width=\dimexpr0.5\textwidth]{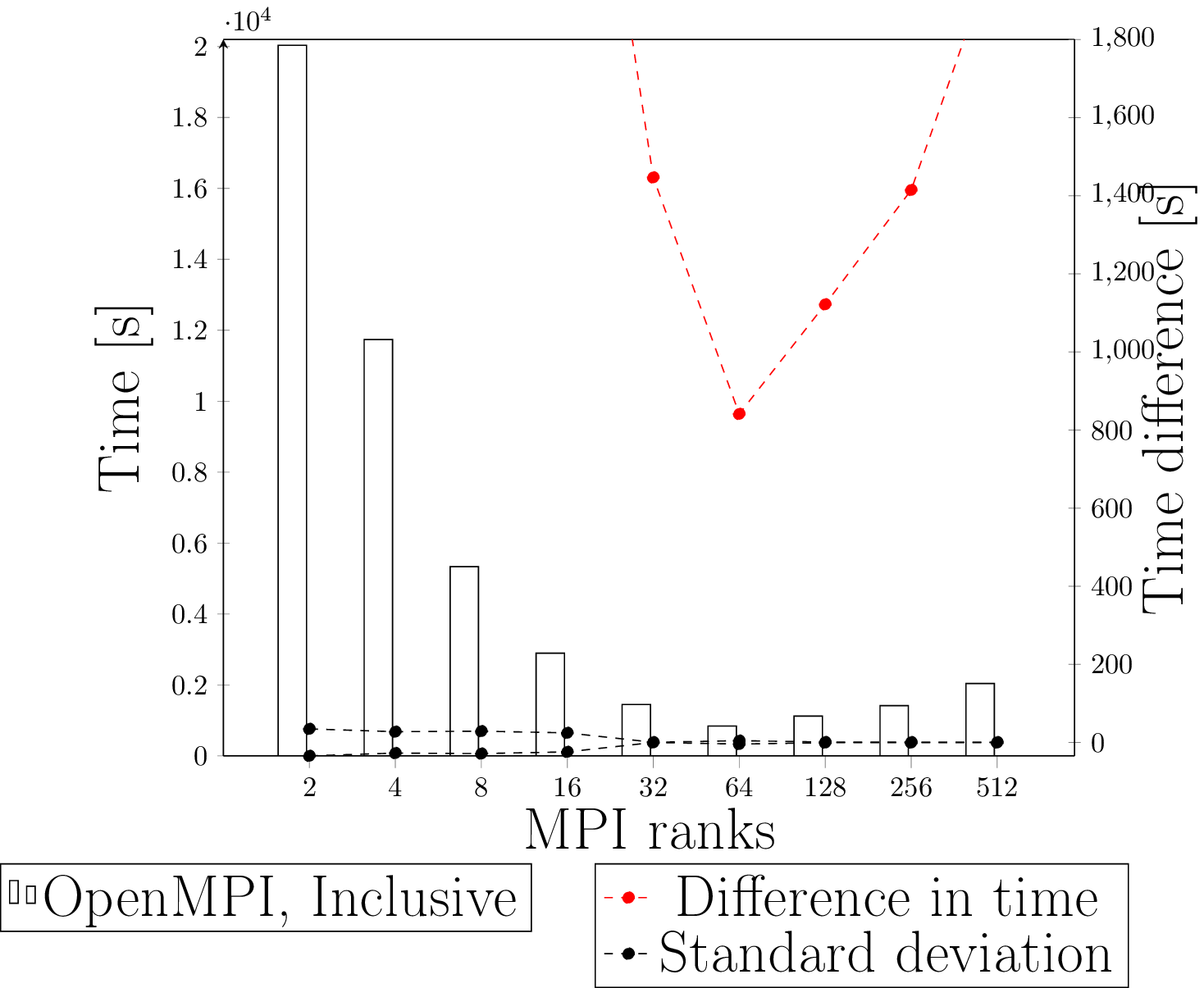}\label{fig:chapter_results_openmpi_intelmpi}}
	\caption{An MPI--based implementation of a distributed prefix sum for $N = 4096$. The left vertical axis corresponds to the bar plot presenting the execution time for  algorithms. The right vertical axis corresponds to the plot of a difference in execution time between the first and second algorithm. The standard deviation is plotted with a reflection over the horizontal axis.}  \label{fig:chapter_results_weak_large_intelmpi}
\end{figure}
Figure~\ref{fig:chapter_results_openmpi} compares the performance of two scan implementations in OpenMPI. A negative difference suggests that the inclusive prefix sum tends to perform better for all runs utilizing less than $128$ cores. Later, the difference is too small to differentiate between them. We have not found any explanation why the exclusive prefix sum performs worse in a serial implementation. We have decided to use the inclusive scan for comparisons with other prefix sums. \\
Finally, we compare best--performing implementations for IntelMPI and OpenMPI. Although the Figure~\ref{fig:chapter_results_openmpi_intelmpi} clearly shows that IntelMPI provides a superior performance for executions spanning multiple cluster nodes, results for $16, 32$ and $64$ ranks are surprising. Again, a likely explanation is that multiple iterations of a parallel scan intensify delays created by ab unequal computation time between workers. 

\section{Alternative strategy}
\label{chap:results_alternative}
Section~\ref{chap:distr_prefix_sum_alternative} introduced a formulation of the distributed prefix sum problem alternative to \textit{scan--then--apply} general strategy. Using the notation from the previous chapter, we can express the default strategy in terms of how result $x_{i}$, located on worker $I$, is computed 
\begin{align}
x_{0, i} &= \overbrace{\underbrace{x_{0, l_{I}}}_{\text{\makebox[0pt]{Global scan}}} \: \odot \: \underbrace{( (x_{l_{I}} \odot x_{l_{I} + 1} \odot \dots \odot x_{i-1}) \odot x_{i} )}_{\text{First local stage}}}^{\text{Second local stage}}
\end{align}
The final result is obtained by two applications of operator $\odot$ to $x_{i}$. On the other hand, the \textit{reduce--then--scan} strategy achieves the same goal with a single application of the binary operator. Result from the first local stage is passed only to the global prefix sum
\begin{align}
x_{0, i} &= \overbrace{(\underbrace{x_{0, l_{I}}}_\text{\makebox[0pt]{Global scan}} \: \odot \:x_{l_{I}} \odot x_{l_{I} + 1} \odot \dots \odot x_{i-1}) \odot x_{i} )}^{\text{Second local stage}}
\end{align}
Within the context of image registration, the approximately associative operator with an unpredictable runtime creates a possibility where a different ordering of execution may lead to improved initial guesses for consecutive calls to function $\mathbf{B}$. Thus, we have to compare these two strategies to find out which one is better suited for this task.
\begin{figure}[htb]
		\centering
	\includegraphics[width=0.8\textwidth]{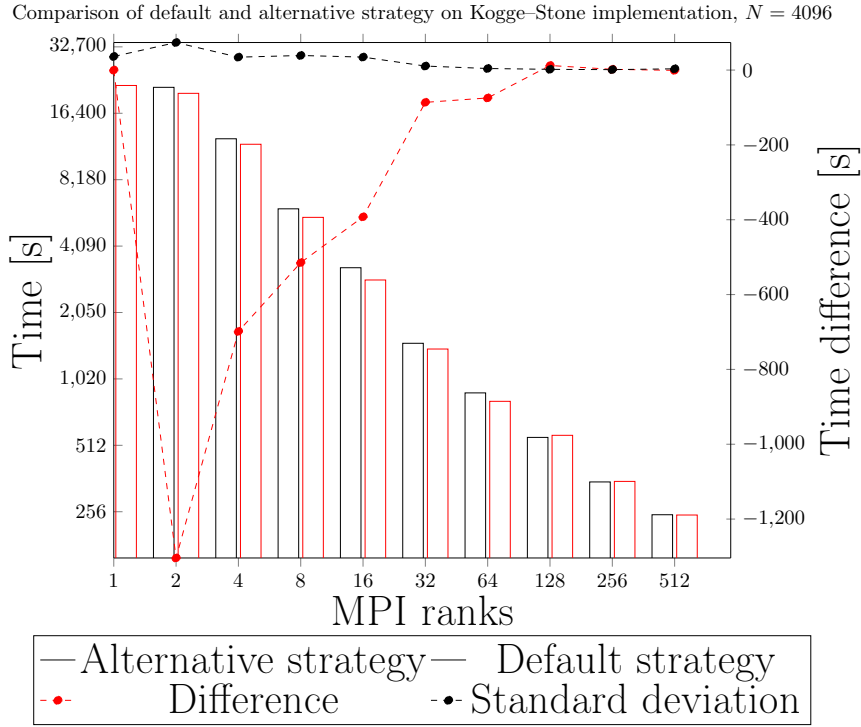}
	\caption{The left vertical axis corresponds to runtime for two strategies plotted as bars. The right vertical axis corresponds to a difference between the default and alternative strategy, presented with the standard deviation of difference. Both vertical and the left horizontal axes are logarithmic with base two.}
	\label{fig:chapter_results_altern}
\end{figure}
\begin{figure}[htb]
	\centering
	\includegraphics[width=0.7\textwidth]{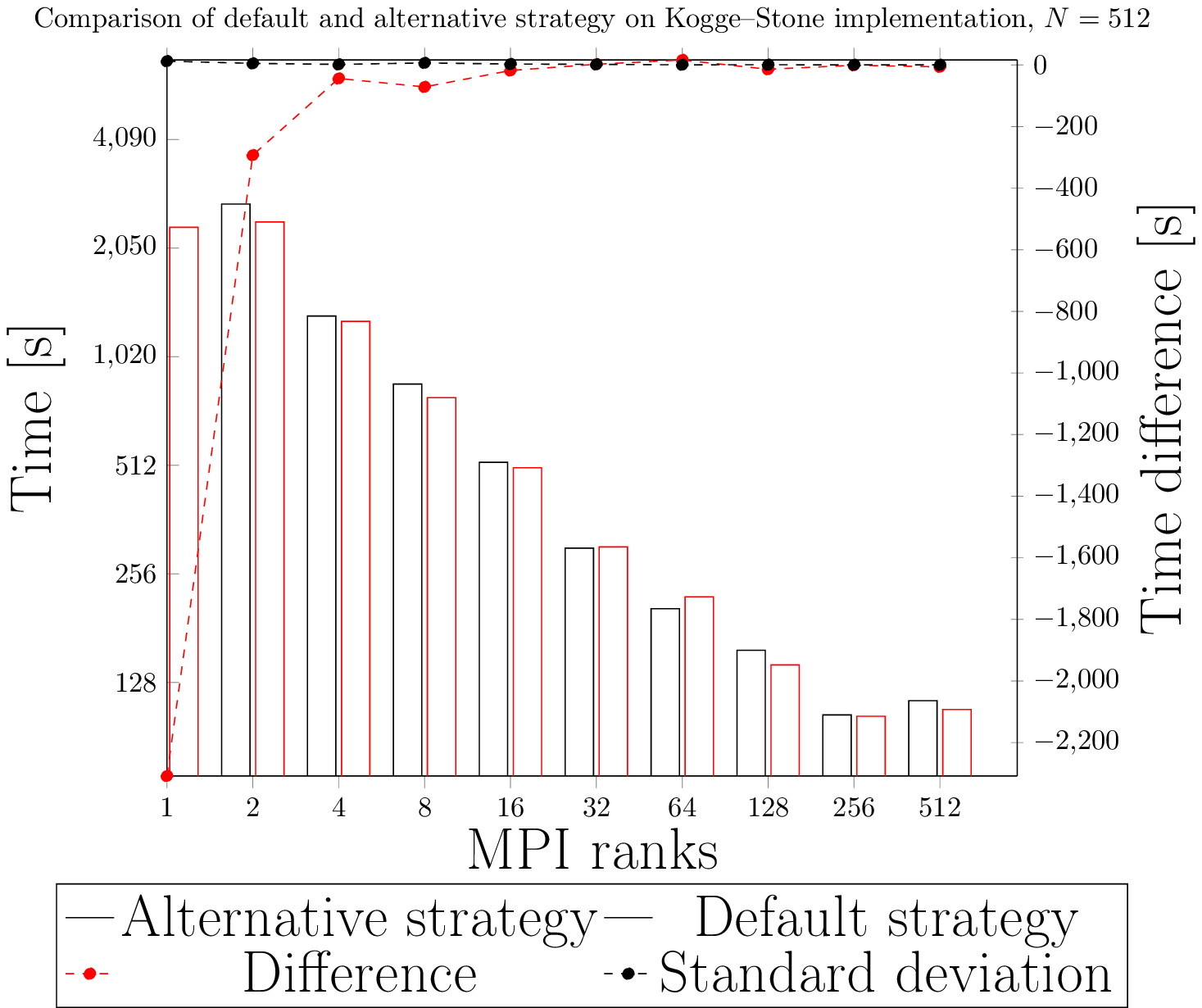}
	\caption{The left vertical axis shows runtime for two strategies plotted as bars. The right vertical axis displays a difference between the default and alternative strategy, presented with the standard deviation of difference. Both vertical and the horizontal ax are logarithmic with base two.}
	\label{fig:chapter_results_alternative_512}
\end{figure}
Figure~\ref{fig:chapter_results_altern} presents the execution time of two strategies for the distributed image registration of 4096 frames, with a Kogge--Stone implementation of global prefix sum. For all measurements, up to the 256 cores, the default strategy performs better but the difference is monotonically decreasing. The alternative strategy is slightly faster than the default strategy for $P = 128$. We have seen in the section~\ref{chap:results_weak} that the default strategy with Kogge--Stone global scan exhibits a strange drop in performance for $N = 4096$ and $P = 128$. In the execution on 512 cores, the difference becomes smaller than the combined standard deviation.\\
To get a better picture, we analyze as well a measurement with a smaller data chunk per worker. Results presented on Figure~\ref{fig:chapter_results_alternative_512} are not drastically different from the previous experiment. \\
We have not found any indication that investigating the alternative strategy may improve the algorithm's performance.


\section{Multithreading}
\label{chap:results_hybrid}
To overcome the poor scalability of prefix sum on a large number of cores, we investigated the possibility of a hybrid parallelization. Instead of allocating $P$ MPI ranks performing the distributed prefix sum, one could allocate $\frac{N}{2}$ or even $\frac{N}{4}$ ranks with two of four threads per rank. There, a better performance could be achieved by utilizing hardware to parallelize the image registration process.\\
A performance analysis of the image registration revealed two parallelizable functions which are responsible for approximately two-thirds of the execution time. According to Amdahl's law, it should allow for a speedup of $1.67$ times on two threads and $3$ times on four threads. \\
An experimental result for with GOMP, a GNU implementation of OpenMP for GCC compiler, is presented in Figure~
\ref{fig:chapter_results_gomp}. We observe that a shared--memory parallelization of the operator is beneficial on computations on a large number of cores and with small chunk of work per MPI rank, even if the speedup of operator parallelization is not linear.\\
An investigation of OpenMP runtime for Intel compiler did not reveal significantly different results.
\begin{figure}[htb]
	\centering
	\includegraphics[width=\textwidth]{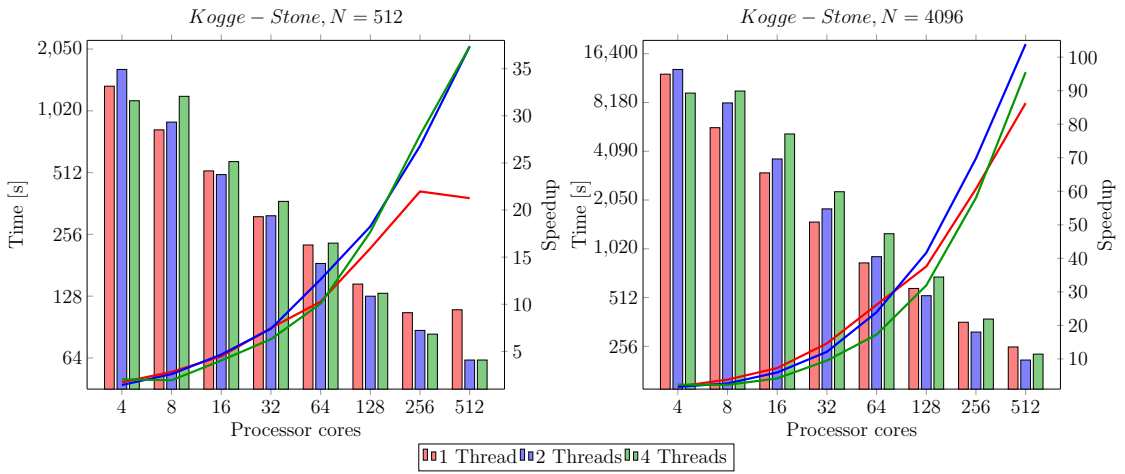}
	\caption{The left vertical axis corresponds to runtime of a hybrid parallelization, plotted as bars. The right vertical axis corresponds to speedup of a hybrid parallelization, plotted as solid lines. Both vertical and the left horizontal ax are logarithmic with base two. The last value for 4 threads is hardly visible because it is too small when compared with other data points.}
	\label{fig:chapter_results_gomp}
\end{figure}


\chapter{Summary}
\label{chap:summary}

In this dissertation, we have discussed a parallelization strategy for registration of series of electron microscopy images. We have approached the problem by representing the registration procedure as a prefix sum. Several parallel prefix sum algorithms have been implemented, evaluated and compared against each other and their respective theoretical predictions. In this chapter, we present the conclusions and suggest future improvements. 

\subsubsection{Scalability}

Strong scaling analysis of our problem reveals that an upper boundary on attainable speedup makes a massively parallel execution quite inefficient. Furthermore, a load imbalance induced by variances in execution time makes it even more wasteful in terms of computational resources. \\
Instead of trying to speedup fixed size problems, we focus on utilizing available resources to solve larger problems. A weak scaling analysis indicates there a more efficient use of available hardware.

\subsubsection{MPI facilities}

Our results show a stunning difference between MPI libraries in the quality of collective operations. A review of the literature suggests that this state of affairs may be caused by a relatively low popularity of the scan primitive in MPI community. Multiple papers have been written about optimal algorithms for collective operations such as barrier, broadcast, scatter, gather or even reduce. The IntelMPI library offers at least nine algorithms with variants for \code{MPI\_Allreduce}, \code{MPI\_Barrier}, \code{MPI\_Bcast} and \code{MPI\_Reduce}\cite{IntelMPIRef}. Sadly, the same cannot be said about the \code{MPI\_Scan}. \\
The image registration problem demonstrates that there is a need for high quality implementations of distributed prefix sum algorithms. 

\subsubsection{Our contribution}

We believe that the parallel image registration is the first example of the parallel prefix sum applied to a problem where the binary operator is
\begin{itemize}
	\item not associative
	\item computationally intensive
	\item of iterative nature with huge variances in convergence time
\end{itemize}
We have investigated algorithms existing in the literature and derived alternative formulations of the problem.
The research on parallel prefix adders has been concentrated on constructing work--efficient strategies with a minimal span. Parallel and distributed prefix sum algorithms have been designed for memory--bound operations and, as a result, they are optimized to minimize the cost of communication and memory access. Our work shows that there is a class of problems where better alternatives exist as neither work--efficiency nor optimized communication is desired.

\subsubsection{Limitations and future work}

The scalability of our parallelization scheme is inherently limited by the logarithmic nature of a parallel prefix sum and an unpredictable runtime for image registration operators. Our results suggest that a parallelization of the prefix sum operator is the only way of significantly improving the efficiency of a distributed prefix sum.

However, there are optimizations which can be applied to our problem. A major cause of a poor efficiency on many MPI ranks is a huge variation in execution times between workers. This impacts the global stage where each iteration involves an implicit synchronization through point--to--point communication. A shared--memory parallel implementation could decrease effects of a load imbalance by allocating one MPI rank per node and performing work--stealing inside a node. Negative effects of synchronization and communication in the global stage are decreased because they grow with the number of nodes, not the number of workers. Such improvement could boost the efficiency of computations spanning among a limited number of nodes, but it would not enable efficient, massively parallel computations.

Another way of improving performance is by redistributing the work after the global stage. In parallel prefix sum algorithms, many workers are not required to perform exactly $\log_{2}{P}$ iterations, and they are allowed to start computing final deformations earlier. Furthermore, in the general strategy the very first worker is not performing any work at all after the first local phase. \\
Sadly, an analysis of results suggests that workers with a largest theoretical span are not always the slowest ones. Besides, the second local phase tends to require much less computational effort than other stages. Even a successful reduction of the last phase on the slowest worker would only slightly decrease the total execution time. And there is evidence to suggest that a redistribution policy based on a theoretical prediction of span could cause more harm than good.

\addcontentsline{toc}{chapter}{Bibliography}
\printbibliography[nottype=online, title=References]
\printbibliography[type=online, title={Webpages}]

\end{document}